\documentclass[12pt,  oneside, a4paper]{article}
\usepackage[utf8]{inputenc}
\usepackage[a4paper, total={6.5in,9in}]{geometry}
\usepackage[dvipsnames]{xcolor}

\usepackage{hyperref}
\usepackage{amsmath}
\usepackage{mathtools}
\usepackage{lipsum}
\usepackage{blindtext}
\usepackage{titlesec}
\usepackage{graphicx}
\usepackage{amssymb}
\usepackage{amsthm}
\usepackage{rotating}
\usepackage{subcaption}
\usepackage{caption}
\usepackage{svg}
\usepackage{verbatim}
\usepackage{enumerate}
\usepackage{subfloat}
\usepackage{url}

\newtheorem{theorem}{Theorem}[section]
\newtheorem{lemma}[theorem]{Lemma}
\newtheorem{claim}[theorem]{Claim}

\newtheorem{aspt}{Assumption}[section]
\newtheorem{defn}{Definition}[section]
\newtheorem{remark}{Remark}[section]

\newcommand{\pr}{\mathbb{P}}
\newcommand{\numbuyers}{K}
\newcommand{\numusers}{N}
\newcommand{\budget}{B_k}
\newcommand{\buydec}{N_{k}}

\newcommand{\Q}{\mathcal{Q}}
\newcommand{\ind}{\mathbf{1}}

\newcommand{\stepone}{\text{(i)}\xspace}
\newcommand{\steptwo}{\text{(ii)}\xspace}
\newcommand{\stepthree}{\text{(iii)}\xspace}

\newcommand{\defneq}{:=}

\usepackage{cleveref}
\newcommand{\footremember}[2]{%
    \footnote{#2}
    \newcounter{#1}
    \setcounter{#1}{\value{footnote}}%
}

\title{Equilibria of Data Marketplaces with Privacy-Aware Sellers under Endogenous Privacy Costs}
\author{Diptangshu Sen\footremember{alley}{Georgia Institute of Technology. Email: dsen30@gatech.edu}
\and Jingyan Wang\footremember{trailer}{Georgia Institute of Technology. Email: jingyanw@gatech.edu} 
\and Juba Ziani\footremember{somethingelse}{Georgia Institute of Technology. Email: jziani3@gatech.edu}
}
\date{}

\begin{document}

\maketitle

\begin{abstract}
We study a two-sided online data ecosystem comprised of an online platform, users on the platform, and downstream learners or data buyers. The learners can buy user data on the platform (to run a statistic or machine learning task). Potential users decide whether to join by looking at the trade-off between i) their benefit from joining the platform and interacting with other users and ii) the privacy costs they incur from sharing their data.

First, we introduce a novel modeling element for two-sided data platforms: the privacy costs of the users are endogenous and depend on how much of their data is purchased by the downstream learners. Then, we characterize marketplace equilibria in certain simple settings. In particular, we provide a full characterization in two variants of our model that correspond to different utility functions for the users: i) when each user gets a constant benefit for participating in the platform and ii) when each user's benefit is linearly increasing in the number of other users that participate. In both variants, equilibria in our setting are significantly different from equilibria when privacy costs are exogenous and fixed, highlighting the importance of taking endogeneity in the privacy costs into account. Finally, we provide simulations and semi-synthetic experiments to extend our results to more general assumptions. We experiment with different distributions of users' privacy costs and different functional forms of the users' utilities for joining the platform.
\end{abstract}

% paper body
\newpage
\section{Introduction}
In today's digital era, tremendous amounts of data  are generated and processed daily. This proliferation of data is used towards a multitude of purposes, such as providing population-level statistics, refining sophisticated advertising strategies, enhancing recommender systems, better understanding user preferences, or training automated loan and hiring algorithms, to name a few. Some organizations hold vast amounts of user or customer data, while others continually seek to augment their datasets by acquiring and purchasing data from external sources. In light of this, the study of online data transactions is crucial in the context of the data-centric world we live in.

One critical issue that arises when user data is collected and transacted is data privacy. Users and consumers are becoming more and more aware of how their data can be shared or leaked, and privacy considerations are becoming an increasingly important factor for many when it comes to engaging with new applications or platforms. A study conducted by the Pew Research Center \cite{pew_privacy} found that roughly $80\%$ of adult Americans are ``very/somewhat concerned'' about how their data would be used once collected, and that ``the potential risks of collecting data about them outweigh the benefits''.

These concerns are exacerbated by numerous real-world instances where online and social media platforms have disclosed private user data to third parties, sometimes without explicit consent. Facebook (now Meta) is known to have shared user data with a large number of third-parties, with perhaps the most notorious and impactful example being the Cambridge Analytica scandal, where Facebook shared the private data of up to 87 million users with the political consulting firm Cambridge Analytica for the sake of political targeting~\cite{Meta}. A lesser known example is that of PayPal revealing that it shares personal data with hundreds of third parties~\cite{Paypal}. This data encompasses personal identifiable information (PII), payment details, and app metadata, among others, highlighting the extensive reach of data sharing practices. Tax preparation companies have also been sharing private taxpayer data with tech giants like Google and Meta for years~\cite{taxprep}. As platforms share more and more data with an increasingly larger number of third-party buyers, users and participants on online platforms are increasingly exposed to heightened risks of privacy violations. In this work, we aim to model how user privacy risks and costs are not just exogenous and fixed, but a function of \emph{how widely their data is shared and transacted}. 

We study this phenomenon in two-sided online platforms: the platform can collect data from participating users, and can share or resell this data to potentially numerous downstream data buyers. Our main modeling novelty and highlight is that we explicitly consider how users' privacy costs can depend on \emph{how many third parties acquire their data}, in particular treating user privacy costs as \emph{endogenous} and a function of the platform's data pricing and re-selling decisions. We are interested in understanding the impact of this endogeneity on user behavior and levels of user participation in two-sided online platforms. We argue that this endogeneity \emph{fundamentally changes} how users decide to interact with online platforms, with the platform's price acting as a knob that impacts the amount of data collected by third parties, user privacy costs, and user decisions to join a platform.

\textbf{Summary of contributions:} Our main contributions are organized as follows: 
\begin{itemize}
\item In Section~\ref{sec:model}, we propose a new model of two-sided data platforms with endogenous privacy costs that depend not only on the users' exogenous privacy preferences, but also on the downstream buyers' purchasing decisions.
\item In Section~\ref{sec:eq_conditions}, we show how the problem of finding an equilibria in our setting reduces to finding the solution of an equation with a single variable, the total user participation rate $\alpha$ on the platform at equilibrium. We also highlight how this equilibrium condition can be used to approximately compute equilibria efficiently. 
\item In Section~\ref{sec:eq_characterization}, we provide \emph{closed-form} characterizations of our equilibria, as a function of the price at which the platform sells data, in two settings: i) when users get a constant benefit from joining the platform and ii) when users get a linearly increasing benefit (in the participation rate) when they join. In both cases, we make the assumption that the users' privacy preferences are drawn from a uniform distribution. We make explicit comparisons to how participation equilibria would look like if privacy costs were exogenous, highlighting that endogeneity has a significant impact on equilibrium structure.
\item In Section~\ref{sec:experiments}, we experimentally expand on the equilibrium characterizations of Section~\ref{sec:eq_characterization}. We do so by relaxing our assumptions across several dimensions: i) we consider a more general distribution of users' privacy preferences inspired by real-life surveys on how humans value their privacy and ii) we consider more general, non-linear participation benefit functions for the users that encode diminishing return effects. 
\end{itemize}

\textbf{Related work:} Over the past decade, there has been a relatively large body of work concerned with studying and understanding properties of data markets and data collection on online platforms. Data markets are different from other markets in many ways, one of which is that data is freely replicable~\cite{sarkar2019ec}. Another specificity of data is privacy concerns: online users' personal data can be sensitive, and users can incur significant privacy losses when their data is shared online. On this front,~\cite{ghosh2011selling} is one of the first works that identifies privacy as a commodity and proposes mechanisms to elicit private data truthfully.

The line of work on data transactions spans many different dimensions, such as the study of mechanisms to design and sell data~\cite{gkatzelis2015pricing,bergemann2018design,bergemann2019information} and the study of mechanisms for data acquisition from sources that need to be compensated, e.g., \cite{roth2012conducting, dataacquisition, shuran_yiling}. This list is non-exhaustive, but a more comprehensive survey can be found at~\cite{bergemann_data_markets}. 

The works in this area are divided into two main lines of research. In the first one, the platform offers a service to the users that is often trained on their data; the benefit to a user is then endogenous, as it in turn depends on how much data users provide to the platform~\cite{liu2016learning, liu2017sequential, liu2018surrogate, dekel2010incentive, meir2011strategyproof, meir2012algorithms, perote2003impossibility, satml2023, fallah2023optimal}. The second one is when users are compensated directly through payments by the learner or marketplace collecting their data~\cite{cummings2015accuracy,cai2015optimum,abernethy2015low, shuran_yiling, roth2012conducting, dataacquisition}; in this case, their benefit from joining is exogenous. Our work aligns with the first category, where a user derives benefit from other users participating on the platform, rather than from monetary payments.

When it comes to data marketplaces, much of the existing literature focuses on a single side of the marketplace at a time; works that model the marketplace in its entirety, simultaneously considering data purchasing and data selling decisions are much rarer, with perhaps the most notable example being the recent work of~\cite{anish_data_marketplace}. While their focus is on the design of revenue-maximizing and incentive-compatible mechanisms, we focus on the understanding of user behavior and their participation levels at equilibrium. There is a central platform which acts as an intermediary (or market-maker). The platform collects data from users, who decide whether they want to participate by trading off their benefit from joining the platform. The platform then acts as an aggregator for the buyer side who want to purchase data to improve their learning tasks subject to budget constraints. The main way in which our work differs from most of the existing literature is that we consider a setting where the data owner's privacy costs are \emph{endogenous}, and scale with how widely buyers acquire their data; e.g., a user whose data is sold to few downstream buyers incurs much lower privacy costs than a user whose data is widely shared. To the best of our knowledge, \cite{gkatzelis2015pricing} is the only other work that explicitly models these effects in the context of data markets. However, we note that the problem we study and our techniques are different from theirs. Their setting focuses on algorithmic mechanism design for eliciting the privacy costs incurred by users and choosing how to compensate them monetarily for their data. We focus instead on understanding how users decide to organically participate, at equilibrium, on an online platform. We do not elicit user privacy valuations, and we do not design incentive-compatible monetary payments and allocations; users instead join because they derive utility from interacting with other users or from a service offered by the platform.

\section{Model}\label{sec:model}
We consider a marketplace comprised of two sides: a user side, that derives utility from joining the platform, and a data buyer side, that benefits from access to user data. The marketplace is centered around a central platform or intermediary that handles user-side data collection and aggregation, and buyer-side data transactions.

\noindent 

\textbf{Platform:} The platform collects data from users that decide to participate, and can re-sell this data to the data buyers. The platform sells data through a posted-price mechanism, controlled by a single price $P$ at which the data points are sold. 

\textbf{Users:} We assume that there is a non-atomic population of \emph{individually rational} users. The user population has total mass $\numusers$. Each user is interested in participating on the platform, but decides whether to do so by trading off their benefit from joining the platform for the privacy cost they incur from their data being sold downstream. In what follows, we separately discuss the models for the users, the buyers, and their interactions.

Let $\alpha \in [0,1]$ be the fraction of the users on the platform. If a user chooses to participate, they earn a \emph{benefit $Q(\alpha)$}, where this benefit function $Q$ is allowed to depend on $\alpha$. This has two potential motivations: i) a user may derive more utility from participating on a platform that has more similar users to her, or ii) the platform offers a data-dependent service to the users, and the quality of service increases with the amount of data the platform has access to (for example, this service could come in the form of better recommendations).

However, a participating user also incurs a \emph{privacy cost}\footnote{For simplicity, we assume a user cannot join a platform and then opt out of data collection.}, as a user's data may be shared by the platform with downstream data buyers. Each user is associated with a real number $v \in [0,1]$, where $v$ denotes the privacy valuation that captures how much the user cares about the privacy of their data. For simplicity, we call a user with privacy valuation $v$ also as ``user $v$''. The higher the $v$, the more a user cares about their privacy, where a value of $1$ corresponds to the most stringent privacy attitudes whereas $0$ corresponds to a user that is completely uninterested in privacy considerations. We define a probability distribution $f(.)$ to describe the distribution of privacy valuations: i.e., a user picked uniformly at random has privacy valuation $v$ with density $f(v)$.

Now, let $n_v$ be the random variable\footnote{Randomness may come for example from how data on the platform is allocated to buyers, as per Section~\ref{sec:eq_conditions}} that defines the number of buyers who purchase user $v$'s data. We define:
\[
     \ind_{k,v} = \begin{cases}
     1, \quad if \quad \text{buyer $k$ buys user $v$'s data} \\
     0, \quad o.w.
     \end{cases}
\]
We assume the privacy cost incurred by user $v$ is linear and  given by $c(v,n_v) = v n_v = v \left( \sum_{k=1}^{\numbuyers}\ind_{k,v} \right)$. To characterize the trade-off between $Q(.)$ and $c(.)$, we assume that each user has a \emph{quasi-linear utility} function: i.e., the net utility of the user is measured by the difference between the benefit received from service and the expected privacy cost incurred from participation. At participation rate $\alpha$, the utility that user $v$ gets in expectation is:  
\begin{align}\label{opt:user_utility}
      U_v = \begin{cases}
      Q(\alpha) - \mathbb{E}\left[v \left( \sum_{k=1}^{\numbuyers}\ind_{k,v} \right) \right], \quad \text{if user $v$ participates}\\
      0, \quad o/w. 
      \end{cases}
\end{align}

\begin{remark}[Range of $v$]
Note that $v$ is restricted to lie in $[0,1]$. This is without loss of generality, so long as privacy valuations are bounded, since the privacy valuations can be rescaled without affecting user behavior, as long as the benefit function $Q(.)$ is rescaled by the same factor. 
\end{remark}

\textbf{Buyer behavior:} 
There are $\numbuyers$ \emph{budget-constrained} buyers indexed $1$ through $\numbuyers$. Each buyer $k$ has a budget $\budget$ with which to purchase data from the platform. The buyer's decision problem is to choose what mass $\buydec$ of the total available data to buy. Each buyer aims to maximize the amount of data they collect subject to their budget constraint (for example, in order to maximize the accuracy of some downstream learning task). Therefore, when facing price $P$ for each unit mass of user data, buyer $k$ buys a mass of user data given by:
\begin{align}\label{eq:buyer_decision}
 N_k(\alpha) = \min\left(\alpha N, \frac{B_k}{P}\right).
\end{align} 

Without loss of generality, we make the following assumption on the buyers' budgets:

\begin{aspt}[Index-ordered budgets]\label{ass:budget_order}
We assume that the buyer budgets $B_1$ through $B_{\numbuyers}$ are in increasing order, i.e., 
\[
    0 < B_1 \leq B_2 \leq ..... \leq B_{\numbuyers} < \infty.
\]
We also define $B_0 \defneq 0$ and $B_{\numbuyers + 1} \defneq +\infty$ for technical convenience.
\end{aspt}

\textbf{Summary of Interactions:} We provide an ordered summary of the interactions in the marketplace:
\begin{enumerate}
    \item The platform sets a price $P$ for every unit mass of data that they sell.
    \item The buyers choose how much data to buy, based on the price $P$. This can be seen as buyers asking for $\frac{B_k}{P}$ amount of data, but may obtain less data if there is less data available on the platform than what the buyer asks for (i.e., $\alpha N \leq B_k/P$). Each buyer only pays for the data eventually allocated to him.
    \item Users (data owners) decide whether they want to participate on the platform. They do so based on knowledge of the amount of data bought by buyers at $P$, as well as knowledge of the user distribution of privacy valuations.\footnote{We note in fact that a user only needs to know the participation rate $\alpha$ at equilibrium and the number of buyers that buy their data, which are both statistics that can be computed and released by the platform itself.}
    \item The platform allocates data to the buyers based on the demand from step 2 and the available data from step 3.
\end{enumerate}

\section{Equilibrium Conditions and Computation}\label{sec:eq_conditions}
\subsection{Equilibrium Conditions}

We start by characterizing user best responses and use them to derive a simple necessary and sufficient equilibrium condition on the user participation rate $\alpha$. Formally, we define an equilibrium in our market as a situation in which \emph{rational} users do not want to deviate from their current participation decision:

\begin{defn}\label{defn:equilibrium}
Given platform price $P$, and given that buyers choose to buy $N_k(\alpha)$ as given in Equation~\eqref{eq:buyer_decision}, we say $\alpha$ is a ``user participation equilibrium'' if and only if when the participation rate is $\alpha$:
\begin{itemize}
    \item all participating users cannot increase their utility by leaving the platform;
    \item all non-participating users cannot increase their utility by joining the platform.
\end{itemize}
We assume that if an agent is indifferent between participating and not participating, we break ties in favor of participation.
\end{defn}

Before deriving our equilibrium conditions, we make the following assumption on how buyers select which users they buy data from.
\begin{aspt}\label{ass:homgenous_buying}
Each buyer buys data uniformly at random from the participating users, i.e., buyers have no preference over the data of specific individuals in the user group.  
\end{aspt}

\noindent
We can now provide our equilibrium conditions. These conditions are provided in the form of a necessary and sufficient condition that $\alpha$ must satisfy at equilibrium. We focus on understanding when non-trivial equilibria, i.e. those where at least a non-trivial fraction of the users decide to participate and $\alpha > 0$, arise:

\begin{claim}\label{clm:user_decision}
Let $P$ be the price set by the platform. 
$\alpha > 0$ is a participation equilibrium if and only if it satisfies the following relation: 
\begin{align}
        \alpha = F\left( \frac{\alpha \numusers Q(\alpha) }{\sum_{k=1}^{\numbuyers}N_k(\alpha)} \right), \label{eq:alpha_opt}
\end{align}
with $F(\cdot)$ the cumulative distribution of users' privacy parameter $v$ 
and $N_k(\alpha)$ given by Equation ~\eqref{eq:buyer_decision}.
\end{claim}

Given a fixed platform price $P$, we often re-write the equilibrium conditions as the following non-linear system of equations, for purposes of technical exposition: 
\begin{align}\label{eq:nonlinear_sys}
    &\alpha = F\left( \frac{\alpha \numusers Q(\alpha) }{\sum_{k=1}^{\numbuyers} N_k(\alpha)} \right); \\
    & N_k(\alpha) = \min \left(\alpha \numusers, \frac{B_k}{P} \right) \quad \forall~ k \in [\numbuyers]. \nonumber
\end{align}

\begin{proof}
The proof is divided into $3$ steps: in step $1$, we write user $v$'s expected privacy cost when participating in closed-form. In step $2$, we show how \textit{individual rationality} gives us a condition that the privacy valuation $v$ must satisfy for user $v$ to participate. Finally, in step $3$, we incorporate the distribution of user privacy valuations to obtain the equation the participation rate $\alpha$ must satisfy. 

\paragraph{Step 1:} We characterize the users' expected utilities from participation. Recall that user $v$'s privacy cost is given by $v \cdot \left( \sum_{k=1}^{\numbuyers}\ind_{k_v}\right)$ where $\ind_{k,v}$ is the indicator function that buyer $k$ buys user $v$'s data. Therefore,  the expected privacy cost incurred by user $v$ is:
\begin{align*}
   \mathbb{E}[\text{Cost}_{v}] 
   = \mathbb{E} \left[ v \cdot \left( \sum_{k=1}^{\numbuyers}\ind_{k,v}\right) \right] 
    = v \cdot \sum_{k=1}^{\numbuyers}\mathbb{E}[\ind_{k,v}]
\end{align*}
which can be rewritten as:
\begin{align}\label{eq:cost}
     \mathbb{E}[\text{Cost}_{v}]  = v \cdot \sum_{k=1}^{\numbuyers}\pr\left[ \text{buyer $k$ buys user $v$'s data} \right].
\end{align}
Now we use Assumption \ref{ass:homgenous_buying} to derive an expression for the probability term. Since buyers buy the participating data points uniformly at random, and a total $N_k(\alpha)/\alpha N$ fraction of the data is bought by buyer $k$, we have:
\begin{align*}
    \pr \left[ \text{buyer $k$ buys user $v$'s data} \right] 
    = \frac{N_k(\alpha)}{\alpha \numusers}
\end{align*}
for all participating users independent of their valuation $v$.  Plugging the probability back in Equation~\ref{eq:cost} for the expected cost, we obtain 
\begin{align}
    \mathbb{E}[\text{Cost}_{v}] = v \cdot \left( \sum_{k=1}^{\numusers}\frac{N_k(\alpha)}{\alpha \numusers} \right). \label{exp:privacy_cost}
\end{align}

\paragraph{Step 2:} From Equation~\ref{exp:privacy_cost}, we can now provide participation conditions for each user that we will use to characterize $\alpha$. Since user $v$'s net utility is given by: 
\begin{align*}\label{opt:user_utility}
      U_v = \begin{cases}
      Q(\alpha) - \mathbb{E}\left[v \left( \sum_{k=1}^{\numbuyers}\ind_{k,v} \right) \right], \quad \text{if user $v$ participates}\\
      0, \quad o/w, 
      \end{cases}
\end{align*}
and since user $v$ is individually rational, they choose to participate if and only if the net utility from participation exceeds the net utility from abstention, i.e., $Q(\alpha) - \mathbb{E}[\text{Cost}_v] \geq 0$. Now, plugging in the expression of the expected cost derived in Equation~\eqref{exp:privacy_cost}, we have that user $v$ participates if and only if:
\begin{align*}
    Q(\alpha) - v \cdot \left( \sum_{k=1}^{\numusers}\frac{N_k(\alpha)}{\alpha \numusers} \right) \geq 0,
\end{align*}
which is equivalent to 
\begin{align*}
v \leq \frac{\alpha \numusers Q(\alpha)} {\sum_{k=1}^{\numbuyers}N_k(\alpha)}.
\end{align*}

\paragraph{Step 3:} We now can derive the condition that the participation rate $\alpha$ satisfies at equilibrium. Since a user participates if and only if 
\[
v \leq \frac{\alpha \numusers Q(\alpha)} {\sum_{k=1}^{\numbuyers}N_k(\alpha)},
\]
the fraction of the total user mass that participates on the platform is given exactly by:
\[
   F\left( \frac{\alpha \numusers Q(\alpha)} {\sum_{k=1}^{\numbuyers}N_k(\alpha)} \right),
\]
remembering that $F$ is the c.d.f. of the distribution of users' privacy valuations. Therefore, $\alpha$ constitutes an induced user participation rate at equilibrium if and only if
\begin{align*}  
    \alpha &= F\left( \frac{\alpha \numusers Q(\alpha)} {\sum_{k=1}^{\numbuyers}N_k(\alpha)} \right).
\end{align*}
This concludes the proof of the claim.
\end{proof}
 
\textbf{Consequences for Equilibrium Computation:} We note that a simple \textit{grid search} approach approximately computes \emph{all} market equilibria for any fixed platform price $P$. Our problem reduces to solving Equation~\eqref{eq:alpha_opt} which only depends on a single real variable, $\alpha \in (0,1]$. We can then discretize over $(0,1]$ and return the values of $\alpha$ that satisfy Equation~\eqref{eq:alpha_opt} up to a pre-specified tolerance level, for any given market price $P$. The experimental results of Section~\ref{sec:experiments} are all computed by our grid search approach.

\section{Closed-Form Characterization of Market Equilibria in Simple Settings}\label{sec:eq_characterization}

We provide closed-form equilibrium characterizations in two variants of our model: one in which each user gets a constant benefit from participation, and the other in which each user's benefit increases linearly with the number of other users that participate. We highlight that these equilibrium characterizations can significantly differ depending on how users' benefits are structured, showing how participation equilibria are sensitive to the structure and the quality of service and interactions the users face on the platform. We provide experimental results for additional classes of benefit functions in Section~\ref{sec:experiments}.

\subsection{Preliminaries}
Before stating our results, we will set up some basic assumptions and definitions in an attempt to simplify the exposition. 

\begin{aspt}[Uniform Valuations]\label{ass:uniform_dist}
We assume that the distribution over the users' privacy parameter $v$ is uniform in the interval $[0, 1]$.
\end{aspt}
\noindent

In this section, we characterize the equilibria of our problem for any given value of $P > 0$, i.e., we compute the value(s) of $\alpha$ that arise at equilibrium. We focus on non-trivial equilibria where $\alpha > 0$. Throughout this section, we operate under the uniform distribution assumption. In this case, given $P$, the conditions for $\alpha$ to be an equilibrium become:
\begin{equation}\label{opt:platform}
    \begin{split}
    &\alpha = \min\left(1, \frac{\alpha \numusers Q(\alpha)}{\left( \sum_{k=1}^{\numbuyers}N_k(\alpha) \right)} \right); \\
    &N_k(\alpha) = \min \left(\alpha \numusers, \frac{\budget}{P} \right) \quad \forall~k \in [\numbuyers].
    \end{split}
\end{equation}

We define two types of non-zero equilibria depending on the value of the user participation rate $\alpha$ at said equilibrium. If $\alpha = 1$, we call the equilibrium ``an equilibrium with full participation'', while if $0<\alpha < 1$, we call it ``an equilibrium with partial participation''.

\begin{defn}[Cumulative Budget Notation]\label{defn:Bleq}
    We define $B_{\leq m}$ as the sum of the budgets for buyers $1$ through $m$, i.e., 
    \[
     B_{\leq m} 
     = \begin{cases}
     \sum_{k=1}^{m} \budget & if~ m \in [\numbuyers], \\
     0 & if ~ m = 0. 
     \end{cases}
    \]
\end{defn}

\subsection{Market Equilibria in the Constant Benefit Case}\label{sec:constant}

In this section, we characterize the equilibria of our market in the case where participating on the platform yields a constant utility $Q(\alpha) = \Q$ for each user. That is, the utility obtained by a participating user is independent of all other users' participation decisions.

Our main theorems provide closed-form characterizations of the equilibria in this setting. We divide our results into three cases, depending on how the benefit $\Q$ to users scales compared to the other problem parameters. Theorem~\ref{thm:eq_exist_highQ} considers the case when $\Q$ is large enough. Theorem~\ref{thm:eq_exist_modQ} considers the moderate case in which the range of possible values of $\Q$ is both upper and lower bounded. Finally, Theorem~\ref{thm:eq_exist_lowQ} considers the case in which $\Q$ is small. Our proofs are given in Appendix~\ref{app:constant}. We use the notation $\gamma(k) \triangleq \frac{(\numbuyers - k)B_{k} + B_{\leq k}}{\Q \numusers}$ in our theorem statements.

\begin{theorem}[High Benefit Case]\label{thm:eq_exist_highQ}
Suppose, $\Q \geq \numbuyers$. Then for all $P > 0$, there exists a unique equilibrium given by $\alpha = 1$.
\end{theorem}

In the high benefit case, the only equilibrium that arises is the one with full participation. Intuitively, the benefit for participation is so high that users have strong incentives to join the platform, even if their privacy valuations are high and even if buyers buy all user data. 

\begin{theorem}[Moderate Benefit Case]\label{thm:eq_exist_modQ}
Suppose $\numbuyers > \Q  \geq \frac{B_{\leq \numbuyers}}{B_{\numbuyers}}$. There exists a price threshold $\bar{P} < \gamma(K)$ such that:
\begin{itemize}
\item \textbf{Full Participation} For all $P \geq \bar{P}$, there is a unique (non-trivial) equilibrium given by $\alpha = 1$.
\item \textbf{Partial Participation} For all $P < \bar{P}$, there is a unique (non-trivial) equilibrium satisfying $\alpha < 1$. This equilibrium is given by  
\[
\alpha =  \frac{1}{\numbuyers - k}\left[\Q - \frac{B_{\leq k}}{P \numusers} \right]
\]
where k is the unique integer with $0 \le k \le K-1$ that satisfies
\[
\gamma(k) < P \leq \gamma(k+1). 
\]
\begin{figure}[ht]
    \begin{center}
    \includegraphics[width=0.5\textwidth]{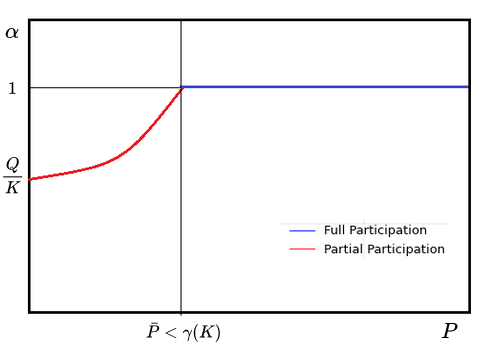}
    \caption{\footnotesize{The user participation rate $\alpha$ (y-axis) as a function of the platform price $P$ (x-axis), in the moderate benefit case of $\numbuyers > \Q B_{\numbuyers} \geq B_{\leq \numbuyers}$. Partial participation (in red) arises when the price satisfies $P < \overline{P}$, while full participation (in blue) arises when $P \ge \overline{P}$. .}}\label{fig:alphavsP_constant_modben}
    \end{center}
\end{figure}
\end{itemize}
\end{theorem}

In the moderate benefit case, the type of equilibria we obtain now depends on the price, and in particular on how it impacts the fraction of the available data that is bought by the buyers. At low prices, buyers buy a higher fraction of the available data, increasing users' privacy costs. Then users with lower privacy valuations are willing to participate while users with higher privacy valuations are not, leading to partial participation. These partial participation equilibria disappear as the price $P$ becomes large enough, in  which case the buyers buy less and less data, lowering the privacy costs of the users. We plot these equilibria in Figure~\ref{fig:alphavsP_constant_modben}.

\begin{theorem}[Low Benefit Case]\label{thm:eq_exist_lowQ}
Suppose $\Q < \frac{B_{\leq \numbuyers}}{B_{\numbuyers}}$. There exists a price threshold $\bar{P} = \frac{B_{\leq K}}{\Q \numusers} = \gamma(K)$ such that:
\begin{itemize}
\item \textbf{Full Participation} For all $P > \bar{P}$, there exists a unique (non-trivial) equilibrium given by $\alpha = 1$.
\item \textbf{Partial Participation} For all $P < \bar{P}$, there is a unique (non-trivial) equilibrium satisfying $\alpha < 1$. This equilibrium is given by  
\[
\alpha =  \frac{1}{\numbuyers - k}\left[\Q - \frac{B_{\leq k}}{P \numusers} \right]
\]
where k is the unique integer with $0 \le k \le K-1$ that satisfies
\[
\gamma(k) < P \leq \gamma(k+1). 
\]
\item \textbf{Mixed Participation} At $P = \bar{P}$, the set of participation equilibria is given by
\[
\alpha \in \left[\frac{Q B_K}{B_{\leq K}}, 1\right].
\]
\end{itemize}
\end{theorem}

\begin{figure}[ht]
    \begin{center}
    \includegraphics[width=0.5\textwidth]{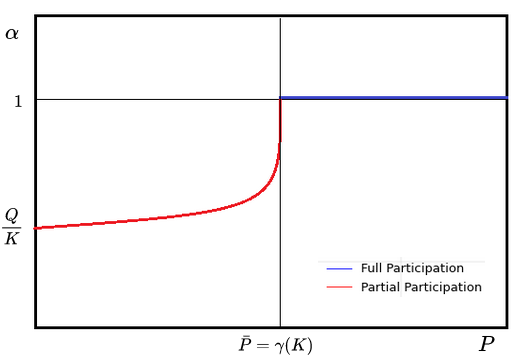}
    \caption{\footnotesize{The user participation rate $\alpha$ (y-axis) as a function of the platform price $P$ (x-axis), in the ow benefit case where $\Q B_{\numbuyers} < B_{\leq \numbuyers}$. Observe the presence of multiple equilibria at $\bar P = \gamma(\numbuyers) = \frac{Q B_K}{B_{\le K}}$.}}\label{fig:alphavsP_constant_lowben}
    \end{center}
\end{figure}
We plot these equilibria in Figure~\ref{fig:alphavsP_constant_lowben}. The insights of Theorem~\ref{thm:eq_exist_lowQ} are similar to those of Theorem~\ref{thm:eq_exist_modQ}, but there are two slight differences. The first is that at a fixed $P$, the partial participation equilibria are lower than in the medium benefit case; this is unsurprising, as the benefit from participation is lesser, and incentives to participate are weaker. Perhaps more interestingly, we note that there exists a single price where there is a continuum of equilibria that ``connect'' the partial participation equilibria with the full participation equilibria. The intuition is the following: at this threshold price, we reach a point where for $\alpha$ big enough, there is no buyer in the marketplace that can buy all the available data on the marketplace anymore; instead, each buyer $k$ buys exactly $B_k/\bar P$ data points. The total amount of data on the market is $\sum_k B_k/\bar P = B_{\leq K}/\bar{P}$, and each user $v$ has utility $\Q - v \cdot B_{\leq K}/(\alpha N \bar P) = \Q - v \Q/\alpha$. Any participation rate that is high enough is then stable: if $\alpha = v$, a user with valuation $v$ has utility exactly $0$, and only users with privacy valuation $v$ or less are willing to participate\footnote{This is independent of tie-breaking so long as $f(v)$ has no point mass. In this case, the participation rate is exactly $F(v)$, independently of how users with privacy valuation $v$, the only ones that are indifferent between participation and non-participation, behave.}. This continuum disappears once $P > \bar{P}$, noting that a user's utility is $\Q - v \cdot B_{\leq K}/(\alpha N P) > \Q - v \Q/\alpha$, and setting an initial participation rate $\alpha < 1$, users whose privacy valuation slightly higher than $\alpha$ still get positive utility and have an incentive to participate; these users then decide to participate and to deviate from the current solution. \\

\textbf{Comparison with the setting of exogenous privacy costs:}
To show the impact of endogeneity, we consider a simpler alternate scenario where user privacy costs are completely exogenous and do not depend on how many buyers obtain a given user's data. Each user now has a privacy cost of participation given by a non-negative real-number $c$. For simplicity of exposition, we persist with the assumption that their privacy costs are uniformly distributed in the interval $[0, V]$, and $F(x) = x/V$. 

In this setting, any user who chooses to participate on the platform gets a fixed benefit $\Q$. Therefore, a user with cost of participation $c$ will participate if and only if $\Q \geq c$. This implies that the user participation rate at equilibrium is given exactly by 
\[
\alpha = F(\Q) =  \min\left(1, \frac{\Q}{V}\right).
\]

Importantly, the participation rate $\alpha$ is independent of $P$. We plot the exogenous equilibria in Figure \ref{fig:exogenous_constant} as a function of the platform benefit $\Q$ instead. We remark that the user participation rate increases linearly with $\Q$ until it reaches full participation $\alpha = 1$ at $\Q = V$, beyond which the platform benefit is high enough that it always induces full participation.
\begin{figure}[!ht]
  \centering
  \includegraphics[width=.4\textwidth]{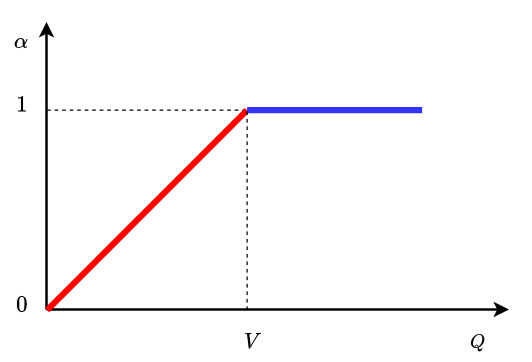}
  \caption{Market equilibria in the constant benefit case when user privacy costs are exogenous and uniformly distributed in $[0, V]$.}
  \label{fig:exogenous_constant}
\end{figure}

We see that when privacy costs are exogenous, we still have a \textit{low benefit} ($\Q < V$) and \textit{high benefit} ($\Q \geq V$) regime. However, for a fixed $\Q$, the participation rate $\alpha$ at equilibrium is \textbf{unique} (either partial participation or full participation) and \emph{always independent of the price $P$}. On comparing with the endogenous setting, we see that the equilibrium outcomes are exactly identical in the \textit{high benefit} case, with $\alpha = 1$ being the unique equilibrium irrespective of price $P$. However, there are stark differences between the two settings at smaller values of $\Q$. In the endogenous privacy cost setting, there is base level of participation which is determined by $\Q$, but higher levels of $\alpha$ can also be achieved at the same $\Q$ as $P$ increases; changing $P$ can in fact induce a range of market equilibria $\alpha$ (both partial participation and full participation types) for a fixed $\Q$.

An important aspect of the endogeneity of our model is that it shows how a platform can incentivize participation by providing more privacy, here entirely through the way it prices data re-sale to third parties. $P$ acts as a knob that can be turned to control the privacy promised by the platform. This is in contrast with the world in which privacy attitudes are exogenous and one has to incentivize further participation is through improved quality of service (or through lowering users' privacy costs through implementing explicit privacy techniques such as differential privacy, which we leave to previous work, e.g.~\cite{fallah2023optimal,satml2023}).

\subsection{Market Equilibria in the Linear Benefit Case}\label{sec:linear}

We consider the setting where the benefit offered by the platform is linearly increasing in the number of participating users, i.e., there exists a constant $C > 0$ such that $\Q = C \alpha \numusers$. This assumption is not too far-fetched. For example, consider social networking platforms like Facebook. Users derive much higher utility from participating if there are already many other users participating. In this case, the updated equilibrium condition, derived directly from Equation~\eqref{eq:nonlinear_sys}, is as below :
\begin{align}\label{opt:lin}
    & \alpha = \min \left(1, \frac{C \alpha^2 \numusers^2}{\sum_{k=1}^{\numbuyers}N_k} \right)~~~\text{and}~~~
    N_k = \min \left( \alpha \numusers, \frac{B_k}{P} \right) \quad \forall ~k \in [\numbuyers]. 
\end{align}

We now introduce our main results. We show that the market can induce completely different equilibria structures depending on the marginal benefit $C$ offered by the platform.
Theorem \ref{thm:lin_highben} demonstrates that when the platform offers a high enough marginal benefit, all users are ready to forego their privacy costs and participate on the platform. In Theorem \ref{thm:lin_lowben} and Figure~\ref{fig:alphavsP_lin_lowcase}, we consider a setting where the benefit offered by the platform is not very high and therefore, users care a lot more about their privacy costs. In this setting, we show that there is price threshold which separates the regime where users participate from where they do not. Finally, in Theorem \ref{thm:lin_splcase} and Figure~\ref{fig:alphavsP_lin_splcase}, we consider a special case where $C\numusers = \numbuyers$. We show that it is now possible to have infinitely many equilibria at the same price.
The proofs for all theorems are provided in Appendix~\ref{app:linear}, where we in fact provide a stronger version of our theorems that fully characterize buyer behavior. 

\begin{theorem}[High Benefit Case]\label{thm:lin_highben}
When $C\numusers > \numbuyers$, then for all prices $P > 0$, there exists a unique non-trivial equilibrium with $\alpha = 1$. 
\end{theorem}

This is perhaps the most intuitive case. When $C \numusers > \numbuyers$, the maximum cost incurred by a user is $v \sum_k \frac{N_k}{\alpha \numusers} \leq \numbuyers$. If all users participate and $\alpha = 1$, they then each user gets utility $C \alpha \numusers - v \sum_k \frac{N_k}{\alpha \numusers} \geq C \numusers  - \numbuyers  > 0$. I.e., no user wants to leave the platform, and we have an equilibrium. 

\begin{theorem}[Low Benefit Case]\label{thm:lin_lowben}
When $C\numusers < \numbuyers$, there exists a unique price threshold $\bar{P}$ such that:
\begin{enumerate}
\item For $P < \bar{P}$, no non-trivial equilibrium exists.
\item For $P \geq \bar{P}$, full participation $(\alpha = 1)$ is always an equilibrium.
\item For $P > \bar{P}$, there exists a unique partial participation equilibrium $(0< \alpha < 1)$ of the form $\alpha = \frac{A}{P}$ for an instance-dependent\footnote{$A$ only depends on $(B_1,\ldots,B_k)$, $C$ and $N$.} constant $A$ if and only if there exists $\hat k$ such that $\frac{B_{\hat k}}{N} < \frac{\sum_{j \leq \hat k} B_j}{C\numusers^2 - \numusers (\numbuyers - \hat k)} \leq \frac{B_{\hat k+1}}{N}$. Further, $A$ can be computed in polynomial time in $K$. 
\end{enumerate}
\end{theorem}

\begin{figure}[t]
    \begin{center}
    \includegraphics[width=0.5\textwidth]{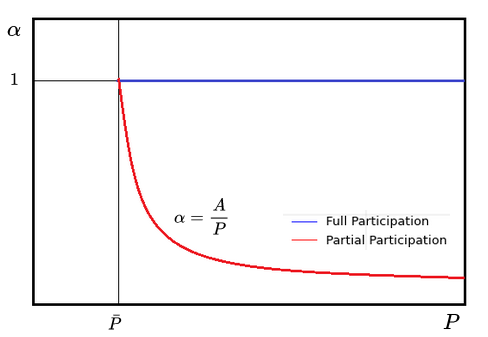}
    \caption{\footnotesize{This figure shows the evolution of user participation rate $\alpha$ on the y-axis as a function of the platform price $P$ on the x-axis. This is for the low benefit case when $C\numusers < \numbuyers$ and $\hat k$ exists. The red line represents partial participation equilibria while the blue line represents full participation equilibria. It is clear that there are multiple equilibria at a single price $P$.}}\label{fig:alphavsP_lin_lowcase}
    \end{center}
\end{figure}
New and interesting partial participation equilibria arise in Theorem~\ref{thm:lin_lowben} and Figure~\ref{fig:alphavsP_lin_lowcase}. Here, we note that $\alpha$ is decreasing in $P$! This goes against the intuition that higher $P$ should lead to less data being bought by the buyers, which increases user incentives to participate. The intuition here is that one can create a self-reinforcing feedback loop that works as follows: as $P$ increases, $\alpha$ decreases; then, user utilities also decrease; this incentives fewer users to participate, leading to a decreasing $\alpha$.  

\begin{theorem}[Special Case]\label{thm:lin_splcase}
When $C\numusers = \numbuyers$, at any price $P > 0$, there exists: 
\begin{enumerate}
    \item a non-trivial equilibrium with $\alpha = 1$. 
    \item Any $\alpha \in \left(0, \frac{B_1}{P\numusers}\right] \cap \left(0, 1 \right)$ is a partial equilibrium.
\end{enumerate}
\end{theorem}

\begin{figure}[h]
    \begin{center}
    \includegraphics[width=0.5\textwidth]{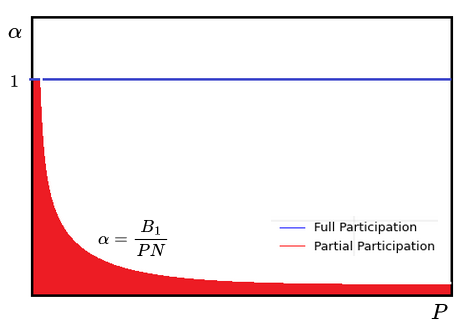}
    \caption{\footnotesize{This figure shows the evolution of user participation rate $\alpha$ on the y-axis as a function of the platform price $P$ on the x-axis. This is for the special case when $C\numusers = \numbuyers$. The shaded red region represents partial participation equilibria while the blue line represents full participation equilibria. It is clear that there can be infinitely many equilibria at a single price $P$.}}\label{fig:alphavsP_lin_splcase}
    \end{center}
\end{figure}

Theorem~\ref{thm:lin_splcase} and Figure~\ref{fig:alphavsP_lin_splcase} show similar insights that there exists a full participation equilibrium at all prices $P > 0$, like in Theorem~\ref{thm:lin_highben}. A major difference here, however, is that there are also infinitely many partial participation equilibria at each price $P$, upper bounded by $\frac{B_1}{P\numusers}$. The intuition here is that for any price $P$ and for sufficiently small $\alpha$ ($0 < \alpha \leq \frac{B_1}{P\numusers}$), every buyer has enough money to buy all the data available on the market. In this scenario, any user $v$ participating on the platform incurs a privacy cost of $v  \cdot \left(\numbuyers \alpha \numusers \right)$. Observe that user $v$ has an incentive to participate if and only if $C \alpha^2 \numusers^2 - v\cdot (\numbuyers \alpha \numusers) =  C \alpha^2 \numusers^2 - v\cdot (C \numusers^2 \alpha) \geq 0$ (as $\numbuyers = C \numusers$), i.e., $v \leq \alpha$. Thus, at price $P$, every $\alpha$ such that $0 < \alpha \leq \frac{B_1}{P\numusers}$ is feasible. \\ 

\textbf{Comparison to the setting with exogenous privacy costs:} For the linear benefit case, we again do a comparative analysis with the setting when all user privacy costs are exogenous. In this case, the benefit of joining the platform is given by $Q(\alpha) = C\alpha \numusers$ when the user participation rate is $\alpha$. At this point, a new user with cost of participation $c$ will join if and only if: 
\[
       C\alpha \numusers - c \geq 0, \quad \text{or equivalently} \quad c \leq C\alpha \numusers.
\]
This implies that the user participation rate at equilibrium is given by $\alpha = F(C \alpha \numusers) = \min \left(1, \frac{C\alpha \numusers}{V}\right)$. Figure \ref{fig:exogenous_linear} depicts all possible market equilibria under different regimes of $C\numusers$. If $C\numusers > V$, $\alpha = 0$ and $\alpha = 1$ are the only two possible equilibria. When $C\numusers = V$, all values of $\alpha \in [0,1]$ are feasible equilibria. Finally, if $C\numusers < V$, there are no equilibria with $\alpha > 0$.\\
\begin{figure}[!ht]
  \centering
  \includegraphics[width=.50\textwidth]{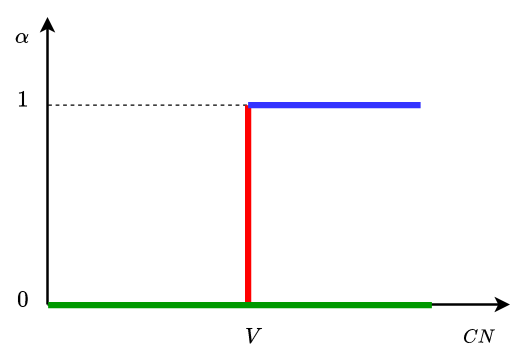}
  \caption{Market equilibria in the linear benefit case when user privacy costs are entirely exogenous and uniformly distributed in $[0, V]$.}
  \label{fig:exogenous_linear}
\end{figure} 

Figure \ref{fig:exogenous_linear} indicates that when user privacy costs are exogenous, the market equilibria has a very simple structure, again completely \emph{independently of the platform price $P$}. For low values of the benefit, there are no non-trivial market equilibria, but once the benefit offered is high enough, $\alpha = 1$ emerges as the only non-trivial equilibrium. The exogenous and endogenous settings are identical with respect to the high benefit case (producing only full participation and $\alpha = 0$ equilibria), but the outcomes can vary largely when benefits offered on the platform are smaller. 

In the exogenous setting, there is a threshold value of benefit below which the market admits no equilibria with $\alpha > 0$. This is not true in the endogenous cost setting which can induce user participation even at low levels of benefit. Instead of a benefit threshold, the endogenous setting has a price threshold (at any given level of benefit $C\numusers$) and all non-trivial equilibria emerge above the price threshold. It is important to note that as the benefit parameter $C$ decreases, the price threshold is higher. Basically, in the small benefit regime, the platform has to set high enough prices so that it can maintain incentives for user participation by lowering their expected privacy costs using the price signal. The endogenous setting admits a rich equilibrium structure above the price threshold, with the existence of both full and partial participation equilibria.

\section{Experiments}\label{sec:experiments}
We now aim to complement our theoretical results with experiments that expand along two dimensions: first, we consider more general distributions of privacy attitudes that are informed by real-life surveys~\cite{privacy_measure} on privacy attitudes and inspired from previous work~\cite{personalized_privacy} on personalized differential privacy. Second, we consider more general benefit functions for participating users that capture the idea of ``diminishing returns'': i.e., a user on the platform may not be able to interact with or get some benefit from every other user on the platform, and in turn may get less marginal benefit from each additional user past a certain point. 

\subsection{Beyond Uniform Privacy Valuations}

In this section, we consider a more general model of privacy valuation that is inspired by~\cite{privacy_measure} and~\cite{personalized_privacy}.~\cite{privacy_measure} run a survey in which participants reported their privacy attitudes, ranked as ``low'', ``medium'', and ``high'' levels of concerns. They show that a $0.107$ fraction of the participants had a low level of privacy, vs $0.355$ for medium and $0.537$ for high. Based on these considerations,~\cite{personalized_privacy} build a distribution of differential privacy parameters $\varepsilon$ that puts a probability mass of $0.537$ on $\varepsilon \in [0,\varepsilon_M)$ (high privacy concerns)\footnote{Note that in differential privacy, a lower value of $\varepsilon$ corresponds to a higher privacy requirement.}, of $0.355$ on $\varepsilon \in [\varepsilon_M,1)$ (medium privacy concerns), and of $0.107$ on $\varepsilon = 1$ (low privacy concerns), for some $\varepsilon_M \in (0,1)$.

For the sake of our experiments, instead of working with privacy requirements in terms of differential privacy parameter $\varepsilon$ where lower epsilon corresponds to higher privacy concerns (see~\cite{dworkroth} for more details on differential privacy), we work with privacy valuations, where a higher valuation corresponds to more stringent privacy concerns. Using the data from~\cite{privacy_measure} and inspired by the approach of~\cite{personalized_privacy}, we define our ranges of privacy valuations to be $v = 0$ with probability $0.107$ (lowest privacy attitude, there is no cost to privacy), $v$ uniform in $(0,v_M]$ (medium privacy attitude) with probability $0.537$, and $v$ uniform in $(v_M,1]$ (high privacy attitude) with probability $0.355$. We refer to this distribution as the \emph{personalized privacy} distribution. In our experiments, we provide results when varying $v_M \in (0,1)$, provided in Figure~\ref{fig:custom_privacy_constant} for the case of a constant benefit function $Q(\alpha) = \Q$, and in Figure~\ref{fig:custom_privacy_linear} for the case of a linear benefit function $Q(\alpha) = C \alpha N$.

\begin{figure}[!ht]
  \centering
  \raisebox{35pt}{\parbox[b]{.09\textwidth}{Low}}%
  \subfloat[][]{\includegraphics[width=.30\textwidth]{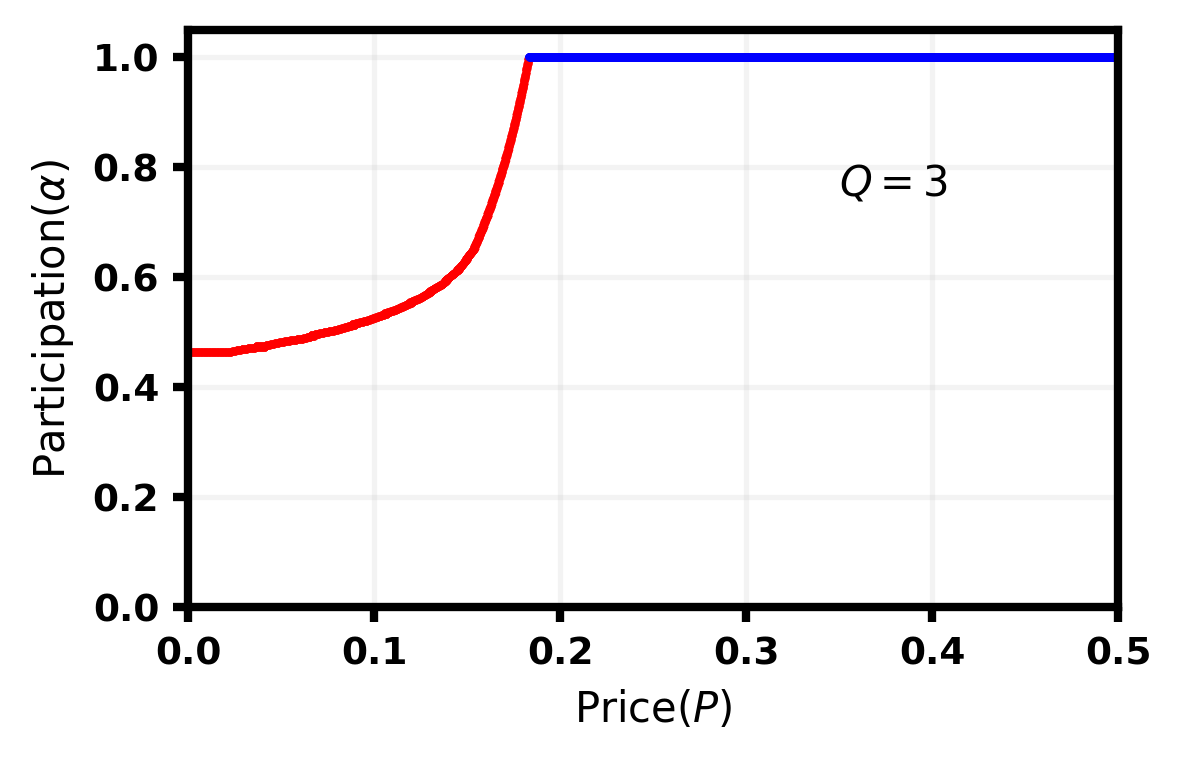}}\hfill
  \subfloat[][]{\includegraphics[width=.30\textwidth]{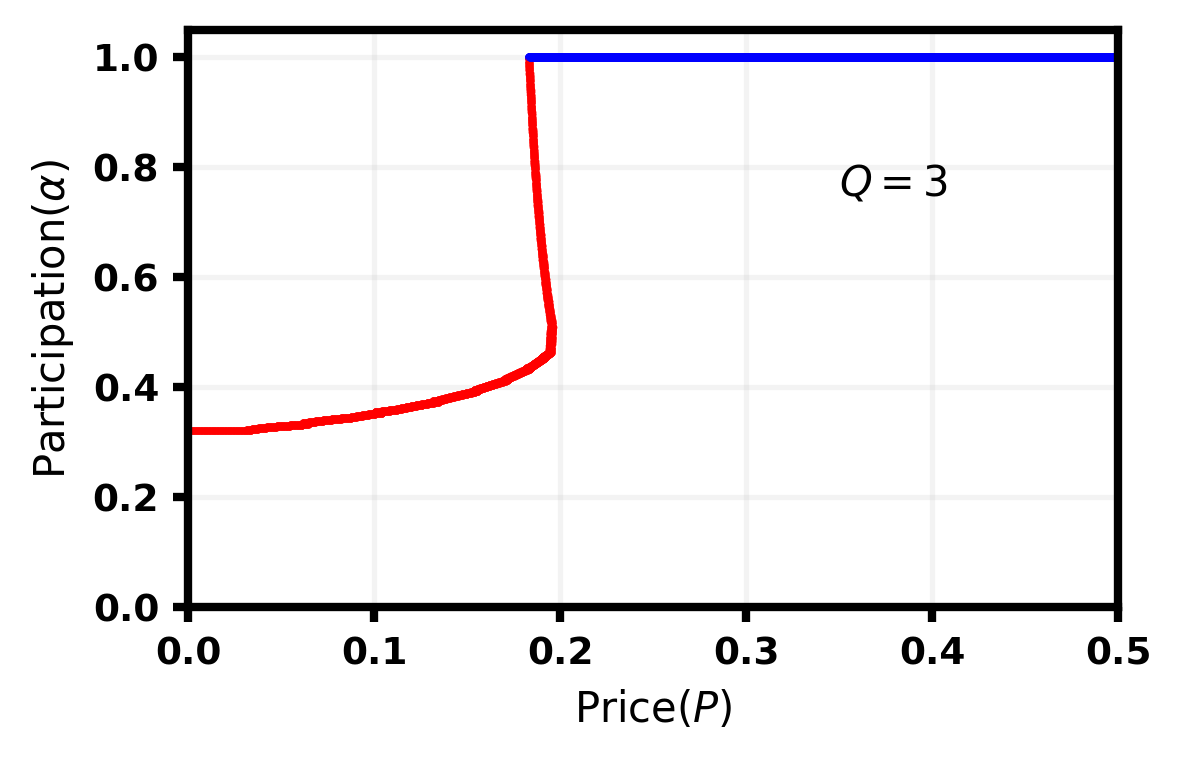}}\hfill
  \subfloat[][]{\includegraphics[width=.30\textwidth]{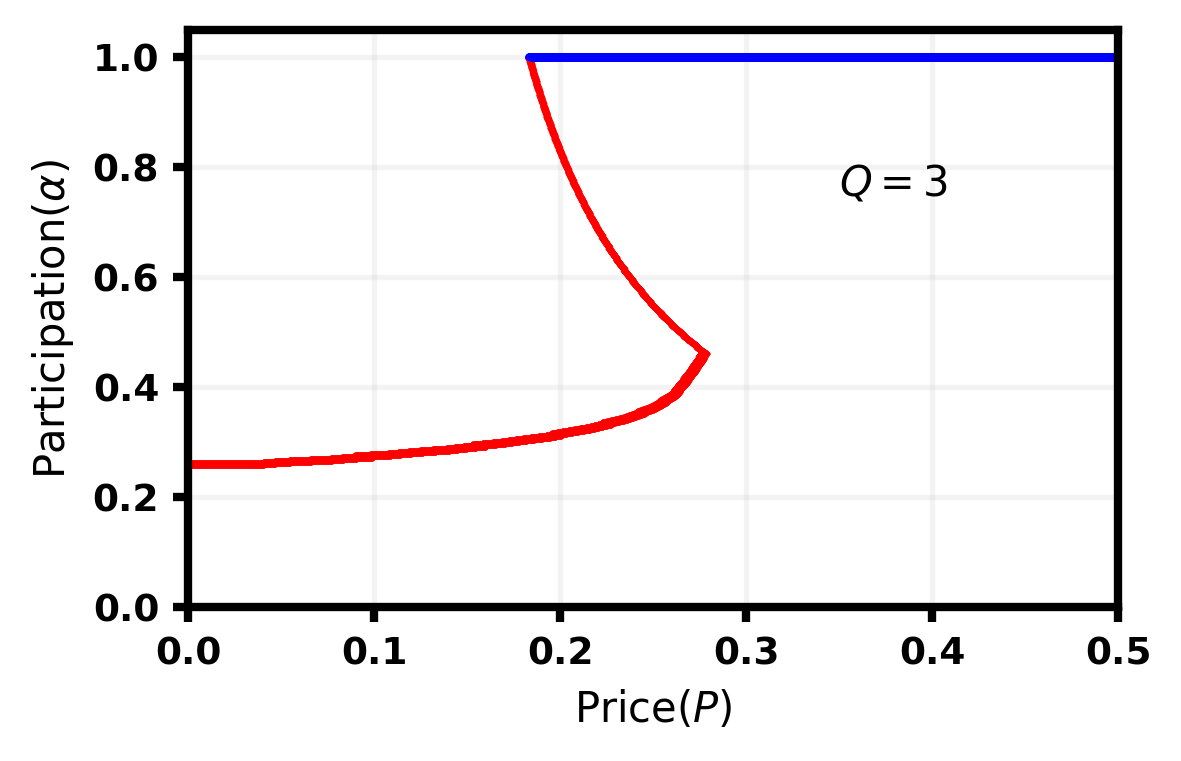}}\par
  \raisebox{35pt}{\parbox[b]{.09\textwidth}{Mod}}%
  \subfloat[][]{\includegraphics[width=.30\textwidth]{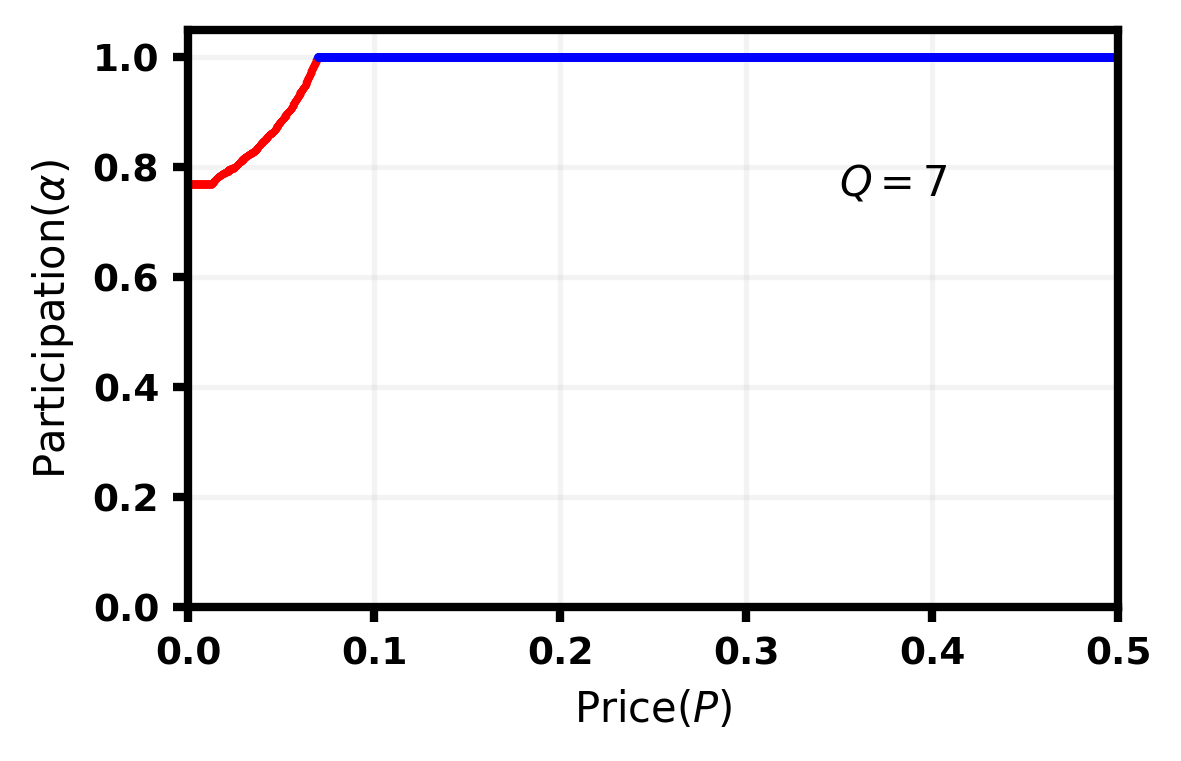}}\hfill
  \subfloat[][]{\includegraphics[width=.30\textwidth]{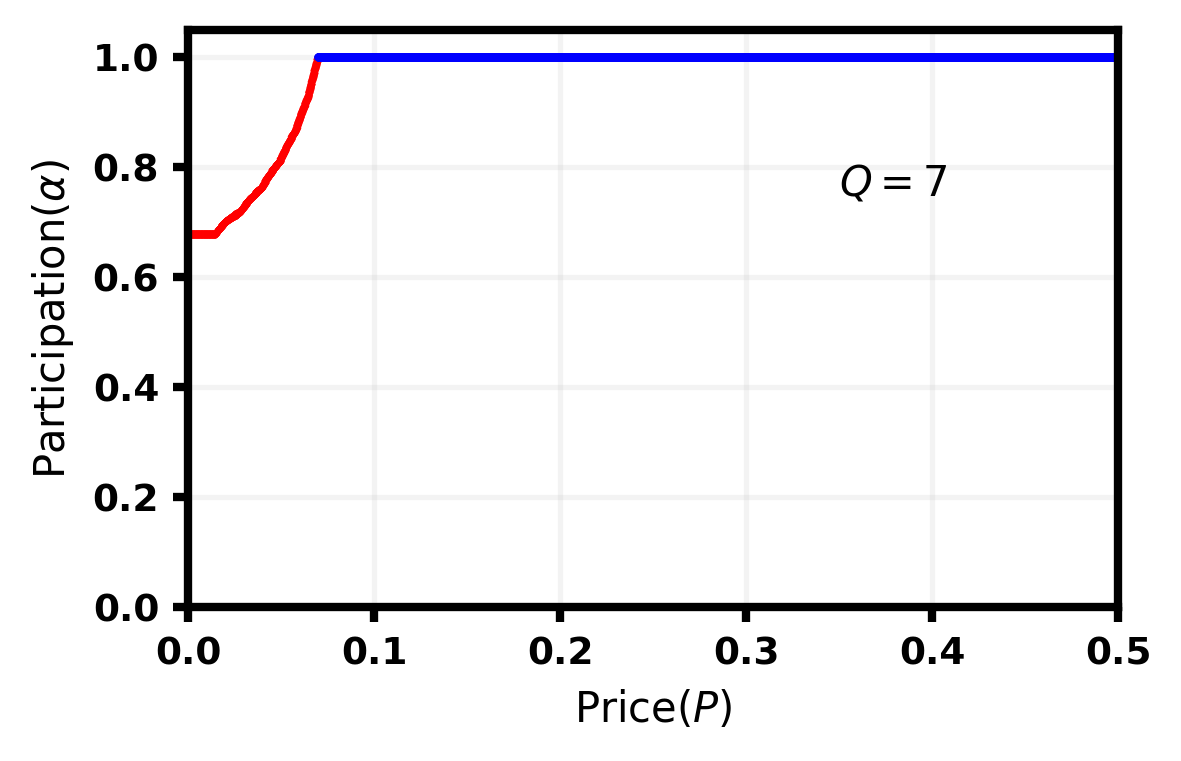}}\hfill
  \subfloat[][]{\includegraphics[width=.30\textwidth]{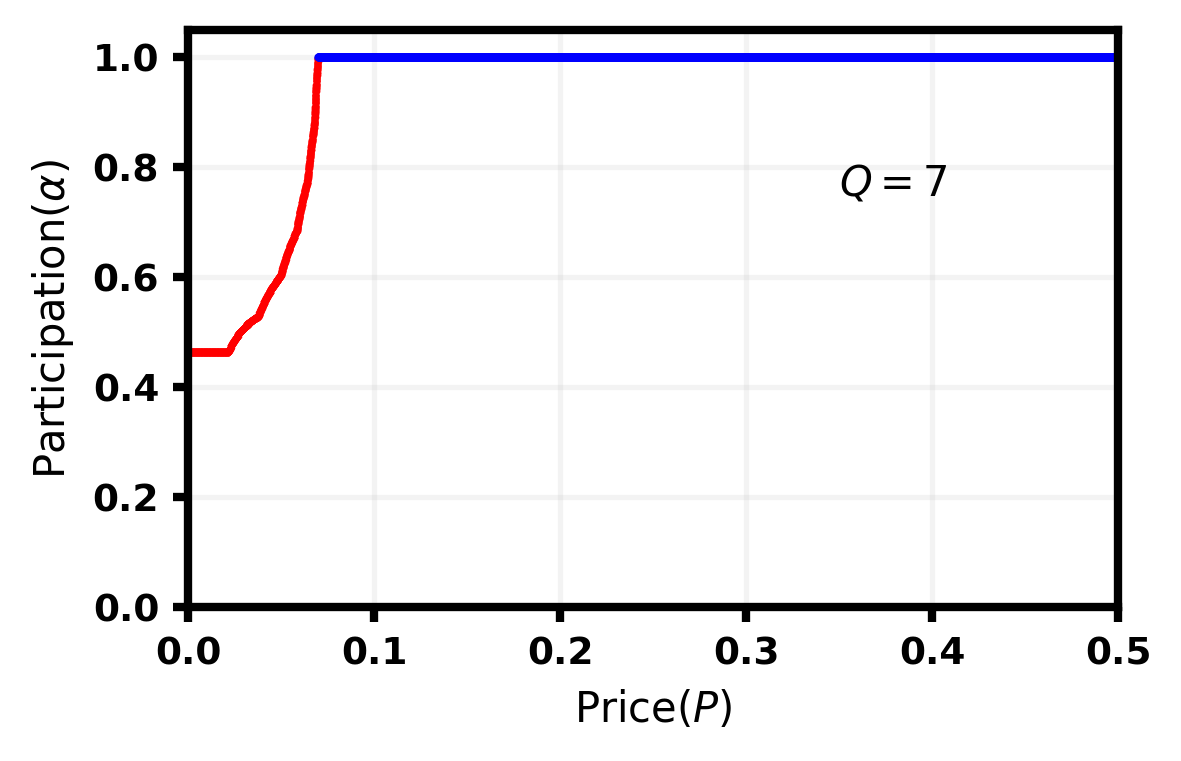}}\par
  \raisebox{35pt}{\parbox[b]{.09\textwidth}{High}}%
  \subfloat[][]{\includegraphics[width=.30\textwidth]{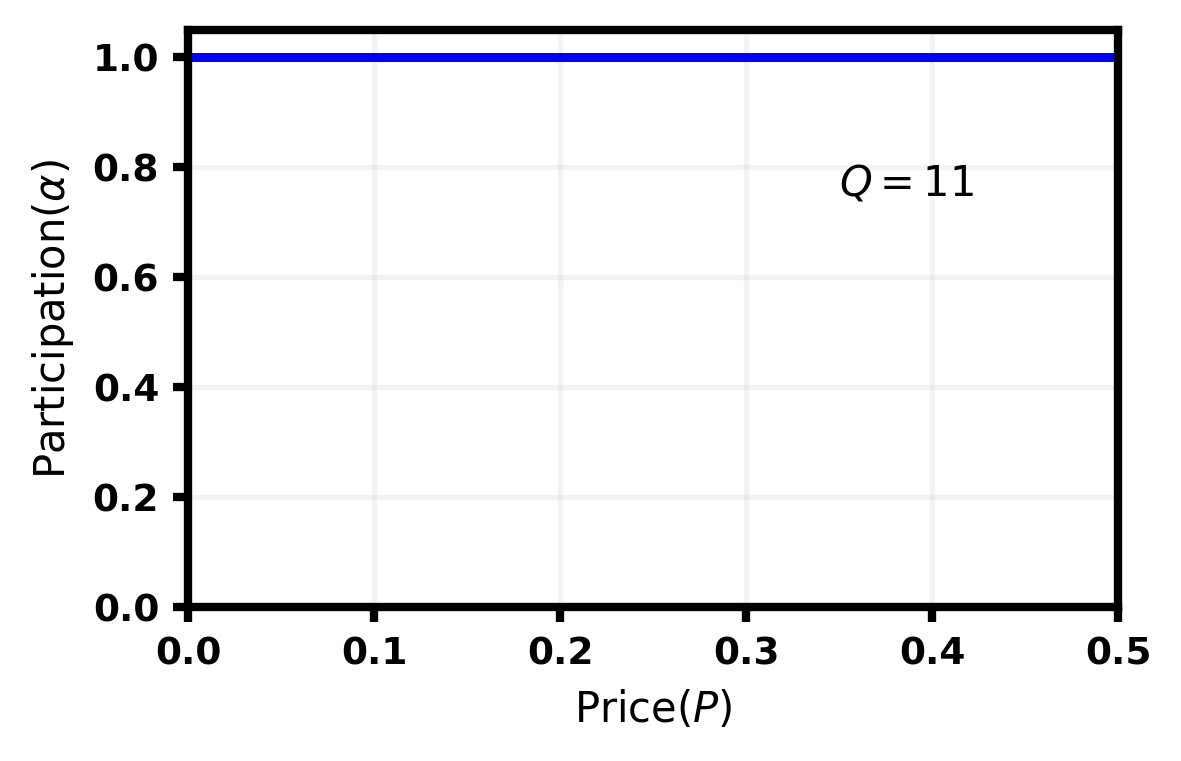}}\hfill
  \subfloat[][]{\includegraphics[width=.30\textwidth]{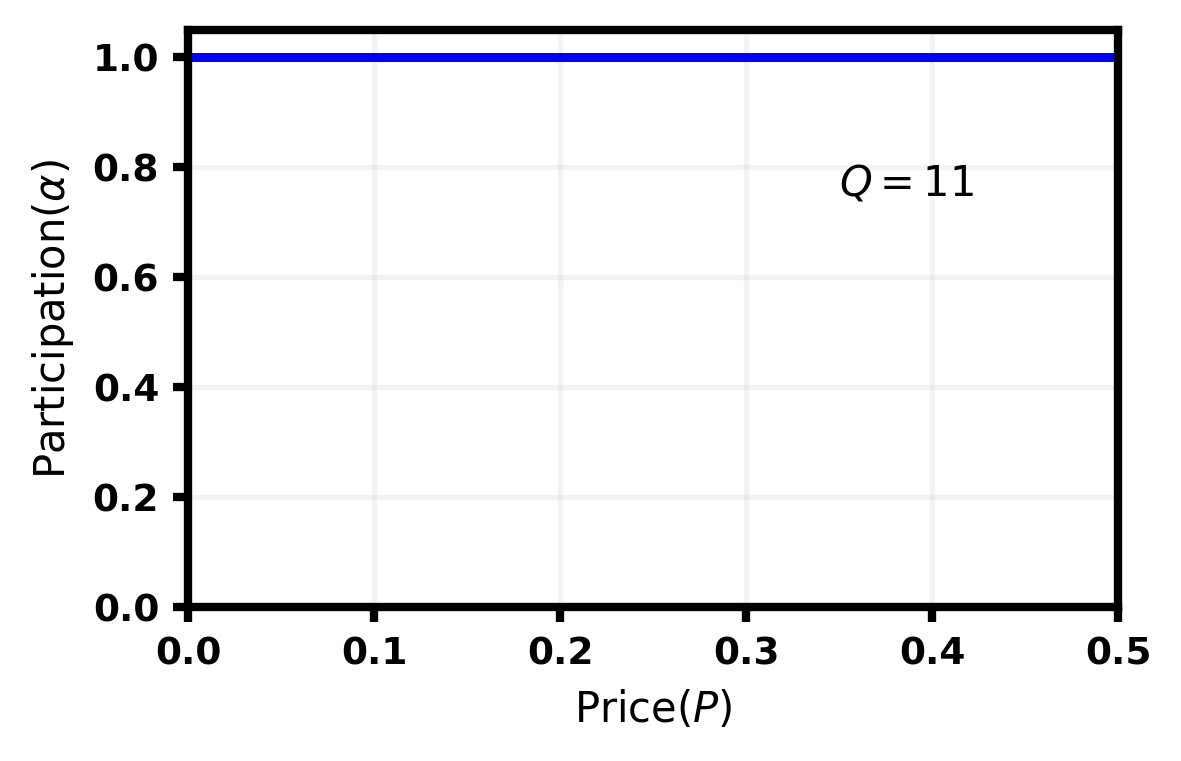}}\hfill
  \subfloat[][]{\includegraphics[width=.30\textwidth]{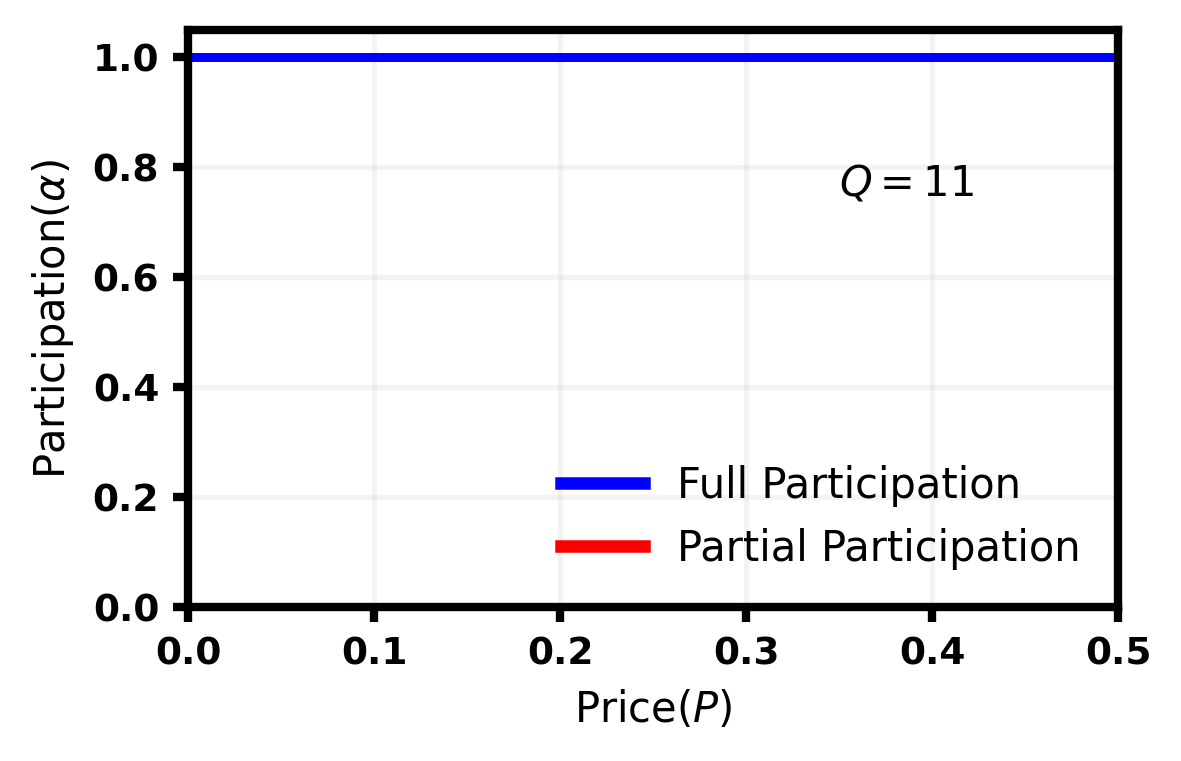}}
  \caption{Plots for the constant case under the personalized privacy distribution of user privacy valuations. From left to right, the values of $v_M$ are $0.3$, $0.5$ and $0.7$. Error tolerance $ = 2 \times 10^{-3}$.}
  \label{fig:custom_privacy_constant}
\end{figure}
Figure \ref{fig:custom_privacy_constant} captures the variation of participation rate $\alpha$ at equilibrium as a function of the platform price $P$ for \textit{low}, \textit{moderate} and \textit{high benefit} regimes. We show results for 3 different values of $v_M$ ($0.3$, $0.5$ and $0.7$) which parameterizes the distribution of user privacy valuations as discussed earlier. It is found that the outcomes in the \textit{moderate} and \textit{high benefit} regimes are consistent across changes in $v_M$ and closely track our results in the uniform distribution case. In the \textit{high benefit} regime, the quality of service $\Q$ is sufficiently high such that all users, irrespective of their privacy attitudes, find it beneficial to participate which explains that $\alpha = 1$ is the only non-trivial equilibrium at all prices $P > 0$. In the \textit{moderate benefit} regime, the general trend is that at smaller values of $P$, we have partial participation equilibria ($0 < \alpha < 1$) because users with high privacy valuations stay off the market. As $P$ increases, $\alpha$ increases steadily until it reaches full participation at some price threshold $\bar P$, beyond which $\alpha = 1$ is the only non-trivial equilibrium. It is important to note that as $v_M$ increases, the base participation rate ($\alpha$ close to $P = 0$) becomes smaller. This is expected because a higher $v_M$ indicates a higher average privacy valuation in the population. In the \textit{low benefit} regime, we make the following observations: At smaller prices, there are partial participation equilibria ($0 < \alpha < 1$) and $\alpha$ tends to increase as $P$ increases which is intuitive (higher prices lead to less data bought on the market hence lower privacy costs). Note that the base participation rates are smaller than their \textit{moderate benefit} regime counterparts due to smaller $\Q$. At higher values of $v_M$, we also observe multiplicity of equilibria, including a sequence of partial participation equilibria where $\alpha$ decreases with increasing price $P$; such equilibria, while counter-intuitive, can in fact be self-sustaining, as we observed in Section~\ref{sec:linear}. %This is highly counterintuitive. 
Finally, as $P$ becomes high enough, users' expected privacy costs decrease sufficiently to ensure that everyone participates.     
\begin{figure}[!ht]
  \centering
  \raisebox{20pt}{\parbox[b]{.09\textwidth}{}}%
  \subfloat[][$C = 0.0005$]{\includegraphics[width=.30\textwidth]{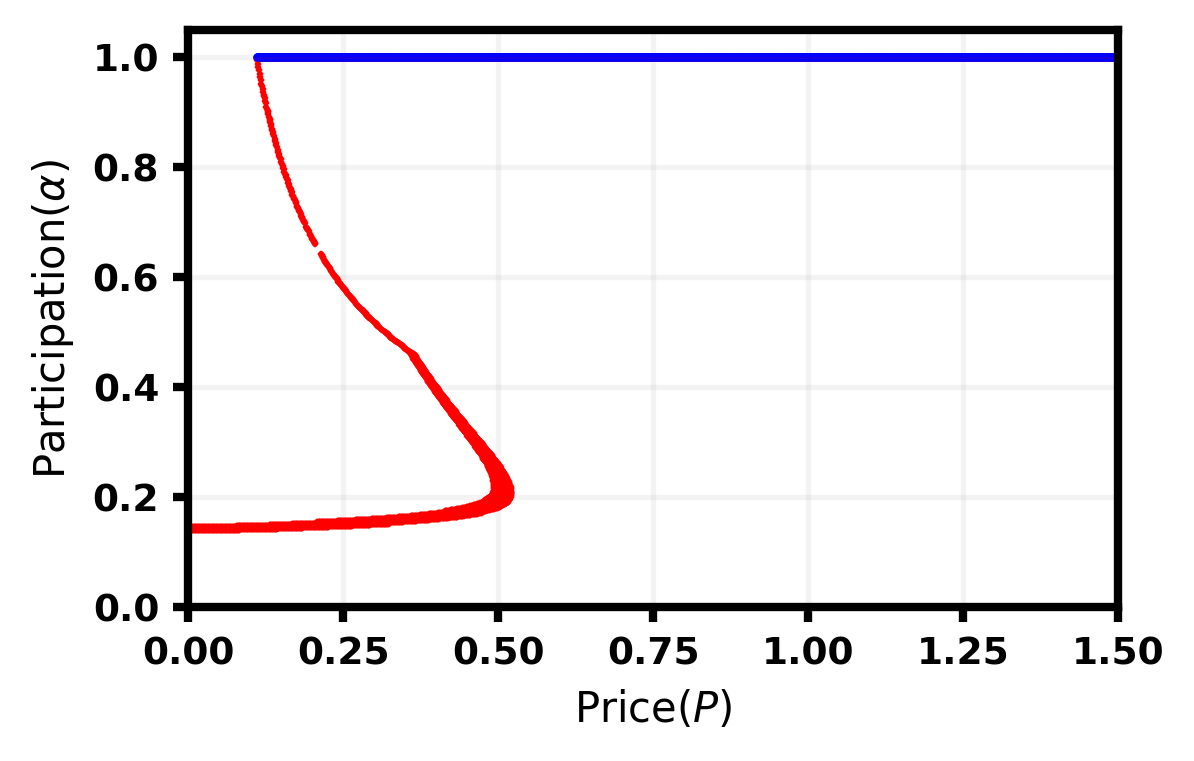}}\hfill
  \subfloat[][$C = 0.0008$]{\includegraphics[width=.30\textwidth]{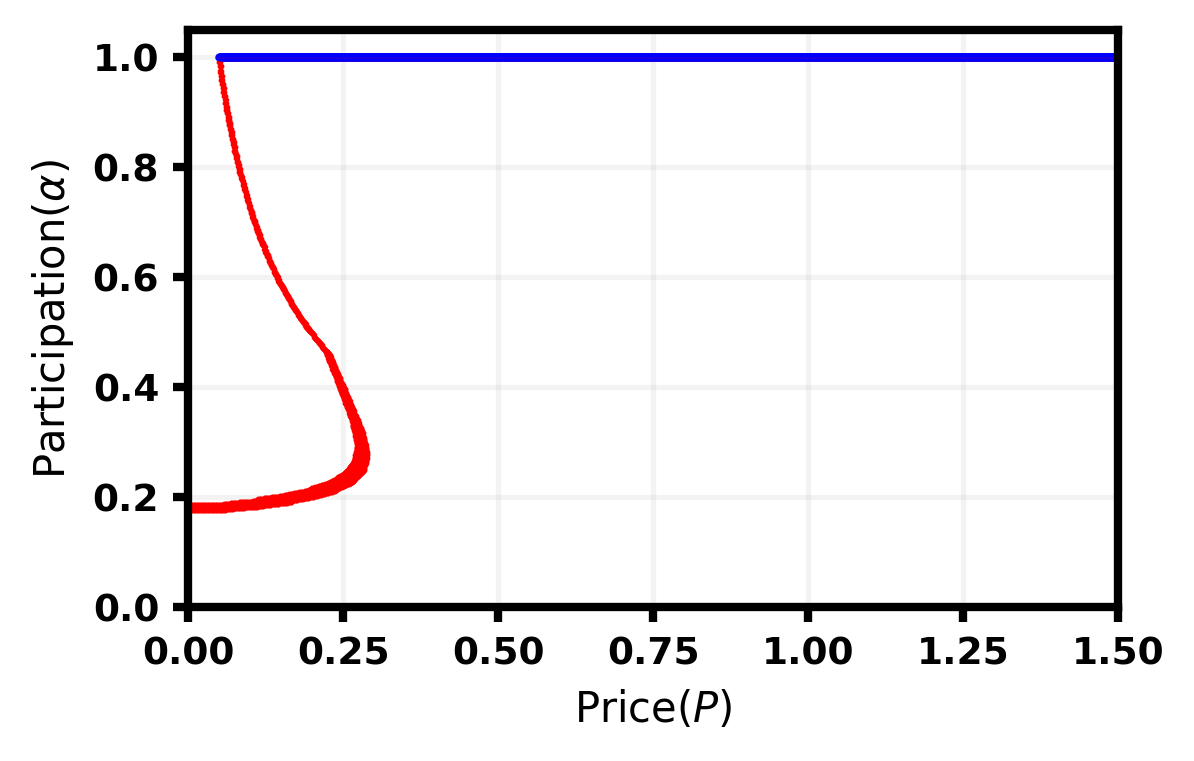}}\hfill
  \subfloat[][$C = 0.001$]{\includegraphics[width=.30\textwidth]{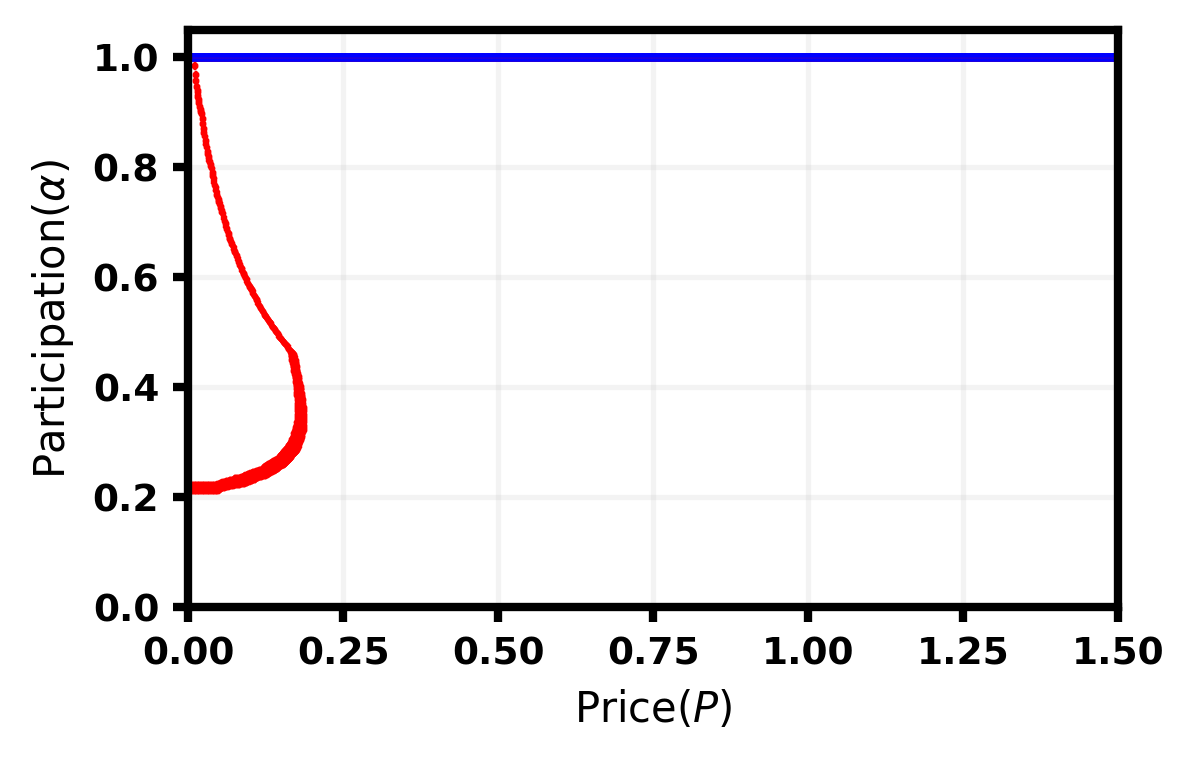}}\par
  \raisebox{20pt}{\parbox[b]{.09\textwidth}{}}%
  \subfloat[][$C = 0.0012$]{\includegraphics[width=.30\textwidth]{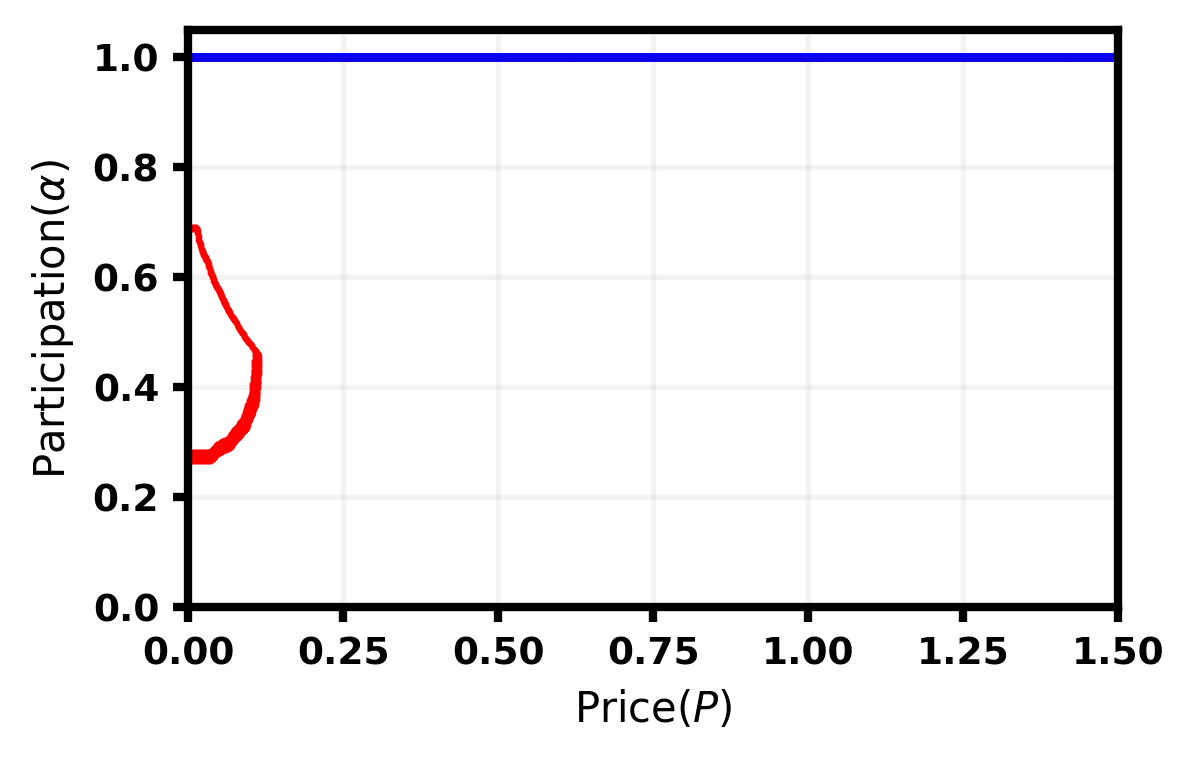}}\hfill
  \subfloat[][$C = 0.0015$]{\includegraphics[width=.30\textwidth]{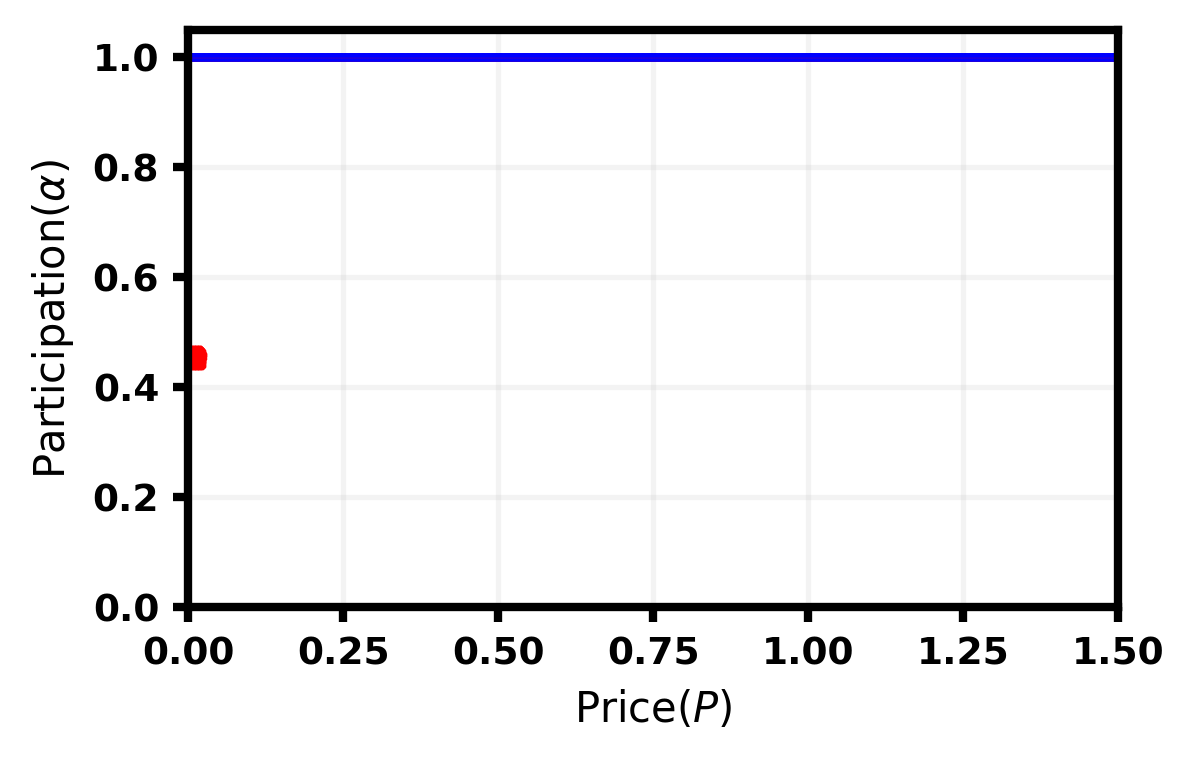}}\hfill
  \subfloat[][$C = 0.002$]{\includegraphics[width=.30\textwidth]{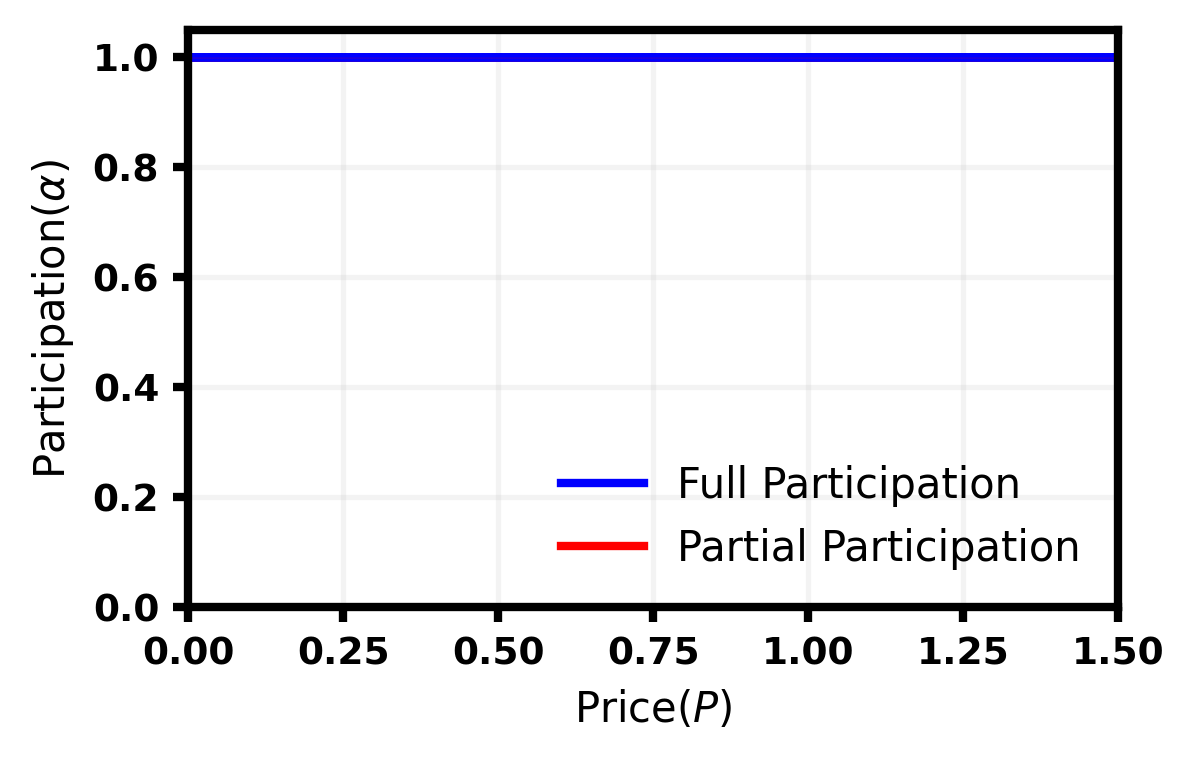}}\par
  \caption{Plots for the linear case under the personalized privacy distribution of user privacy valuations with $v_M = 0.7$ and increasing values of $C$. Error tolerance $ = 3\times 10^{-3}$.}
  \label{fig:custom_privacy_linear}
\end{figure}
In Figure \ref{fig:custom_privacy_linear}, we explore the equilibria characteristics in the case where the benefit function $Q(\alpha)$ is linear in $\alpha$, under the custom distribution of user privacy valuations. At lower values of $C$ (low benefit), there are two types of partial participation equilibria, one where $\alpha$ decreases with increasing $P$ and another where $\alpha$ increases with increasing $P$. This is in contrast with the uniform distribution case where all partial participation equilibria in the low benefit regime decrease with increasing price. As $C$ increases, the partial participation equilibria vanish rapidly until $C$ is high enough such that we have full participation at all prices $P > 0$. 

\subsection{Beyond Linear Utilities from Participation}

We now consider the benefit from participation that are not constant nor linear in $\alpha$. In particular, we are interested in modelling:

\begin{itemize}
\item \emph{Diminishing returns}: as more users join the platform, each existing user gets fewer and fewer marginal benefit from each additional participant. This models how any given user can only meaningfully interact with a small number of participants, and adding more users beyond a certain point is of little benefit. We illustrate these benefit functions in Figure~\ref{fig:dimreturn_functions}.
\item \emph{S-shaped}: here, users initially see an increasing marginal benefit from each additional user, until they reach a critical mass of users they benefit from interacting with. Then, diminishing returns kick in as per the above. We illustrate these S-shaped benefit functions in Figure~\ref{fig:beta_functions}.
 \end{itemize}

\paragraph{Diminishing Returns} In the case of diminishing returns, we study a simple functional form for user benefit, given by $Q(\alpha) = C (\alpha \numusers)^s$ for $s \in [0,1]$. Note that the constant case and the linear case are special cases of diminishing returns when $s = 0$ and $s = 1$ respectively. We vary the parameter $s$ in our experiments, noting that lowering the value of $s$ increases how fast returns diminish. At $s = 1$, the benefit function is an increasing line with constant returns, which in fact corresponds exactly to the linear benefit case; at $s = 0$, it is a flat line where there are no additional returns from adding participants on the platform. We plot our equilibria in Figure~\ref{fig:uniform_privacy_dim1} for $s=0.2$ and in Figure~\ref{fig:uniform_privacy_dim2} for $s = 0.7$, in both cases for several values of C.

\begin{figure}[!ht]
  \centering
  \raisebox{20pt}{\parbox[b]{.11\textwidth}{}}%
  \subfloat[][$s = 0.2$]{\includegraphics[width=.3\textwidth]{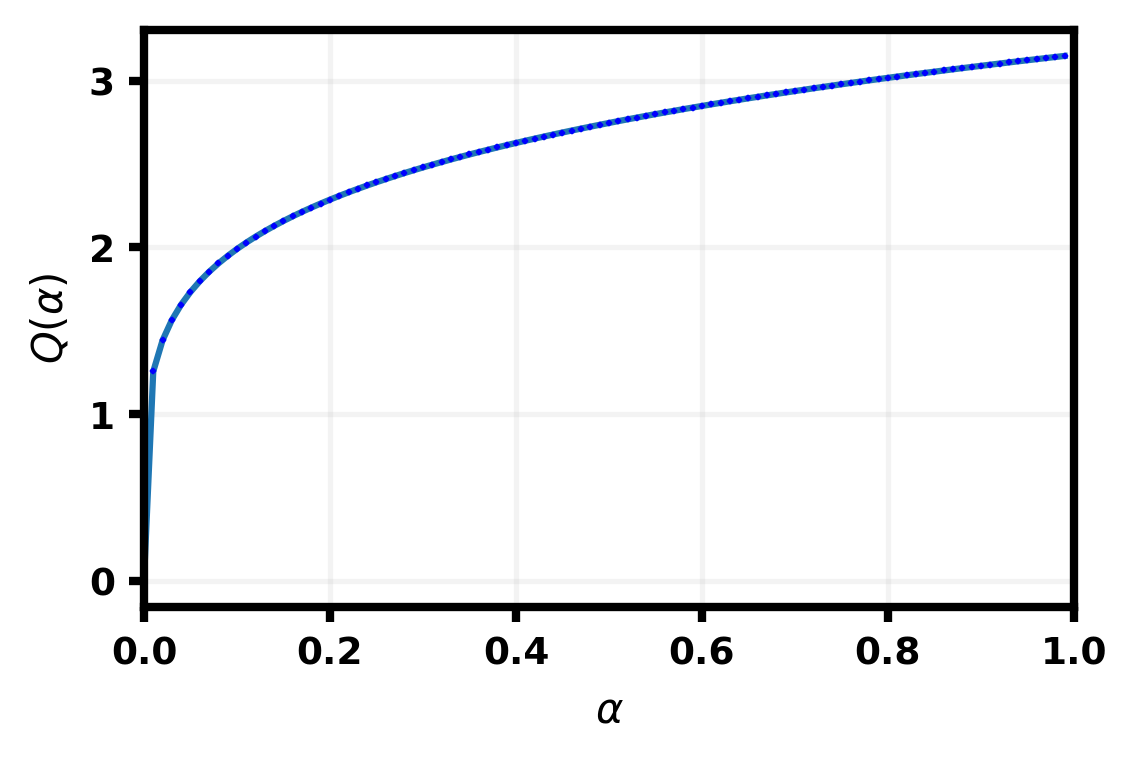}}\hfill
  \subfloat[][$s = 0.5$]{\includegraphics[width=.3\textwidth]{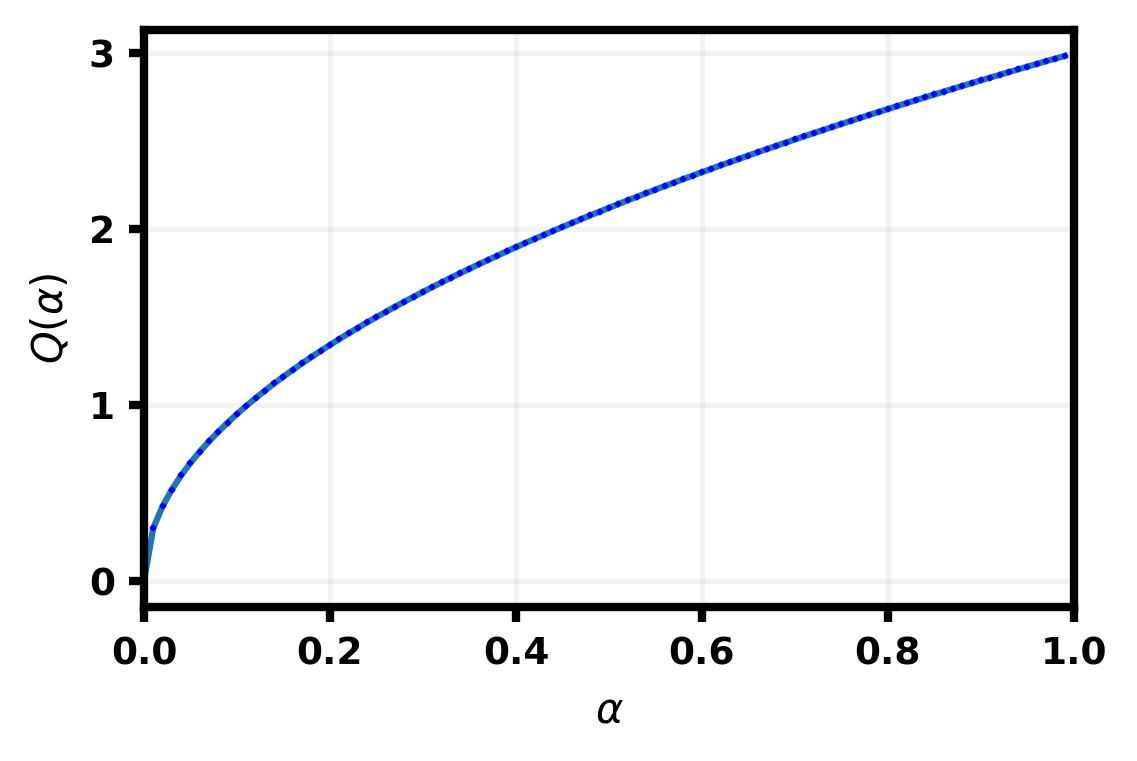}}\hfill
  \subfloat[][$s = 0.7$]{\includegraphics[width=.3\textwidth]{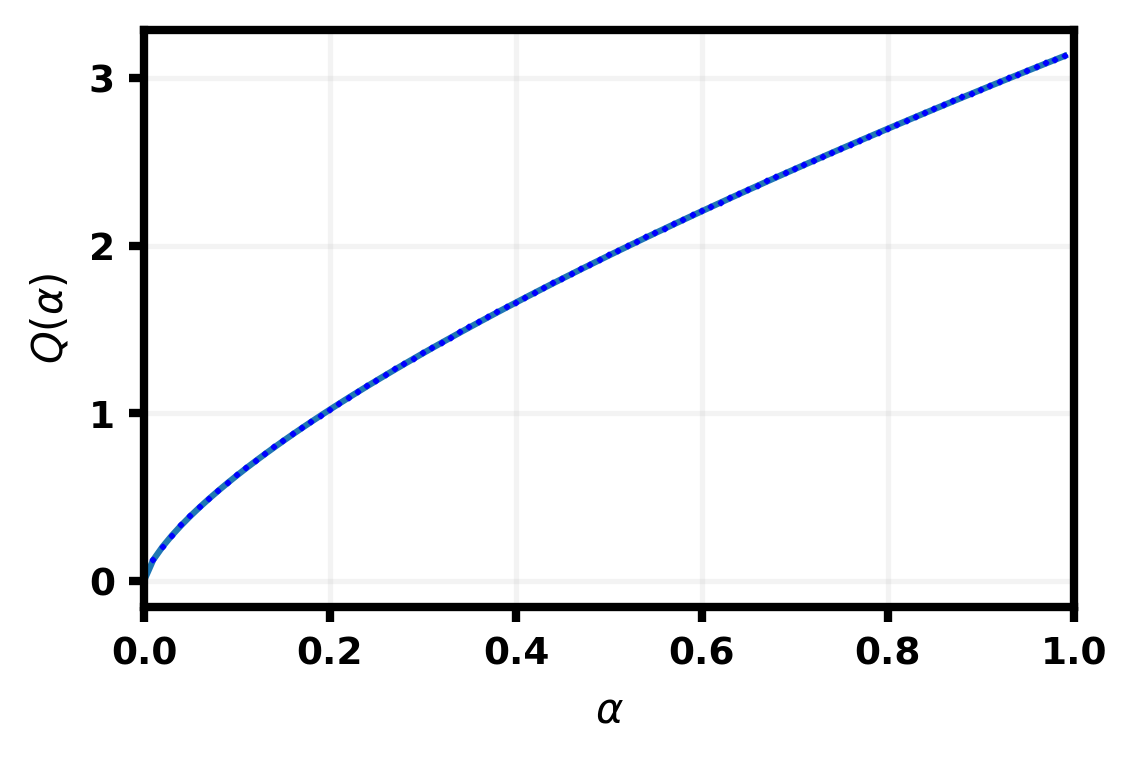}}\par
  \caption{Shape of the benefit function with diminishing returns for varying parameters $s \in [0, 1]$. As $s$ becomes smaller, the slope of the benefit function decreases faster, leading to stronger diminishing return effects.}
  \label{fig:dimreturn_functions}
\end{figure}

\begin{figure}[!ht]
  \centering
  \raisebox{20pt}{\parbox[b]{.11\textwidth}{}}%
  \subfloat[][$C = 0.025$]{\includegraphics[width=.3\textwidth]{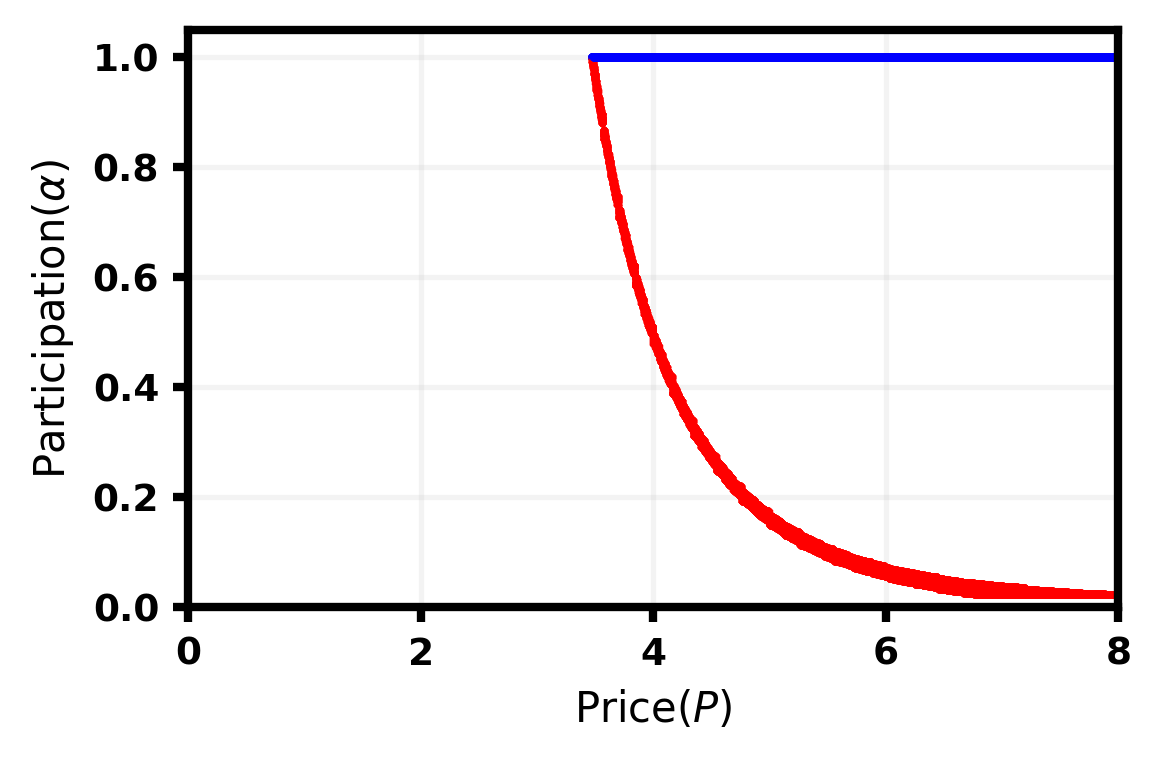}}\hfill
  \subfloat[][$C = 0.1$]{\includegraphics[width=.3\textwidth]{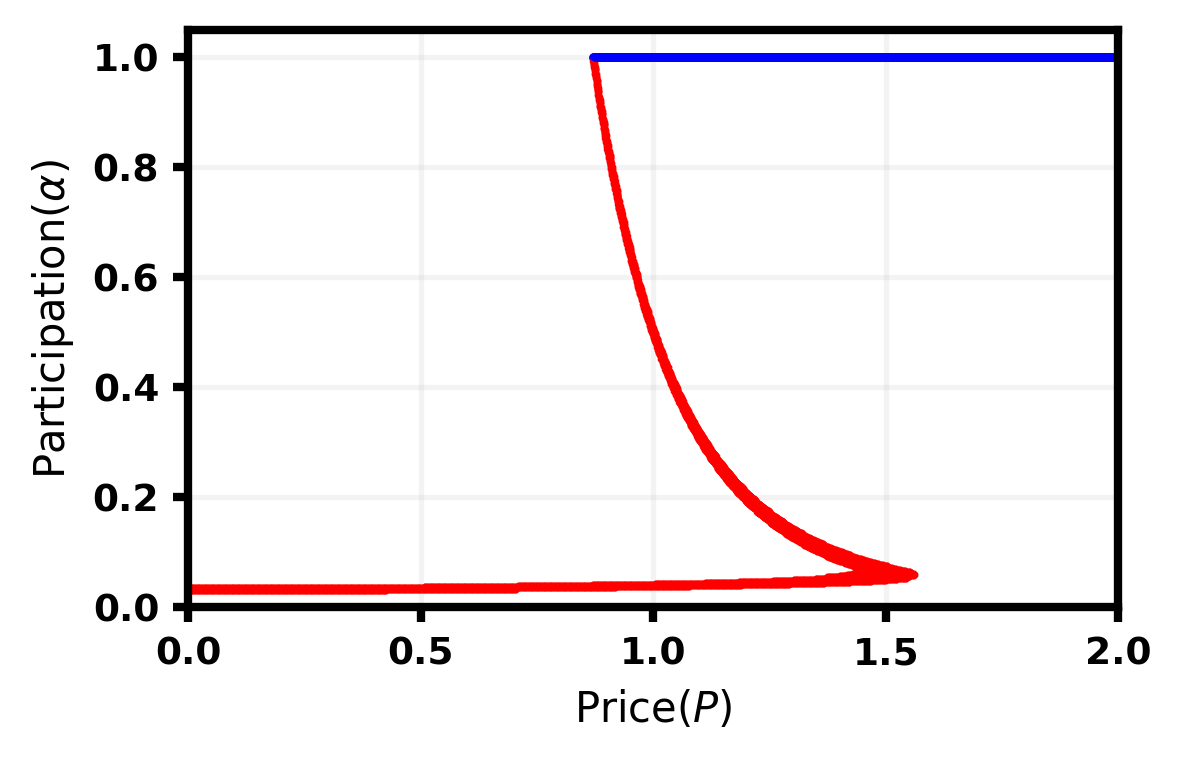}}\hfill
  \subfloat[][$C = 0.2$]{\includegraphics[width=.3\textwidth]{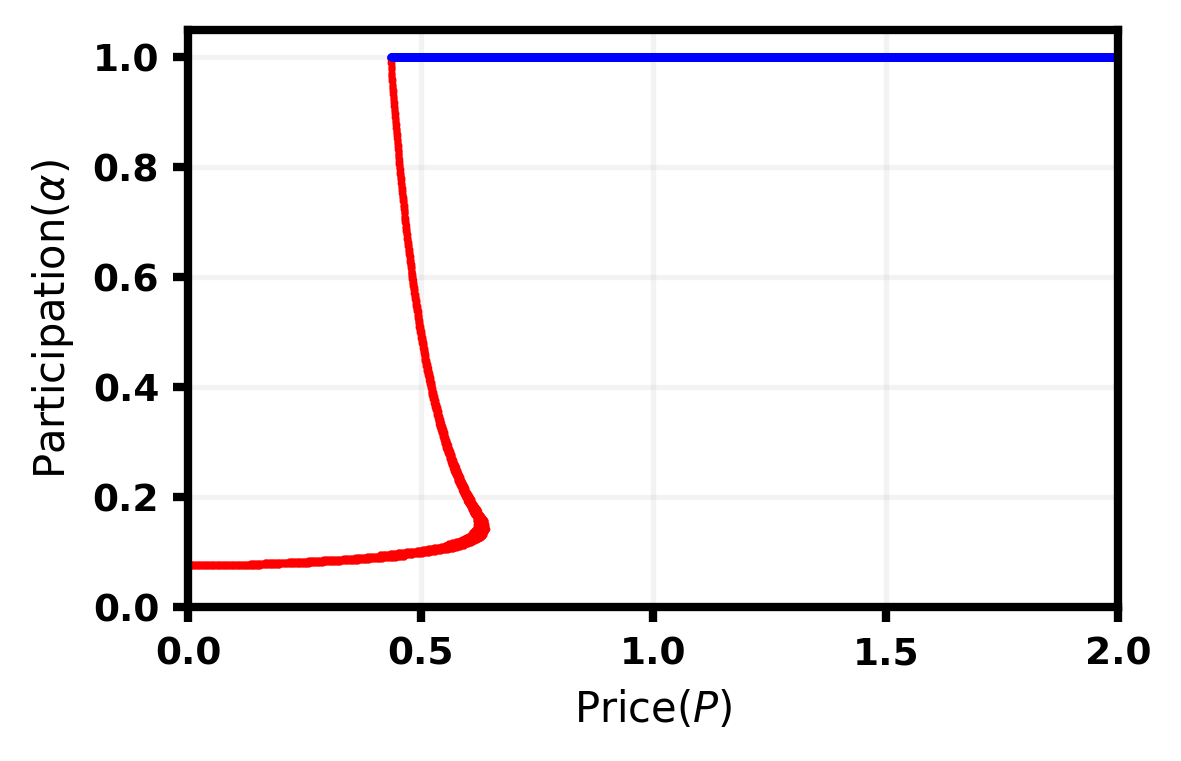}}\par
  \raisebox{20pt}{\parbox[b]{.11\textwidth}{}}%
  \subfloat[][$C = 0.5$]{\includegraphics[width=.3\textwidth]{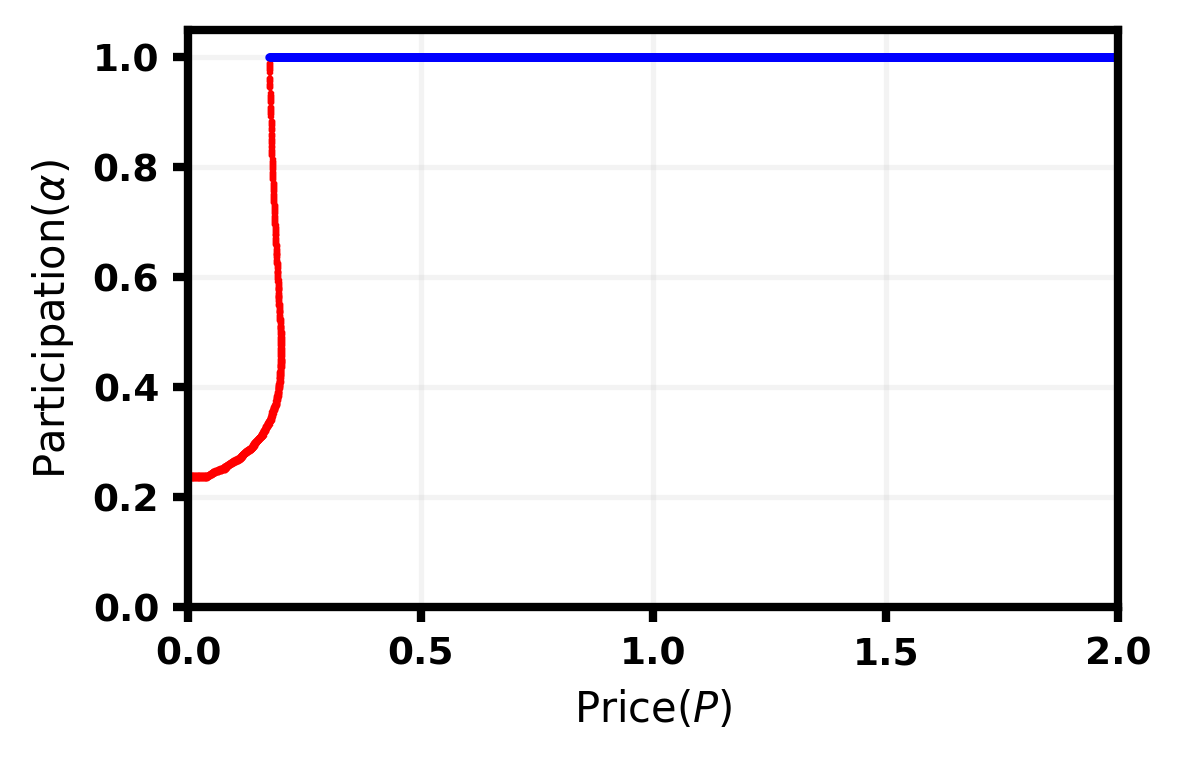}}\hfill
  \subfloat[][$C = 1.0$]{\includegraphics[width=.3\textwidth]{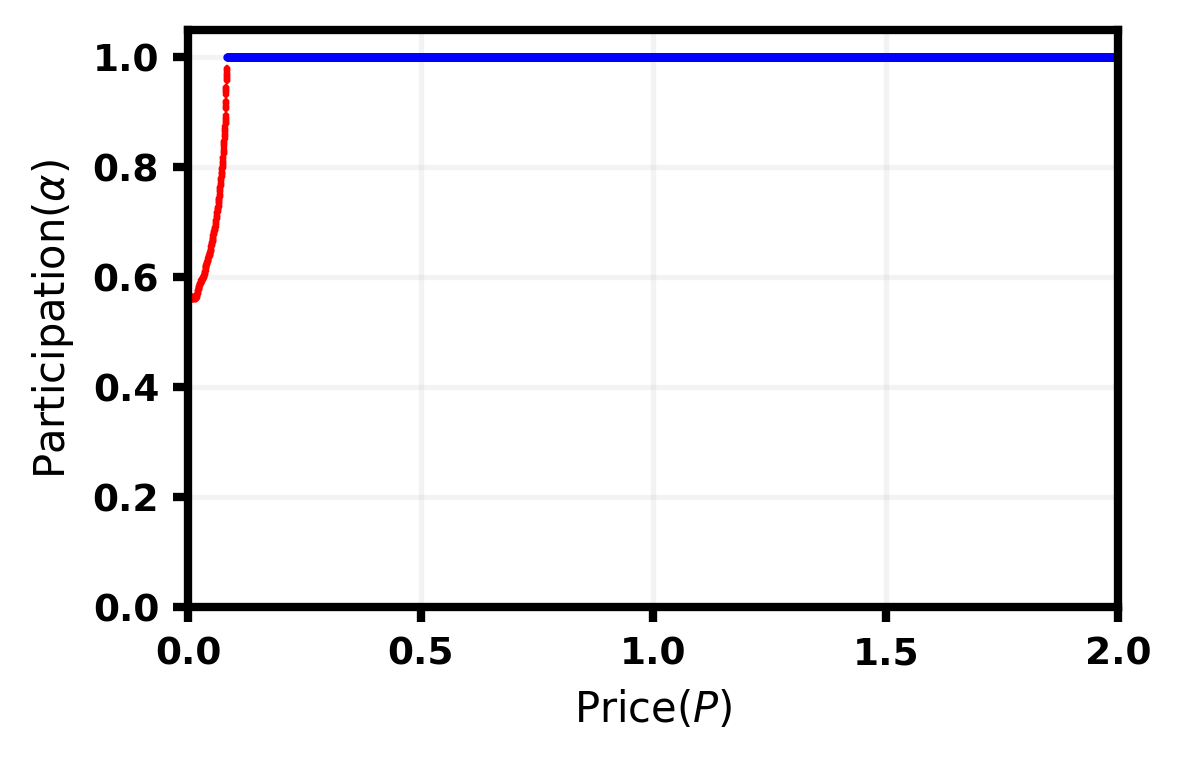}}\hfill
  \subfloat[][$C = 2.0$]{\includegraphics[width=.3\textwidth]{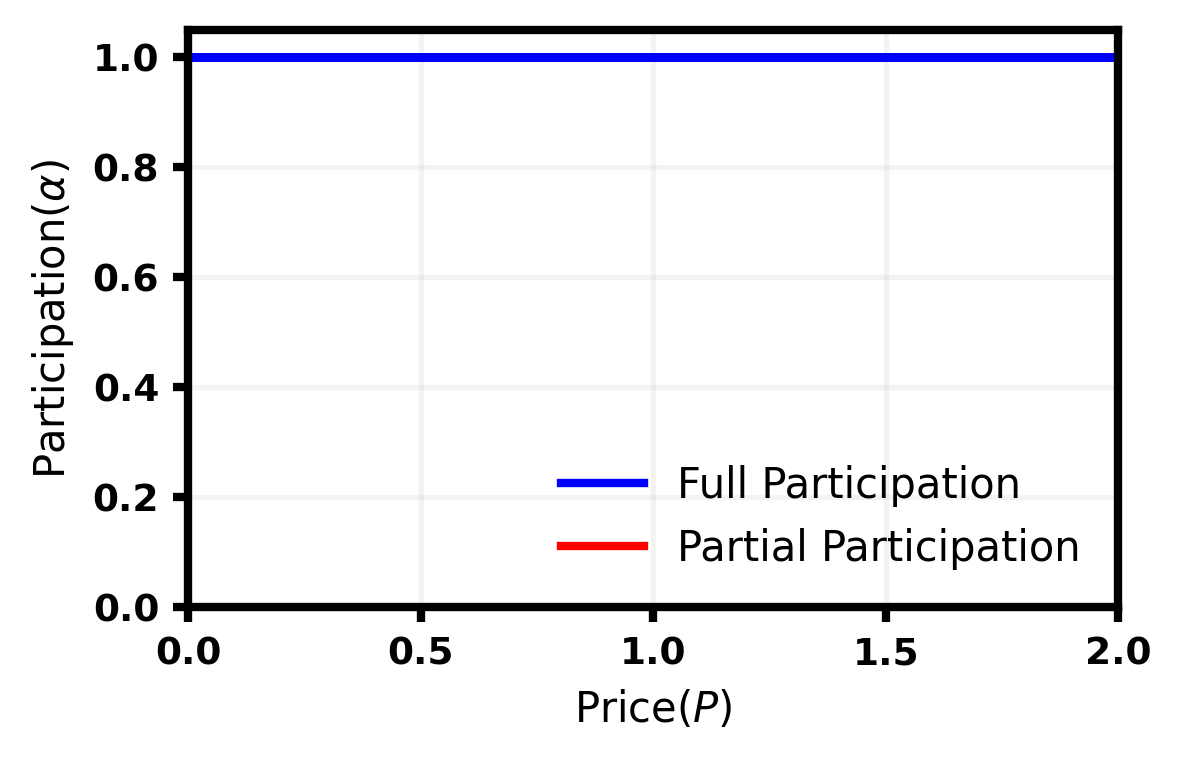}}\par
  \caption{Plots for the case with a diminishing returns quality function ($s = 0.2$) under the uniform distribution of user privacy valuations. Error tolerance $ = 2 \times 10^{-3}$. }
  \label{fig:uniform_privacy_dim1}
\end{figure}

\begin{figure}[!ht]
  \centering
  \raisebox{20pt}{\parbox[b]{.11\textwidth}{}}%
  \subfloat[][$C = 0.001$]{\includegraphics[width=.3\textwidth]{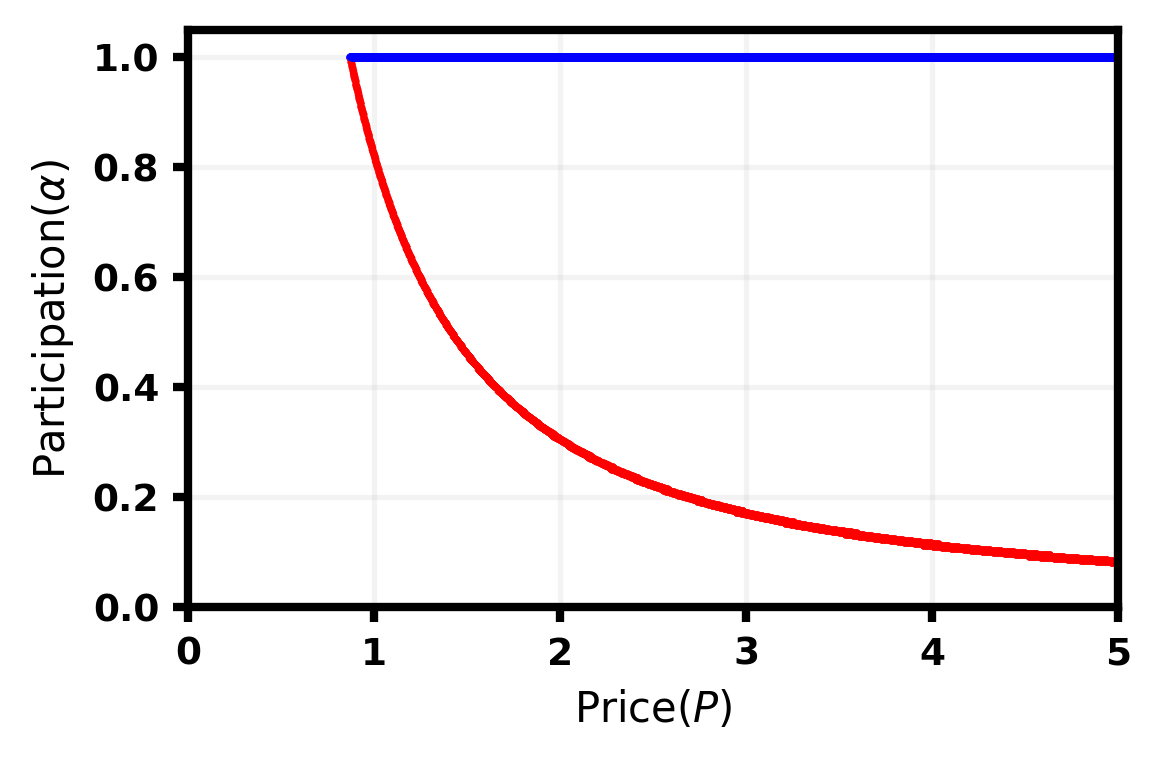}}\hfill
  \subfloat[][$C = 0.0025$]{\includegraphics[width=.3\textwidth]{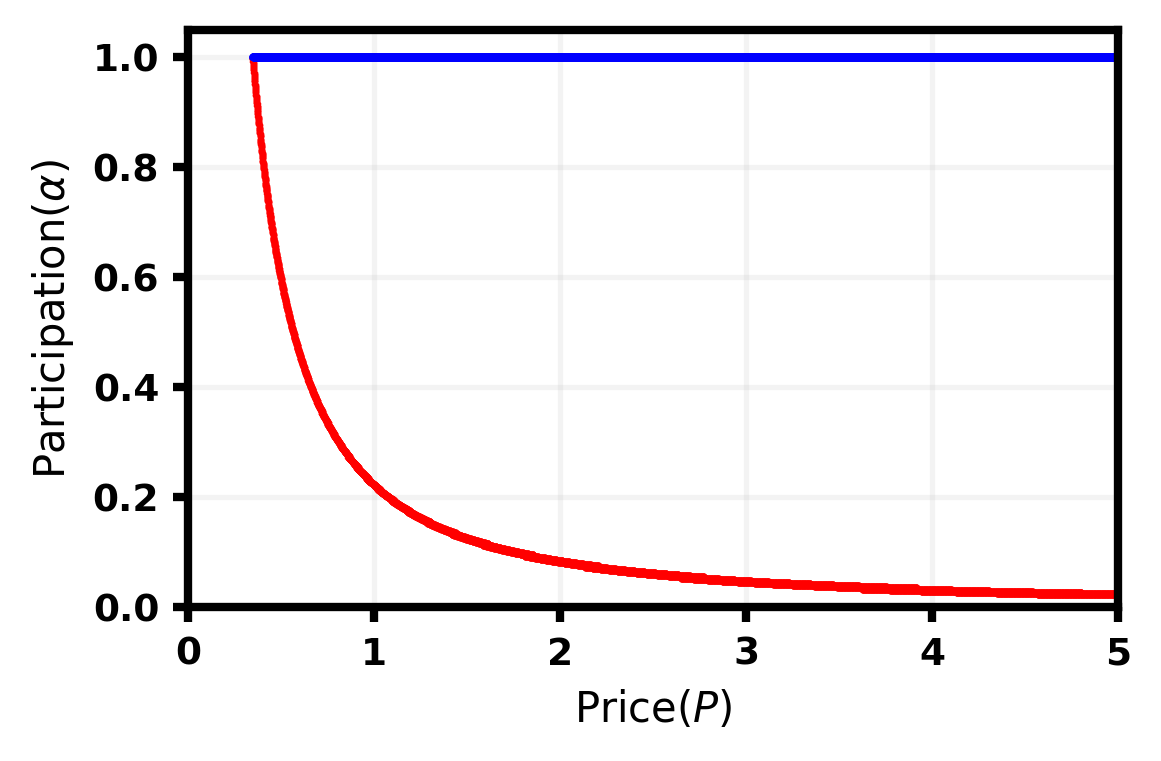}}\hfill
  \subfloat[][$C = 0.005$]{\includegraphics[width=.3\textwidth]{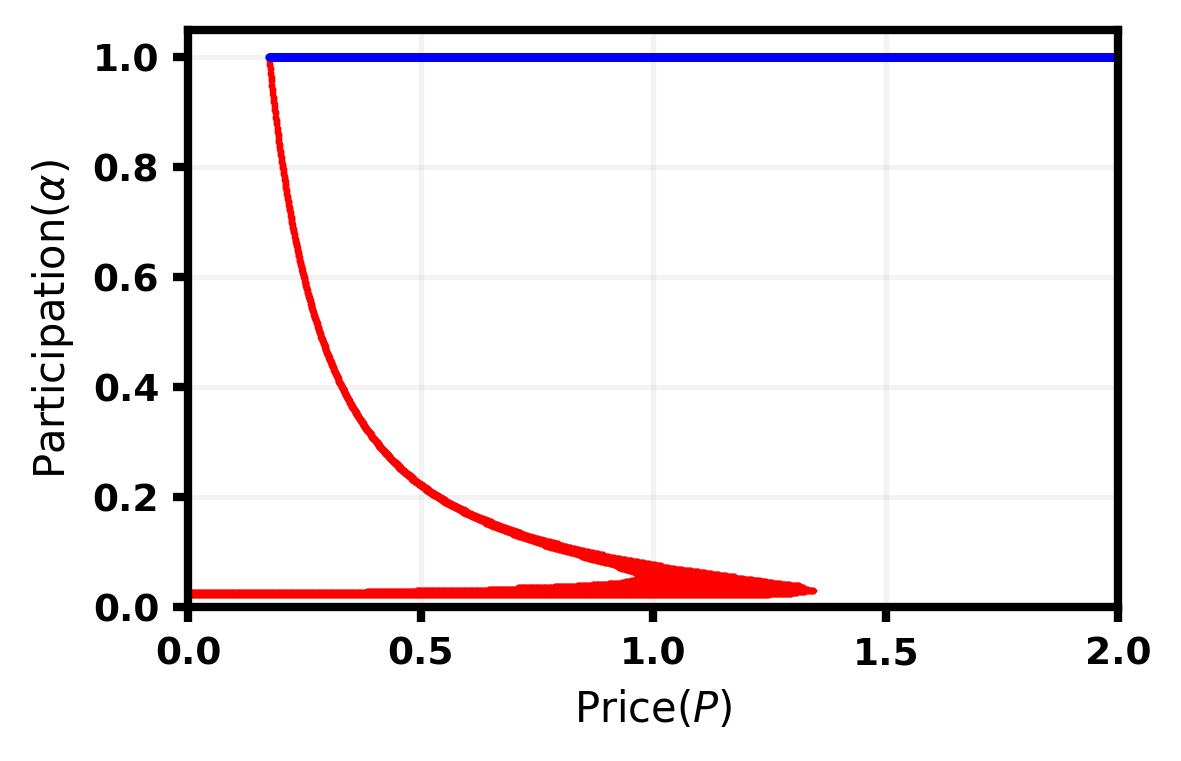}}\par
  \raisebox{20pt}{\parbox[b]{.11\textwidth}{}}%
  \subfloat[][$C = 0.0075$]{\includegraphics[width=.3\textwidth]{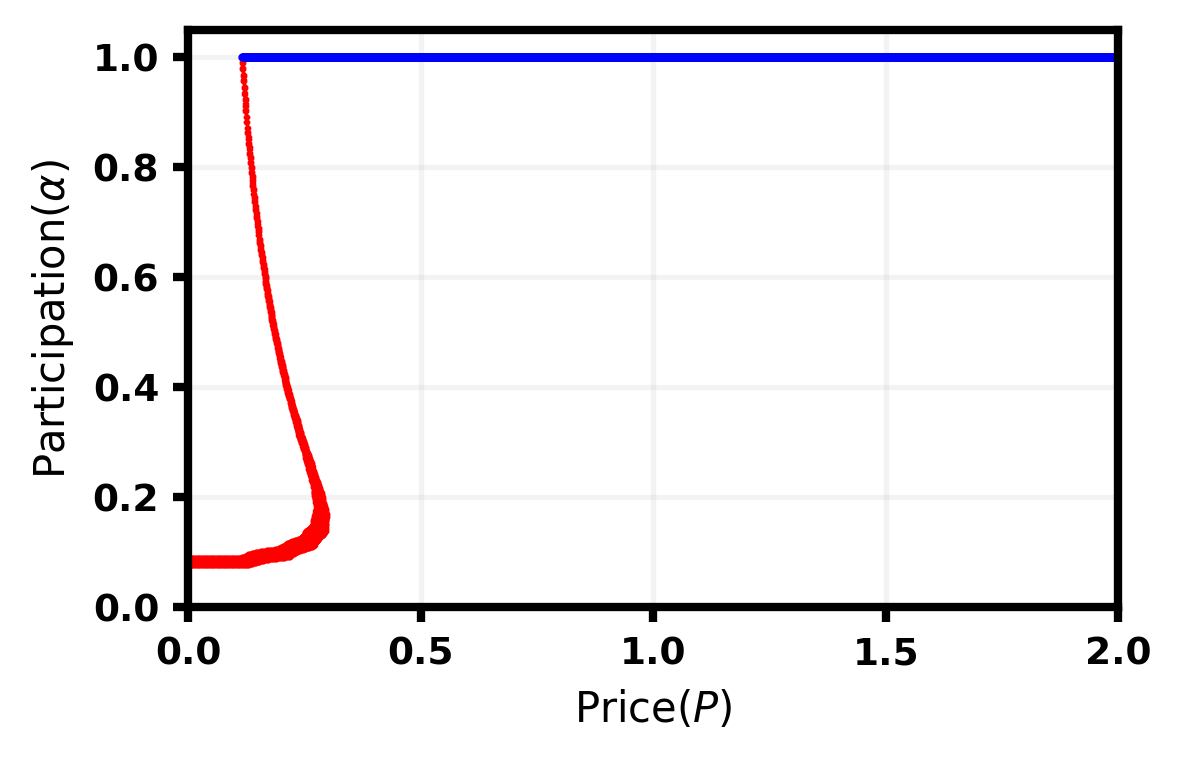}}\hfill
  \subfloat[][$C = 0.01$]{\includegraphics[width=.3\textwidth]{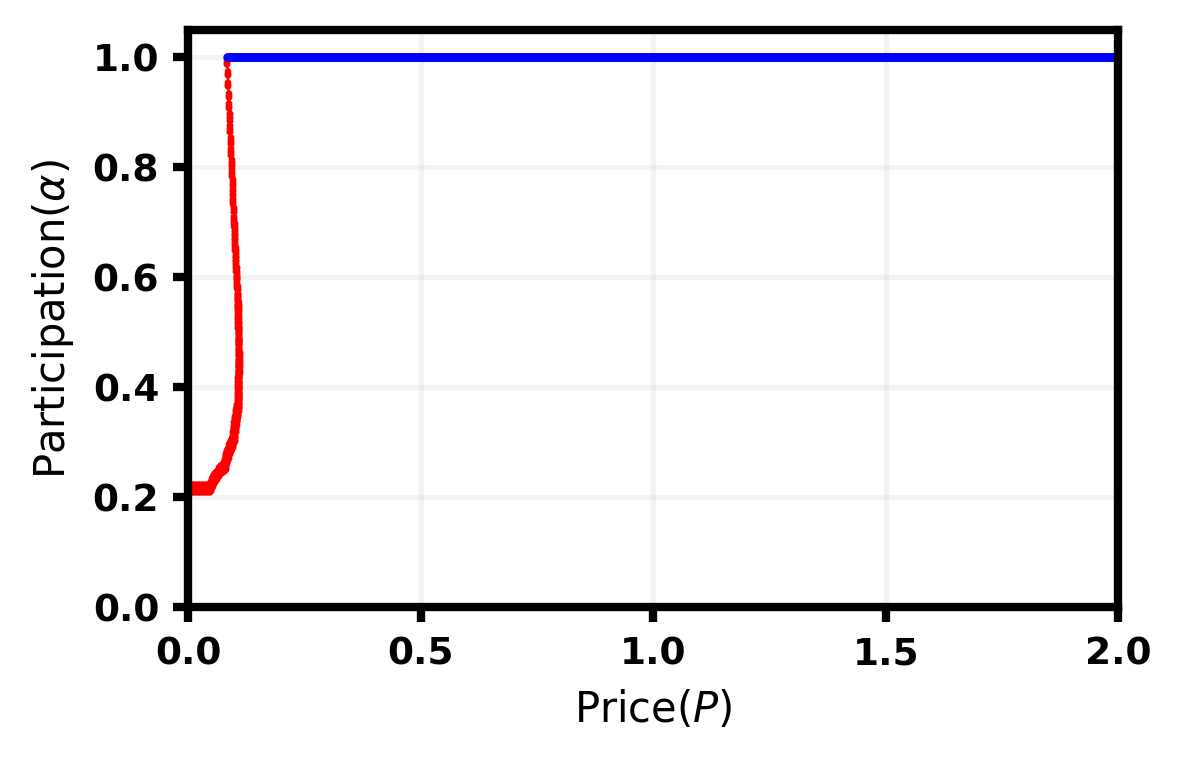}}\hfill
  \subfloat[][$C = 0.1$]{\includegraphics[width=.3\textwidth]{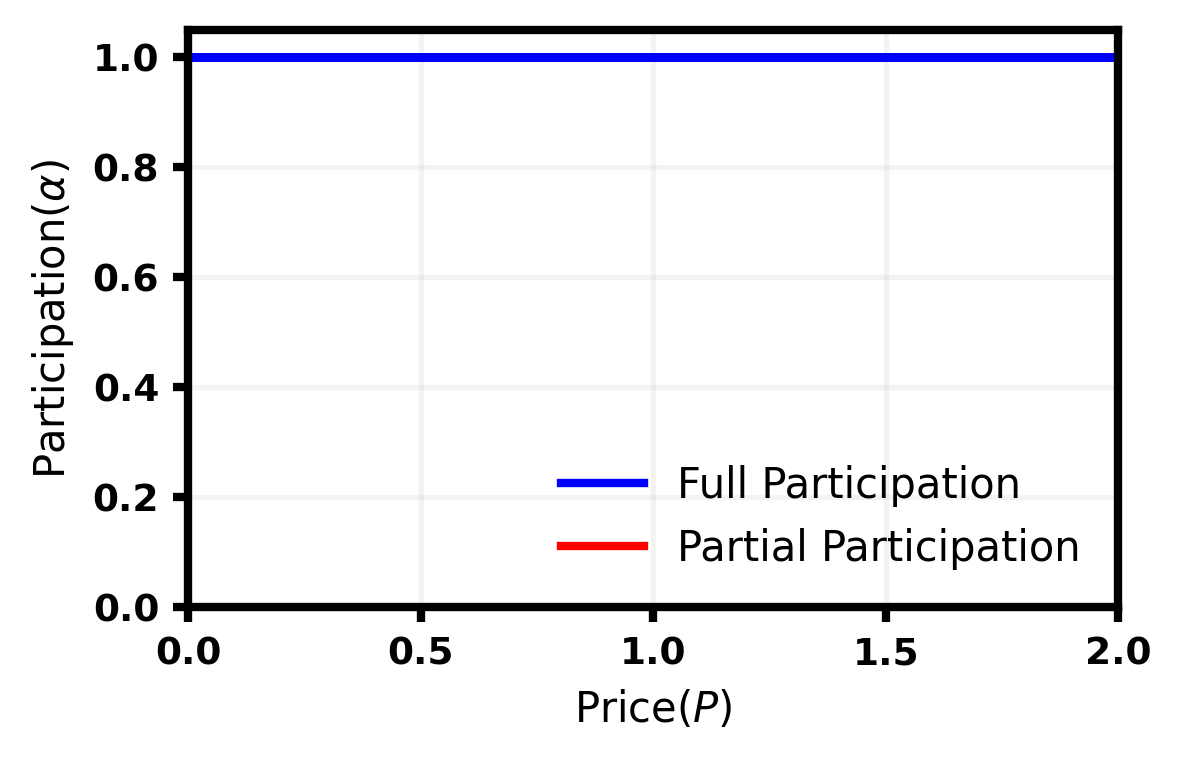}}\par
  \caption{Plots for the case with a diminishing returns quality function ($s = 0.7$) under the uniform distribution of user privacy valuations. Error tolerance $= 2 \times 10^{-3}$.}
  \label{fig:uniform_privacy_dim2}
\end{figure}
In Figures \ref{fig:uniform_privacy_dim1} and \ref{fig:uniform_privacy_dim2}, we observe that at very high values of $C$, we have full participation equilibria for all price $P > 0$. As we decrease $C$ gradually, we see the emergence of partial participation equilibria at smaller prices, where $\alpha$ increases with price, similarly to the theory of the constant case. As we decrease $C$ further, there is a new and interesting transition regime where at smaller prices, there are two types of partial participation equilibria: one where $\alpha$ increases with price $P$ and the other where $\alpha$ decreases with $P$. Finally, as $C$ becomes really small, the equilibrium structure is similar to the linear case low benefit regime: there is a price threshold below which there are no non-trivial equilibria. Above the threshold, we have two types of equilibria: one with full participation and the other with partial participation where $\alpha$ decreases with price. 

\paragraph{S-shaped benefit function} We study S-shaped functions with simple functional form $Q(\alpha) = C G(\alpha)$ where $C$ is a scaling factor and $G(\alpha)$ is a \emph{regularized incomplete beta function} given by
$G(\alpha) = I_{\alpha}(a,b) = \frac{B(\alpha; a,b)}{B(a,b)}$.
$a$ and $b$ are the parameters of the \textit{beta function}. We vary $a$ and $b$ to simulate varying degrees of asymmetry in the S-shaped benefit function, as illustrated in Figure~\ref{fig:beta_functions}. 
\begin{figure}[!ht]
  \centering
  \raisebox{20pt}{\parbox[b]{.11\textwidth}{}}%
  \subfloat[][$Beta(5, 5)$]{\includegraphics[width=.3\textwidth]{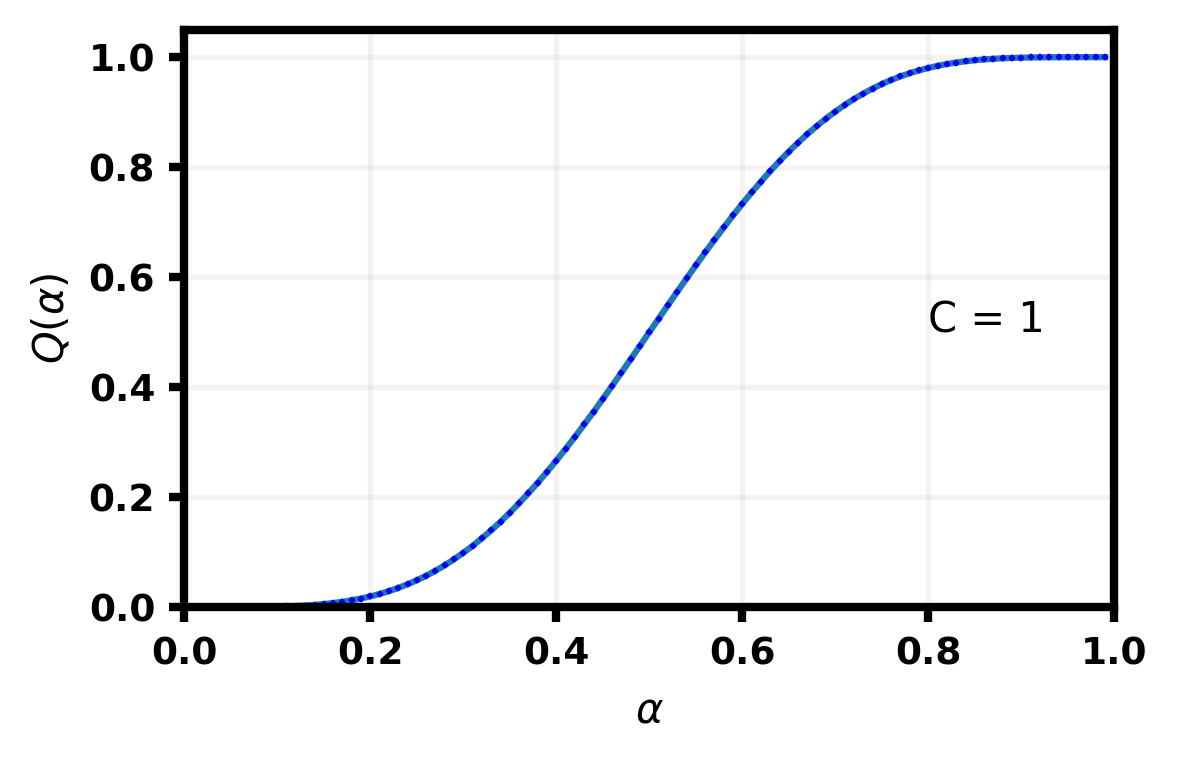}}\hfill
  \subfloat[][$Beta(5, 10)$]{\includegraphics[width=.3\textwidth]{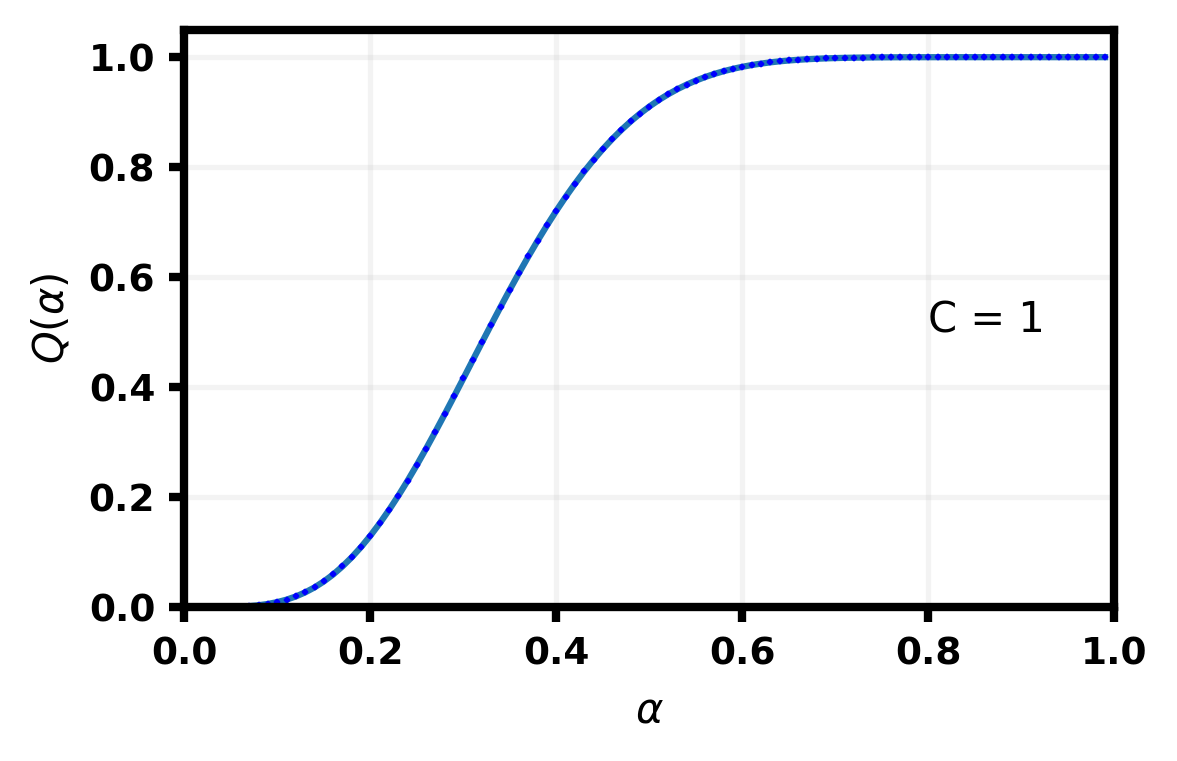}}\hfill
  \subfloat[][$Beta(10, 5)$]{\includegraphics[width=.3\textwidth]{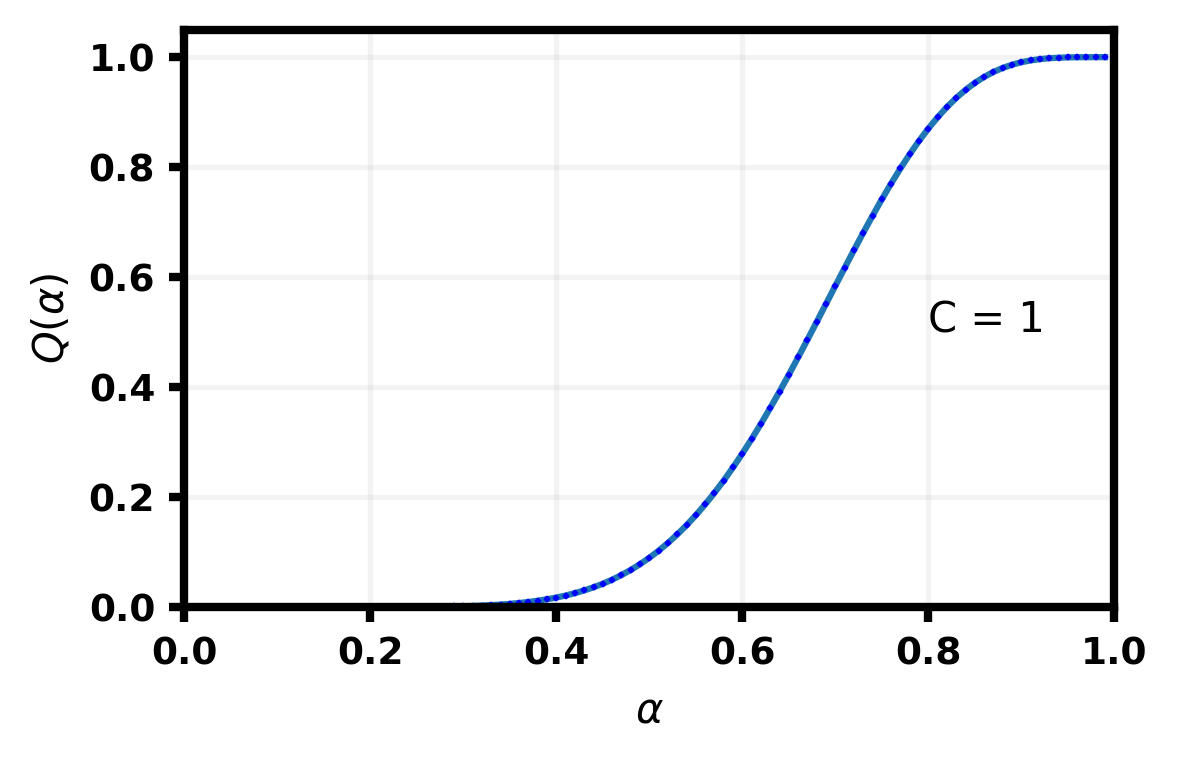}}\par
  \caption{S-shaped benefit function for varying parameters $a,b$}
  \label{fig:beta_functions}
\end{figure}

Then, Figures \ref{fig:uniform_privacy_s1}, \ref{fig:uniform_privacy_s2} and \ref{fig:uniform_privacy_s3} capture the market equilibria for the different s-shaped benefit functions shown earlier. We observe consistent trends across all the variants of the benefit function. At lower values of $C$, there exists a price threshold below which there are no non-trivial market equilibria. Above the threshold, there are both full and partial participation equilibria. There may be one or more types of partial participation equilibria---some which are increasing in $P$ and others which are decreasing in $P$. However, as $C$ increases, the increasing partial participation equilibria vanish as they reach a point where all users participate; those which remain decrease in $P$. 

\begin{figure}[!ht]
  \centering
  \raisebox{20pt}{\parbox[b]{.11\textwidth}{}}%
  \subfloat[][$C = 2$]{\includegraphics[width=.3\textwidth]{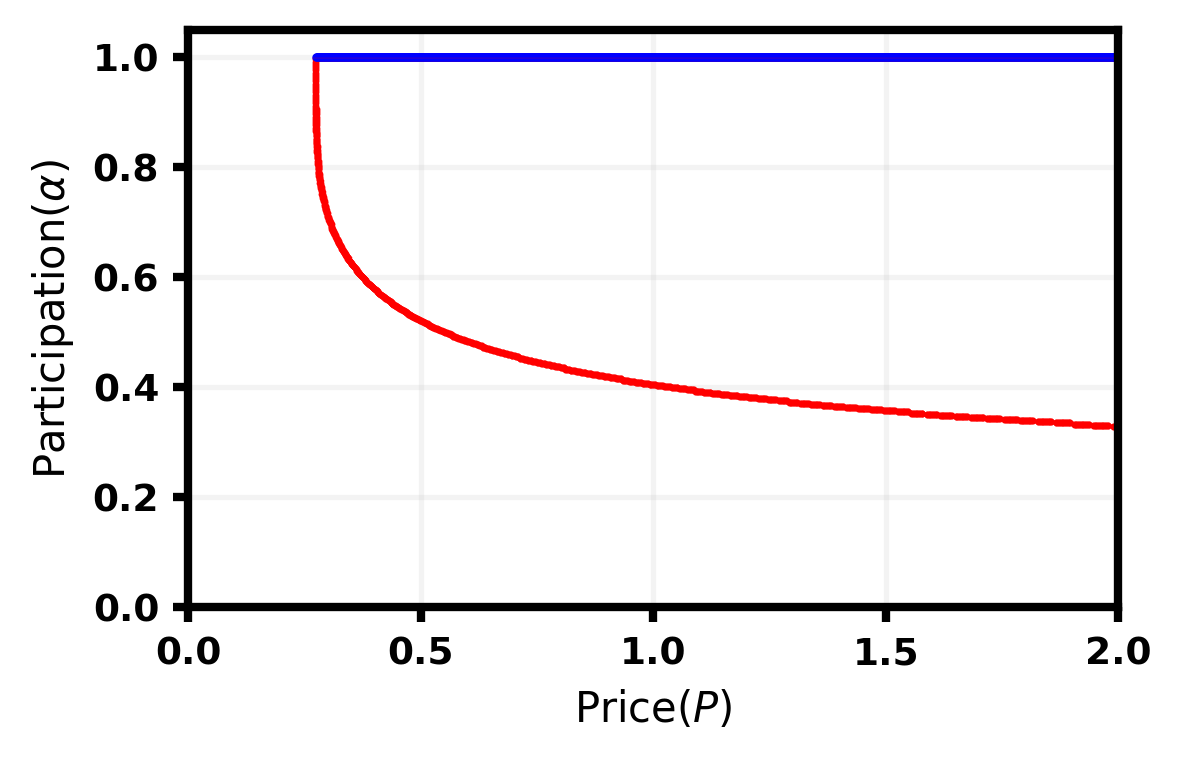}}\hfill
  \subfloat[][$C = 5$]{\includegraphics[width=.3\textwidth]{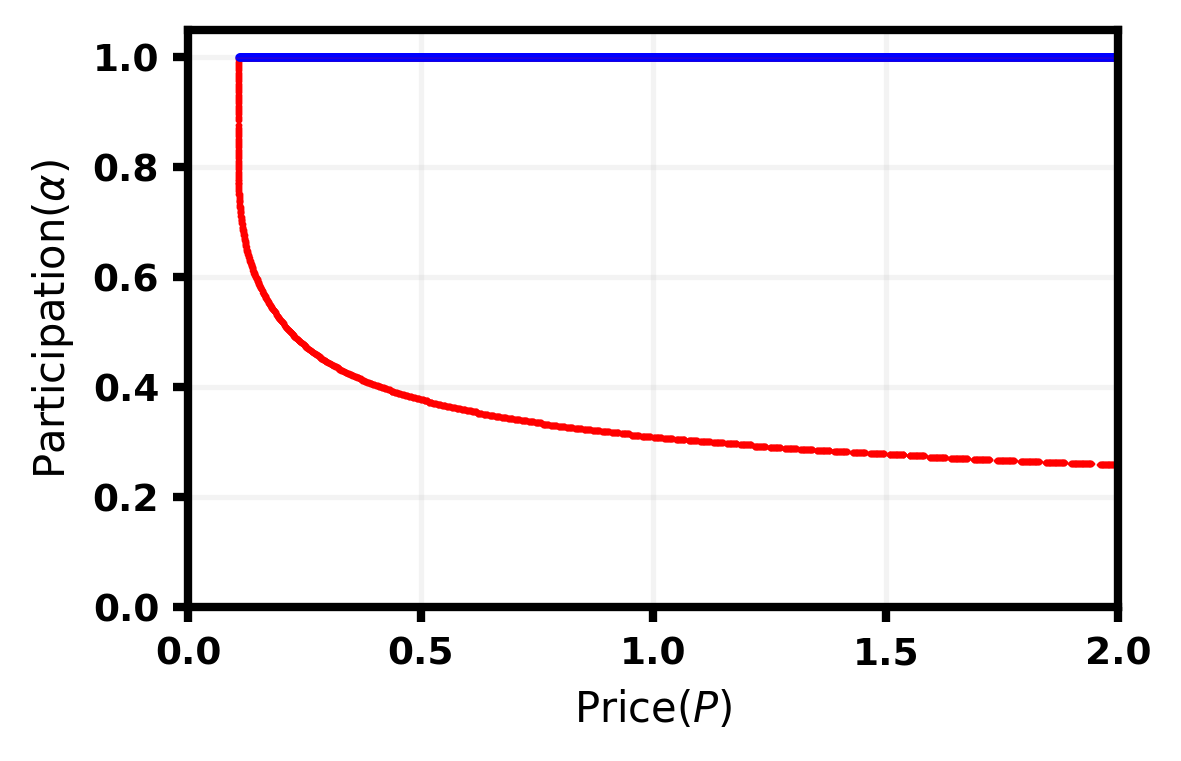}}\hfill%\par
  \raisebox{20pt}{\parbox[b]{.11\textwidth}{}}%
  \subfloat[][$C = 10$]{\includegraphics[width=.3\textwidth]{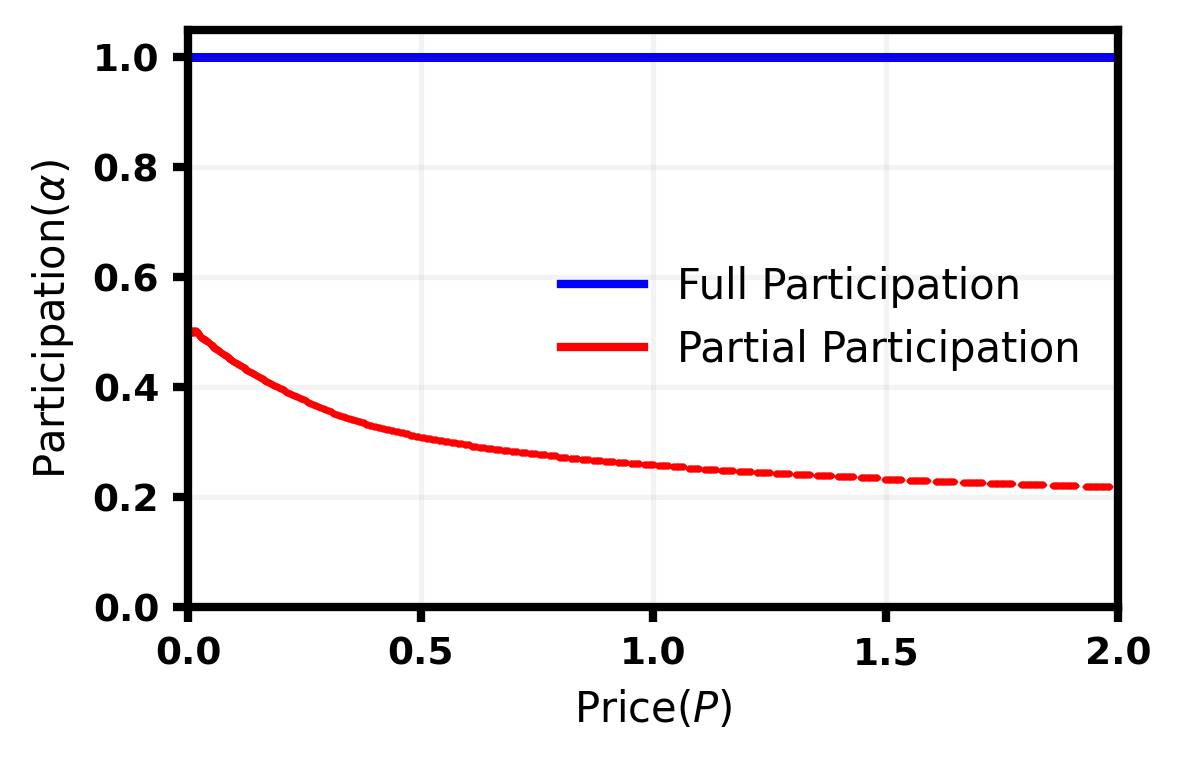}}\par
  \caption{Equilibria under uniformly distributed privacy valuations and a s-shaped benefit function with $(a, b) = (5, 5)$, enumerated for different values of $C$. Error tolerance $= 3 \times 10^{-3}$.}
  \label{fig:uniform_privacy_s1}
\end{figure}

\begin{figure}[!ht]
  \centering
  \raisebox{20pt}{\parbox[b]{.11\textwidth}{}}%
  \subfloat[][$C = 2$]{\includegraphics[width=.24\textwidth]{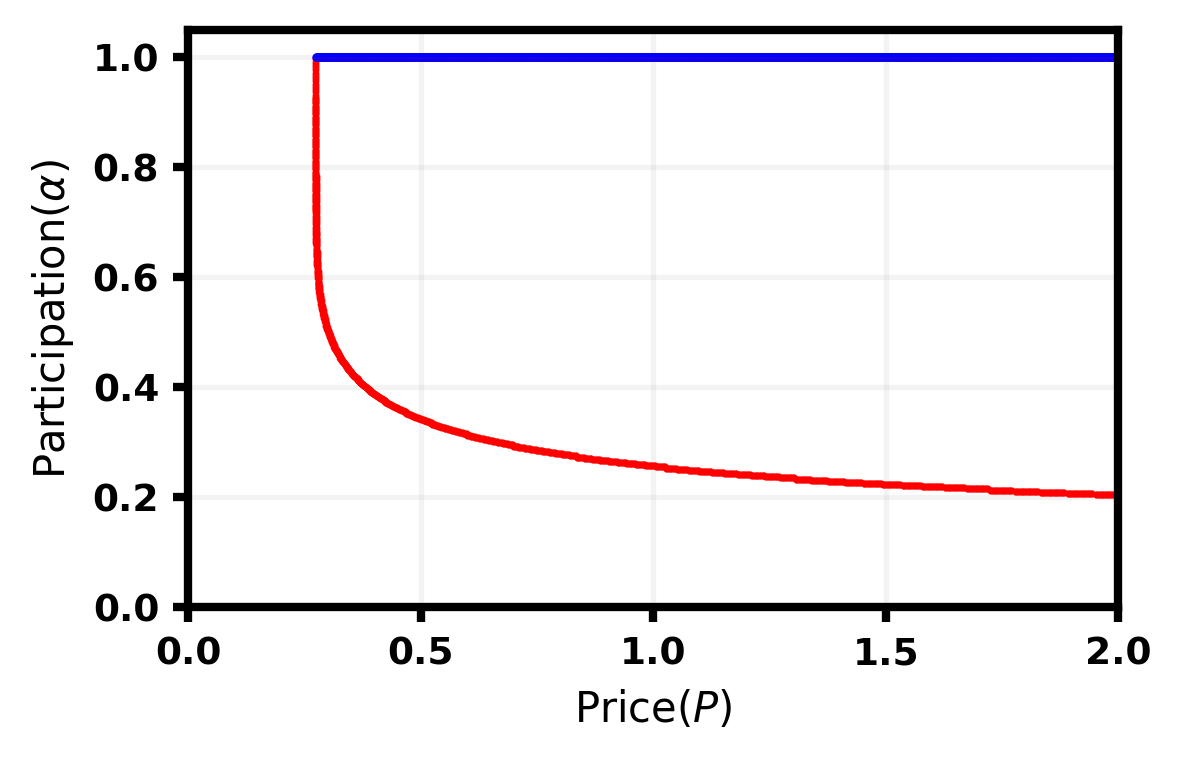}}\hfill
  \subfloat[][$C = 5$]{\includegraphics[width=.24\textwidth]{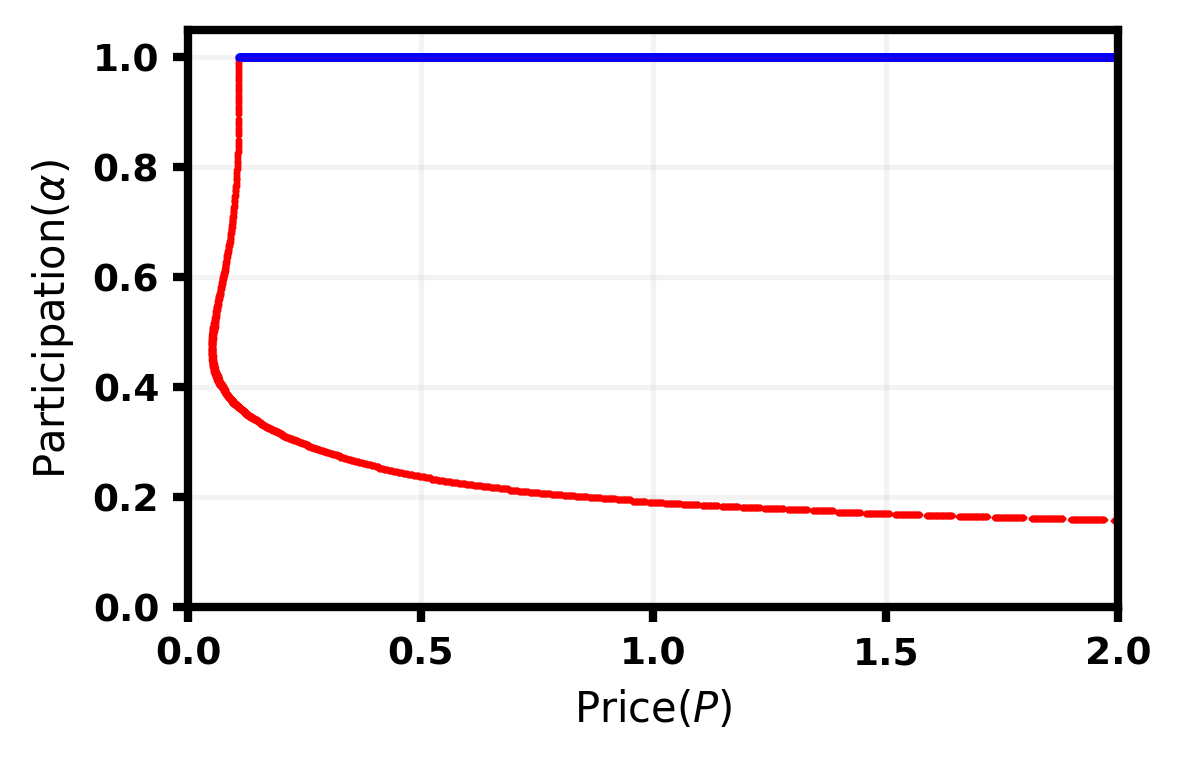}}\hfill%\par
  \raisebox{20pt}{\parbox[b]{.11\textwidth}{}}%
  \subfloat[][$C = 7$]{\includegraphics[width=.24\textwidth]{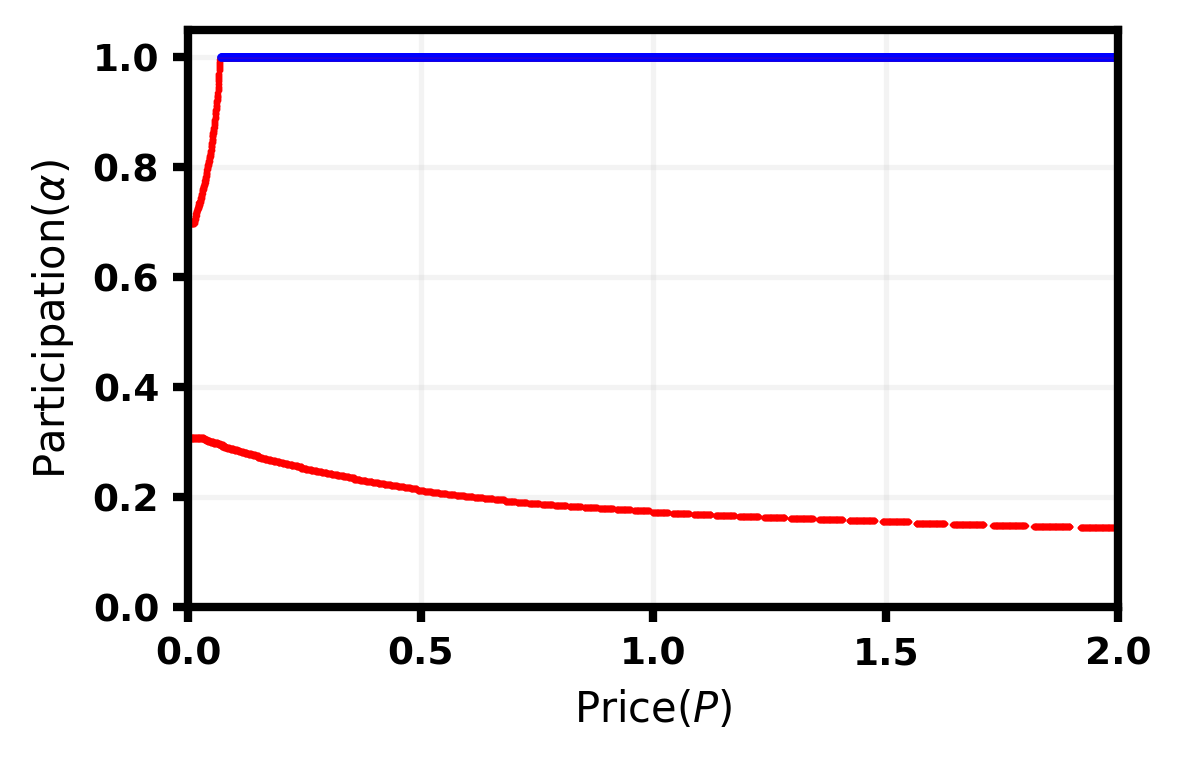}}\hfill
  \subfloat[][$C = 10$]{\includegraphics[width=.24\textwidth]{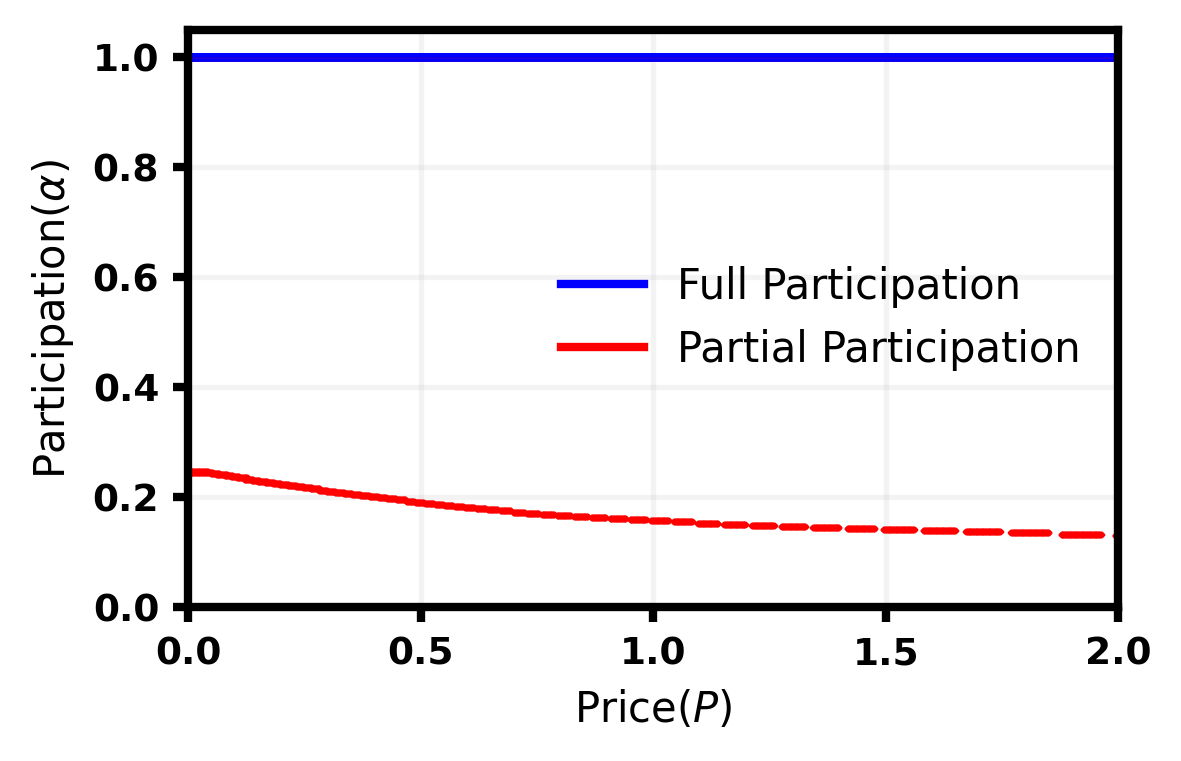}}\par
  \caption{Equilibria under uniformly distributed privacy valuations and a s-shaped benefit function with $(a, b) = (5, 10)$, enumerated for different values of $C$. Error tolerance $= 3 \times 10^{-3}$.}
  \label{fig:uniform_privacy_s2}
\end{figure}

\begin{figure}[!ht]
  \centering
  \raisebox{20pt}{\parbox[b]{.11\textwidth}{}}%
  \subfloat[][$C = 5$]{\includegraphics[width=.3\textwidth]{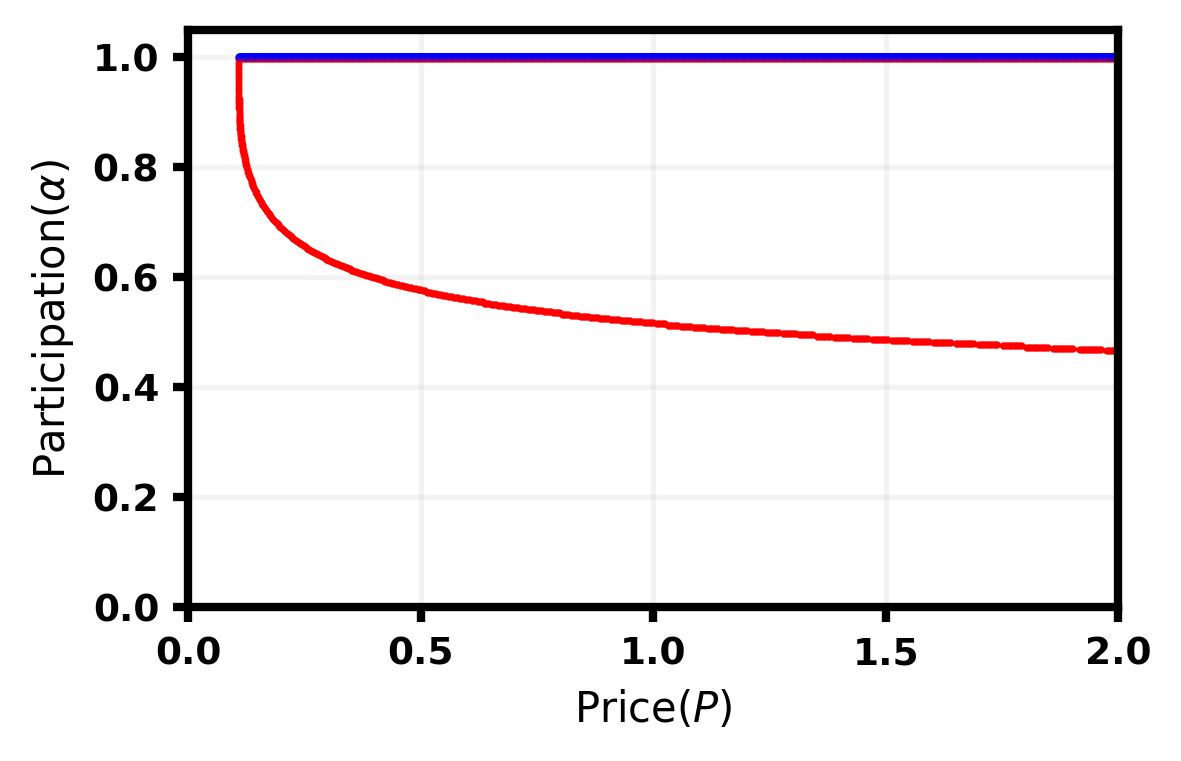}}\hfill%\par
  \raisebox{20pt}{\parbox[b]{.11\textwidth}{}}%
  \subfloat[][$C = 10$]{\includegraphics[width=.3\textwidth]{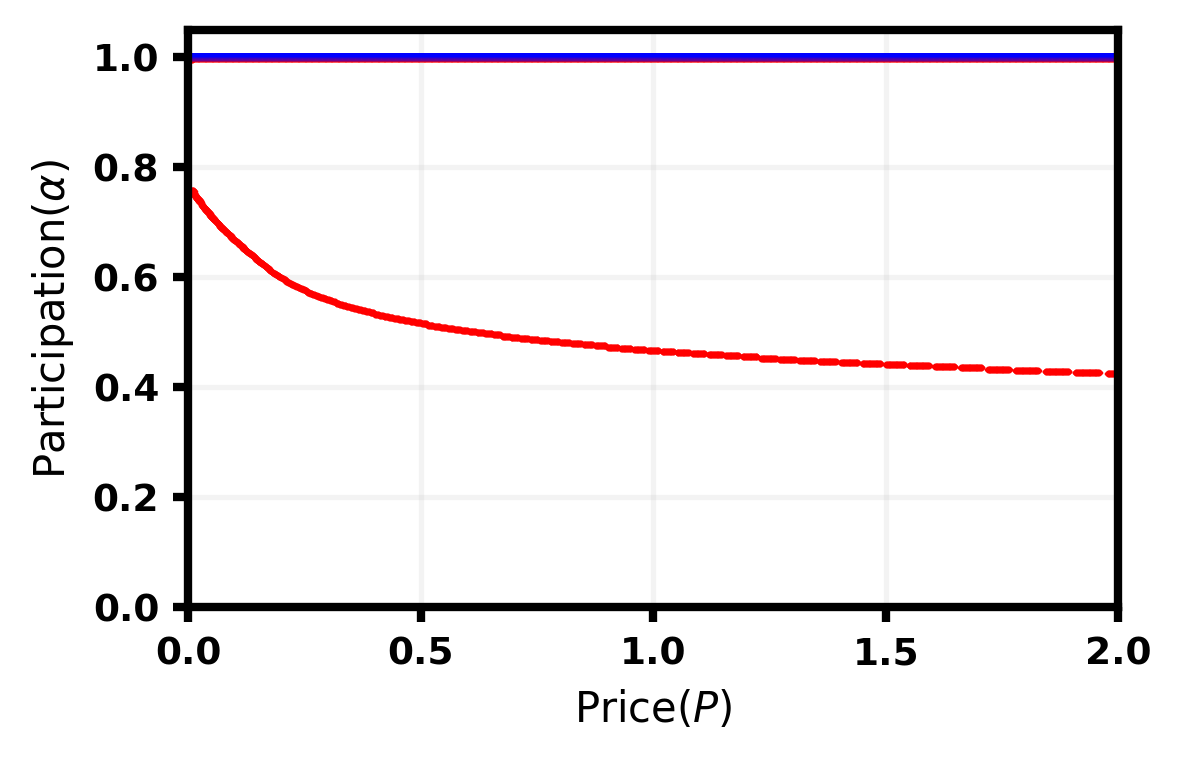}}\hfill
  \subfloat[][$C = 20$]{\includegraphics[width=.3\textwidth]{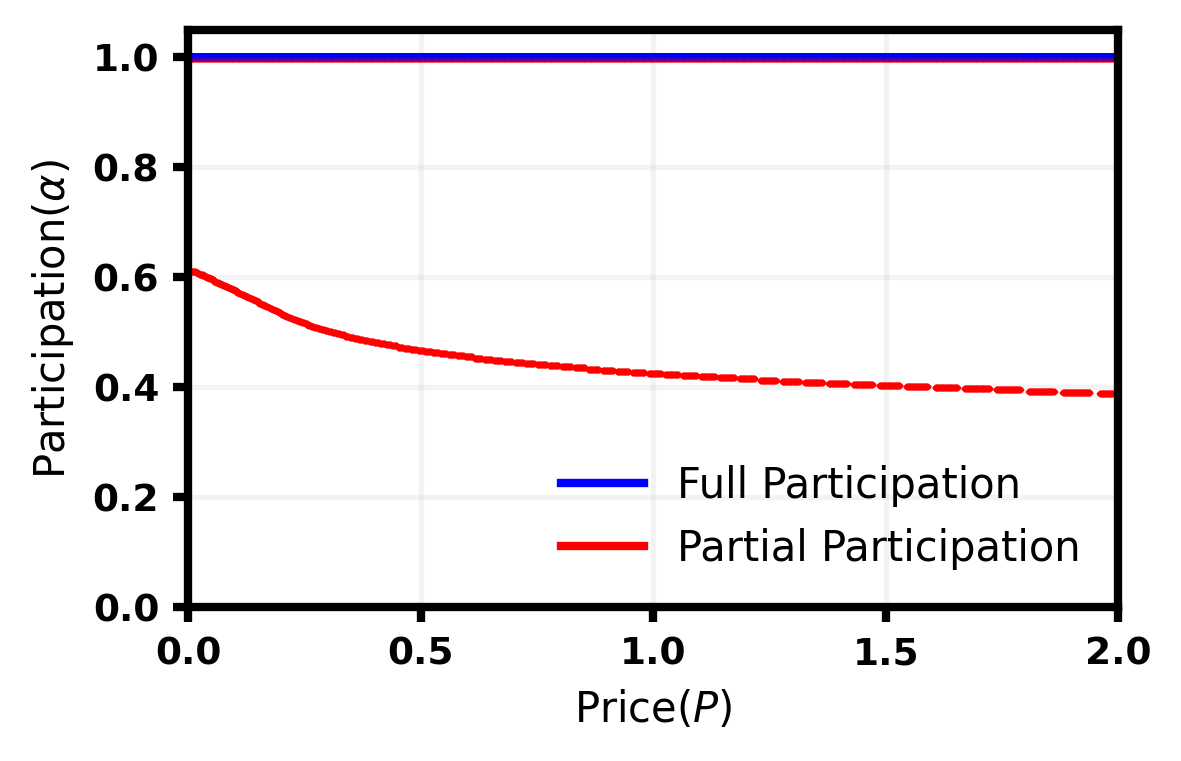}}\par
  \caption{Equilibria under uniformly distributed privacy valuations and a s-shaped benefit function with $(a, b) = (10, 5)$, enumerated for different values of $C$. Error tolerance $= 6 \times 10^{-3}$.}
  \label{fig:uniform_privacy_s3}
\end{figure}

% Bibliography
\bibliographystyle{plain}
\bibliography{mybib}

% Appendix
\appendix

\section{Proofs for the constant case}\label{app:constant}

These results provide a complete characterization of market equilibria for full and partial participation of users. While developing the proof, we will analyze the full participation and the partial participation cases separately. 

Before stating the stronger version of our theorems and starting the proof, we will first note a useful result on the buyer side of the market, which will be useful in simplifying our equilibrium analysis, reducing the analysis of buyer best responses to a single parameter $k^*$.

\begin{claim}\label{clm:N_monotone}
    For a given $(P, \alpha)$-tuple at equilibrium, $\buydec$'s are always non-decreasing in $k$. In turn, there exists a threshold $k^*$ such that
    \begin{align*}        
      \buydec < \alpha \numusers, \quad \forall~k \leq k^* 
    \end{align*}
    and
    \begin{align*}
    \buydec = \alpha \numusers, \quad \forall~ k > k^*.
    \end{align*}
\end{claim}

\begin{proof}
This can be trivially verified from \ref{opt:platform}. Recall that $\buydec$ is given by:
\[  
\buydec = \min \left(\alpha \numusers, \frac{\budget}{P} \right) \quad \forall ~k \in [\numbuyers]. 
\] 
Since $P$ and $\alpha$ are fixed, $\buydec$ is non-decreasing in $\budget$. The claim now follows directly from the budget monotonicity in Assumption \ref{ass:budget_order}.
\end{proof}

This result is quite intuitive: each buyer wants to buy as much as possible, so the larger their budget, the more they buy. When $k^* = 0$, all buyers purchase the entire available data and no buyer exhausts their budget; when $k^* = \numbuyers$, all buyers use up their entire budget to purchase a partial amount of the total number of available data points. In the rest of the proof, when an equilibrium induces a threshold $k^*$ on the buyers, we talk about an \emph{equilibrium at $k^*$}.

\subsection{Part 1: A Full Equilibrium Characterization}\label{app:constant_eq_chart}
We first start by ruling out corner cases, noting that no partial participation equilibrium can exist at $k^* = \numbuyers$ and no full participation equilibrium can exist at $k^* = 0$ under our condition that $\numbuyers > \Q \geq \frac{B_{\leq \numbuyers}}{B_{\numbuyers}}$; we characterize these cases separately in Theorems \ref{thm:eq_exist_highQ} and \ref{thm:eq_exist_lowQ} and Appendix \ref{app:constant_thm_simple}.

\begin{claim}
Suppose $\Q B_{\numbuyers} \geq B_{\leq \numbuyers}$. Then, there is no equilibrium with partial user participation $(0 < \alpha < 1)$ at $k^* = \numbuyers$.
\end{claim}

\begin{proof}
We will prove by contradiction. Suppose, there exists a partial participation equilibrium $(0 < \alpha < 1)$ at $k^* = \numbuyers$. In that case, $N_k = \frac{B_k}{P} (< \alpha \numusers)$ for all $k \in [\numbuyers]$ (by the definition of $k^*$ in Claim \ref{clm:N_monotone}). Since $0 < \alpha < 1$ (by assumption), we must have the following equality (recall, $\alpha = \min\left(1, \frac{\alpha \Q \numusers}{\sum_k N_k}\right)$): 
\[
       \Q \numusers = \sum_{k=1}^{\numbuyers} N_k = \frac{B_{\leq \numbuyers}}{P} \quad \text{which implies} \quad P = \frac{B_{\leq \numbuyers}}{\Q \numusers}.
\]
This price $P$ must satisfy $\frac{B_{\numbuyers}}{P} < \alpha \numusers$ which implies $\alpha > \frac{Q B_{\numbuyers}}{B_{\leq \numbuyers}} \geq 1$. Thus we have a contradiction because $\alpha < 1$ by assumption. This concludes the proof. 
\end{proof}

\begin{claim}
Suppose $\numbuyers > Q$. Then, there is no equilibrium with full user participation $(\alpha = 1)$ at $k^* = 0$. 
\end{claim}

\begin{proof}
Again, we prove by contradiction. Suppose, there exists an equilibrium at $k^* = 0$ with $\alpha = 1$. By the definition of $k^*$ (refer Claim \ref{clm:N_monotone}), we derive $\sum_{k=1}^{\numbuyers} N_k = \numbuyers \alpha \numusers = \numbuyers\numusers$. Additionally, since $\alpha = 1$, from the expression of $\alpha$ in \ref{opt:platform}, we obtain: 
\[
        \Q \numusers \geq \sum_{k=1}^{\numbuyers}N_k = \numbuyers \numusers,
\]
which implies $\Q \geq \numbuyers$ (contradiction !). This concludes the proof. 
\end{proof}

In the rest of the section, we first show a more precise version of Theorem~\ref{thm:eq_exist_modQ} which also characterizes equilibria on the buyer side. We will then, in Section \ref{app:constant_thm_simple}, show how Theorem~\ref{thm:eq_existence_full} simplifies into Theorem~\ref{thm:eq_exist_modQ}. 

Let us quickly recall the definition of $\gamma(k)$ from Theorem \ref{thm:eq_exist_modQ}. Additionally, for ease of exposition, we introduce a few more notations: 
\begin{enumerate}
    \item $\gamma(k) = \frac{(\numbuyers - k)B_k + B_{\leq k}}{\Q \numusers}$;
    \item $\xi(k) = \frac{B_k}{\numusers}$;
    \item $P^{*}(k) = \frac{B_{\leq k}}{\numusers \left( \Q - (\numbuyers - k) \right)}$.
\end{enumerate}

\begin{theorem}\label{thm:eq_existence_full}
Suppose $\frac{B_{\leq \numbuyers}}{B_{\numbuyers}} \leq \Q < \numbuyers$. Fix $k^* \in \{0,\ldots,\numbuyers\}$. There exists an equilibrium with buyer threshold $k^*$ at price $P$ with: 
\begin{itemize}
    \item full participation of users $(\alpha = 1)$ if and only if: 
    \[
        k^* > 0, \quad \xi(k^*) < P \leq \xi(k^*+1) \quad \text{and} \quad P \geq P^*(k^*).
    \]
    In this case, the quantity of data bought by each buyer is given by:
    \[
        N_k = \frac{B_k}{P} \quad \forall k \leq k^* \quad \text{and} \quad N_k = \numusers \quad \forall k > k^*. 
    \]
    \item partial participation of users $(0 <\alpha < 1)$ if and only if: 
    \[
     k^* < \numbuyers, \quad \gamma(k^*) < P \leq \gamma(k^*+1), 
    \]
    and additionally, 
    \[
    P < P^*(k^*) \quad \text{when } \Q > \numbuyers - k^*.
    \]
    In this case, the participation rate is given by:
    \[ 
    \alpha = \frac{1}{\numbuyers - k^*}\left[\Q - \frac{B_{\leq k^*}}{P \numusers} \right],
    \]
    and the quantity of data bought by each buyer is given by:
    \[
        N_k = \frac{B_k}{P} \quad \forall~ k \leq k^* \quad \text{and} \quad N_k = \frac{\numusers}{\numbuyers - k^*}\left[\Q - \frac{B_{\leq k^*}}{P\numusers} \right] \quad \forall k > k^*.
    \]
\end{itemize}
\end{theorem}

The proof of the theorem is given in the rest of the section. We divide this proof into two cases, first the partial participation case, and second the full participation case. 

\paragraph{Case 1: $0 < \alpha < 1$ (Partial Participation)} 
In this segment, we derive necessary and sufficient conditions for the existence of equilibria with partial participation of users (i.e., $0 < \alpha < 1$). We start by obtaining a closed-form expression for the participation rate $\alpha$ for an equilibrium at $k^*$ when it is known that $\alpha$ is fractional.

\begin{lemma}\label{lem:alphaless1}
Any equilibrium at $k^* \in \{0,1,2....(\numbuyers -1)\}$ with $0 < \alpha < 1$ satisfies 
\begin{align}
\alpha = \frac{1}{\numbuyers - k^*} \left[ \Q - \frac{B_{\leq k^*}}{P \numusers} \right]. \label{exp:alpha} 
\end{align}
\end{lemma}
\begin{proof}
Since we know that this is an equilibrium at $k^*$, we have:
\begin{align}
  \sum_{k=1}^{\numbuyers} \buydec= \sum_{k=1}^{k^*} \buydec + \sum_{k=k^*+1}^{\numbuyers} \buydec \stackrel{\stepone}{=} \sum_{k=1}^{k^*} \frac{\budget}{P} + \alpha \numusers(\numbuyers-k^*) \stackrel{\steptwo}{=} \frac{B_{\leq k^*}}{P} + \alpha \numusers(\numbuyers-k^*). \label{exp:sum_Nk}
\end{align}
where step~\stepone results from the definition of $k^*$ in Claim \ref{clm:N_monotone} and step~\steptwo results from the definition of $B_{\leq k^*}$ in Defn~\ref{defn:Bleq}. 
Now, recall from \eqref{opt:platform} that $\alpha = \min \left( 1, \frac{\alpha \Q \numusers}{\left( \sum_{k=1}^{\numbuyers} \buydec \right)} \right)$. Since $0 < \alpha < 1$ by assumption, it must be that:
\[
    \alpha = \frac{\alpha \Q \numusers}{\left( \sum_{k=1}^{\numbuyers} \buydec \right)},
\]
and since $\alpha > 0$, this implies:
\begin{align}
\Q \numusers = \sum_{k=1}^{\numbuyers} \buydec. \label{eq:temp}
\end{align}
Combining Equations \eqref{exp:sum_Nk} and \eqref{eq:temp}, we obtain:
\begin{align*}
        &\Q \numusers = \frac{B_{\leq k^*}}{P} + \alpha \numusers(\numbuyers-k^*); 
\end{align*}
and hence, by rearranging terms,
\begin{align*}
\alpha = \frac{1}{\numbuyers - k^*} \left[ \Q - \frac{B_{\leq k^*}}{P \numusers} \right]. 
\end{align*}
\end{proof}

We then show necessary and sufficient conditions for the existence of a partial participation equilibrium at $k^* < \numbuyers$. 
\begin{lemma}\label{lem:price_alphaless1}
    Fix $k^* \in \{0,1,2...(\numbuyers-1)\}$. There \textbf{exists} an equilibrium at $k^*$ with $0 < \alpha < 1$ if and only if: 
    \begin{align}
     \gamma(k^*) < P \leq \gamma(k^*+1). \label{eq:condP1} 
    \end{align}
    Additionally, if $\Q > \numbuyers - k^*$, we also need:
    \begin{align}
    P < P^*(k^*).  \label{eq:condP2} 
    \end{align}
\end{lemma}
\begin{proof}
We need to establish that conditions \eqref{eq:condP1} and \eqref{eq:condP2} are both necessary and sufficient for the existence of equilibrium at $k^*$ with $0 < \alpha < 1$. 

\paragraph{Necessary conditions} We will show that if there exists an equilibrium at $k^*$ with $0 < \alpha < 1$, then conditions \eqref{eq:condP1} and \eqref{eq:condP2} must be satisfied. We start by showing that condition \eqref{eq:condP2} must be satisfied, via direct inspection of Equation \eqref{exp:alpha}. Recall from Lemma \ref{lem:alphaless1} that if there exists an equilibrium at $k^*$ with $0 < \alpha < 1$, it must be that:
\[
 \alpha = \frac{1}{\numbuyers - k^*} \left[ \Q - \frac{B_{\leq k^*}}{P \numusers} \right].
\]
When $\Q \leq \numbuyers - k^*$, it is trivial to see that $\alpha < 1$ for any value of $P > 0$. However, when $\Q > \numbuyers - k^*$, in order for $\alpha$ in Equation \eqref{exp:alpha} to be strictly less than $1$, we need to ensure that:
\[
   \Q - \frac{B_{\leq k^*}}{P \numusers} < \numbuyers - k^*; 
\] 
and then by slight rearrangement of terms, we have:
\[
   P < \frac{B_{\leq k^*}}{\numusers \left( \Q - (\numbuyers-k^*) \right)} = P^*(k^*).
\]
Now we establish that condition \eqref{eq:condP1} is also satisfied, using the monotonicity property of $\buydec$'s at any equilibrium (refer Claim \eqref{clm:N_monotone}). We note that if there exists an equilibrium at $k^*$, it must be that:
\[
   N_{k^*} < \alpha \numusers \quad \text{and} \quad N_{k^*+1} = \alpha \numusers,
\]
Recall that, for any $k \in [\numbuyers]$, $N_k = \min \left(\alpha \numusers, \frac{B_k}{P} \right)$. Therefore, from the condition above, we get:
\[
   \frac{B_{k^*}}{P \numusers} < \alpha \leq \frac{B_{k^*+1}}{P \numusers}.
\]
Plugging in the expression for $\alpha$ derived in \eqref{exp:alpha} in the above inequality, we can verify that condition \eqref{eq:condP1} is satisfied: 
\begin{align*}
    &\frac{B_{k^*}}{P \numusers} < \frac{1}{\numbuyers - k^*} \left[ \Q - \frac{B_{\leq k^*}}{P \numusers} \right] \leq \frac{B_{k^*+1}}{P \numusers} \\
    &\iff \frac{(\numbuyers - k^*)B_{k^*}}{P \numusers} < \Q - \frac{B_{\leq k^*}}{P \numusers} \leq \frac{(\numbuyers - k^*)B_{k^*+1}}{P \numusers} \\
    &\iff \frac{(\numbuyers-k^*)B_{k^*} + B_{\leq k^*}}{P \numusers} < \Q \leq \frac{(\numbuyers-k^*)B_{k^*+1} + B_{\leq k^*}}{P \numusers} \\
    &\iff  \frac{(\numbuyers-k^*)B_{k^*} + B_{\leq k^*}}{\Q \numusers} < P \leq \frac{(\numbuyers-k^*-1)B_{k^*+1} + B_{\leq k^*+1}}{\Q \numusers}\\
    &\iff \gamma(k^*) < P \leq \gamma(k^*+1).
\end{align*}
This concludes the proof of the necessary conditions for the existence of an equilibrium at $k^*$ with $0 < \alpha < 1$. 

\paragraph{Sufficient conditions} For the proof of sufficiency, we need to show that for any price which satisfies the conditions \eqref{eq:condP1}-\eqref{eq:condP2}, there exists an equilibrium at $k^*$ with $0 < \alpha < 1$. Given a price $P'$ which satisfies conditions \eqref{eq:condP1}-\eqref{eq:condP2}, it suffices to come up with a feasible participation rate $\alpha'$ for which:
\begin{subequations}
\begin{align}
    0 < \alpha' < 1; \label{eq:condP3}
\end{align}
\begin{align}
\frac{B_{k^*}}{P \numusers} < \alpha' \leq \frac{B_{k^*+1}}{P \numusers}; \quad \text{and}\label{eq:condP4}
\end{align}
\begin{align}
    \sum_{k=1}^{\numbuyers} N_k = \Q \numusers. \label{eq:condP5}
\end{align}
\end{subequations}
Fix $k^* < \numbuyers$. Now, pick any price $P'$ that satisfies conditions \eqref{eq:condP1}-\eqref{eq:condP2}. Define $\alpha' = 
\frac{1}{\numbuyers - k^*} \left[ \Q - \frac{B_{\leq k^*}}{P'\numusers} \right]$. The first step is to verify that $\alpha'$ satisfies condition \eqref{eq:condP3}. 

Since $P'$ satisfies condition \eqref{eq:condP1}, this implies $P' > \gamma(k^*) =  \frac{(\numbuyers - k^*)B_{k^*} + B_{\leq k^*}}{\Q \numusers}$. Note that $k^* < \numbuyers$ (by choice) and $B_{k^*} \geq 0$. This means: 
\begin{align*}
    P > \frac{B_{\leq k^*}}{\Q \numusers} 
    \quad \text{which implies} \quad \Q > \frac{B_{\leq k^*}}{P \numusers}.  
\end{align*}
Therefore, 
\[
\alpha'=\frac{1}{\numbuyers - k^*}\left[\Q - \frac{B_{\leq k^*}}{P\numusers} \right] > 0.
\]
Now we argue that $\alpha' < 1$. There are two cases to consider here:
\begin{enumerate}
    \item $\Q \leq \numbuyers - k^*$: When $k^* = 0$, we actually have the strict inequality $\Q < \numbuyers$. In this case, $\alpha'$ reduces to $\frac{\Q}{\numbuyers}$ which is clearly $< 1$.
    When $\Q \leq \numbuyers - k^*$ and $k^* > 0$, we have: 
    \[
    \alpha' = \frac{1}{\numbuyers - k^*}\left[\Q - \frac{B_{\leq k^*}}{P\numusers} \right] \leq \frac{1}{\numbuyers - k^*}\left[\numbuyers-k^* - \frac{B_{\leq k^*}}{P\numusers}\right] < 1.
    \]  
    \item $\Q > \numbuyers - k^*$: Using the fact that $P'$ satisfies condition \eqref{eq:condP2}, we have:
    \begin{align*}
        &P' < \frac{B_{\leq k^*}}{\numusers \left( \Q - (\numbuyers - k^*) \right)} \\
        \implies &\Q - (\numbuyers - k^*) < \frac{B_{\leq k^*}}{P' \numusers} \\
        \implies & \Q - \frac{B_{\leq k^*}}{P' \numusers} < \numbuyers - k^*.
    \end{align*}
    Therefore,
    \begin{align*}
        \alpha' & = \frac{1}{\numbuyers - k^*}\left[\Q - \frac{B_{\leq k^*}}{P\numusers} \right] \\
        & <  \frac{1}{\numbuyers - k^*} (K-k^*)=1.
    \end{align*}
\end{enumerate}
Now, we verify condition \eqref{eq:condP4}. Actually, it suffices to look at the ratios $\frac{B_{k^*}}{P' \alpha' \numusers}$ and $\frac{B_{k^*+1}}{P' \alpha' \numusers}$ and verify that they are $< 1$ and $\geq 1$ respectively.
\begin{align*}
        \frac{B_{k^*}}{P' \alpha' \numusers} & \stackrel{\stepone}{=} \frac{B_{k^*}}{P' \numusers} \cdot \frac{P' \numusers (\numbuyers - k^*)}{\Q P' \numusers - B_{\leq k^*}} \\
        &= \frac{(\numbuyers-k^*)B_{k^*}}{\Q P' \numusers - B_{\leq k^*}}\\
        &\stackrel{\steptwo}{<} \frac{(\numbuyers-k^*)B_{k^*}}{(\numbuyers - k^*)B_{k^*}}
        \\
        &= 1.
\end{align*}
Step \stepone follows by direct substitution of $\alpha'$. 
The inequality in step \steptwo follows from the fact that $P' > \gamma(k^*)$ and therefore, 
\[
P' > \frac{(\numbuyers - k^*)B_{k^*} + B_{\leq k^*}}{\Q \numusers} \quad \text{which implies} \quad \Q P' \numusers - B_{\leq k^*} > (\numbuyers-k^*)B_{k^*}.
\]
We use similar arguments to verify that $\frac{B_{k^*+1}}{P' \alpha' \numusers} \geq 1$. 
\begin{align*}
    \frac{B_{k^*+1}}{P' \alpha' \numusers} & \stackrel{\stepthree}{=} \frac{B_{k^*+1}}{P' \numusers} \cdot \frac{P' \numusers (\numbuyers - k^*)}{\Q P' \numusers - B_{\leq k^*}} \\
    &= \frac{(\numbuyers-k^*)B_{k^*+1}}{\Q P' \numusers - B_{\leq k^*}}\\
    &\geq 1.
\end{align*}
Step \stepthree is by direct substitution of $\alpha'$ while the last step follows from the fact that $P' \leq \gamma(k^*+1)$. 
\[
     P' \leq \frac{(\numbuyers - k^*)B_{k^*+1} + B_{\leq k^*}}{\Q \numusers}, 
\]
By re-arranging terms, we obtain the inequality: 
\[
    \Q P' \numusers - B_{\leq k^*} \leq (\numbuyers - k^*)B_{k^*+1},
\]
which immediately implies the assertion that $\frac{B_{k^*+1}}{P' \alpha' \numusers} \geq 1$. The final piece of the proof is to verify the condition \eqref{eq:condP5}: 
\begin{align*}
    \sum_{k=1}^{\numbuyers} N_k &= \sum_{k=1}^{k^*} N_k + \sum_{k=k^*+1}^{\numbuyers} N_k \\
    &\stackrel{(a)}{=} \sum_{k=1}^{k^*} \frac{B_k}{P'} + \sum_{k^*+1}^{\numbuyers}\alpha' \numusers \\
    &= \frac{B_{\leq k^*}}{P'} + (\numbuyers - k^*)\alpha' \numusers\\
    &\stackrel{(b)}{=} \frac{B_{\leq k^*}}{P'} + \Q \numusers - \frac{B_{\leq k^*}}{P'} \\
    &= \Q \numusers. 
\end{align*}
Step (a) follows from the monotonicity of $B_k$'s and condition \eqref{eq:condP4} which we have already established, while step (b) is by direct substitution of $\alpha'$. Therefore, $\left( \alpha', (\buydec)_{k=1}^{\numbuyers} \right)$ constitute an equilibrium at $k^*$ induced by price $P'$. This concludes the proof. 
\end{proof}
\noindent

\paragraph{Case 2: $\alpha = 1$ (Full Participation)} 
Now, we derive the conditions for the existence of equilibria with full user participation (i.e., $\alpha = 1$). 
\begin{lemma}\label{lem:price_alpha1}
Fix $k^* \in \{1,2, \ldots\numbuyers\}$. There \textbf{exists} an equilibrium at $k^*$ with $\alpha = 1$ if and only if: 
\begin{align}
   \xi(k^*) < P \leq \xi(k^*+1), \label{eq:condP1_alpha1}
\end{align}
and,
\begin{align}
  P \geq P^*(k^*). \label{eq:condP2_alpha1} 
\end{align}
\end{lemma}
\begin{proof}
We start by proving that conditions \eqref{eq:condP1_alpha1}-\eqref{eq:condP2_alpha1} on $P$ are necessary for the existence of equilibrium at $k^*$ with $\alpha = 1$. We then follow it up with a proof of sufficiency. 

\paragraph{Necessary conditions} We will show that if there exists an equilibrium at $k^*$ with $\alpha = 1$, then conditions \eqref{eq:condP1_alpha1}-\eqref{eq:condP2_alpha1} on $P$ must be satisfied. In order to check that condition \eqref{eq:condP1_alpha1} is satisfied, we will use the monotonicity properties of $\buydec$'s at equilibrium established in Claim \ref{clm:N_monotone}. 

Since the said equilibrium is at $k^*$, we have:
\[
    N_{k^*} < \alpha \numusers \quad \text{and} \quad N_{k^*+1} = \alpha \numusers.
\]
Recall that for any $k \in [\numbuyers]$, $\buydec = \min \left(\frac{B_k}{P}, \alpha \numusers \right)$. Therefore, the above condition is equivalent to:
\[
  \frac{B_{k^*}}{P} < \alpha \numusers \leq \frac{B_{k^*+1}}{P}. 
\]
Since $\alpha = 1$, multiplying by $\frac{P}{\numusers}$ throughout, we conclude that condition \eqref{eq:condP1_alpha1} is satisfied.
\noindent
In order to show that condition \eqref{eq:condP2_alpha1} is satisfied, we will inspect the equilibrium conditions on $\alpha$. At equilibrium, $\alpha$ must satisfy:
\[
  \alpha = \min \left(1, \frac{\alpha \Q \numusers}{\sum_{k=1}^{\numbuyers} \buydec} \right). 
\]
Since $\alpha = 1$, it must be that:  
\begin{align}
\frac{\Q \numusers}{\sum_{k=1}^{\numbuyers}\buydec} \geq 1 \quad \text{which implies} \quad \Q \numusers \geq \sum_{k=1}^{\numbuyers}\buydec. \label{eq:temp2}
\end{align}
Furthermore, we know that the said equilibrium is at $k^*$, so we already have an expression for $\sum_{k=1}^{\numbuyers}N_k$ from Equation \eqref{exp:sum_Nk}.
Combining equations \eqref{eq:temp2} and \eqref{exp:sum_Nk}, we obtain: 
\begin{align*}
        &\Q \geq \frac{B_{\leq k^*}}{P \numusers} + (\numbuyers - k^*) \\
        \implies &P \geq \frac{B_{\leq k^*}}{\numusers \left( \Q - (\numbuyers - k^*) \right)} \\
        \implies & P \geq P^*(k^*). 
\end{align*}
Thus, we have shown that condition \eqref{eq:condP2_alpha1} is also satisfied. This concludes the proof. 

\paragraph{Sufficient conditions} For proof of sufficiency, we need to show that $\alpha = 1$ is a feasible equilibrium at any price $P$ which satisfies conditions \eqref{eq:condP1_alpha1}-\eqref{eq:condP2_alpha1}. 

Fix $k^* > 0$. Now pick any price $P'$ which satisfies conditions \eqref{eq:condP1_alpha1}-\eqref{eq:condP2_alpha1}. In order to show that $P'$ induces an equilibrium at $k^*$ with $\alpha = 1$, it is sufficient to show: 
\begin{subequations}
\begin{align}
    \frac{B_{k^*}}{P'} < \numusers \leq \frac{B_{k^*+1}}{P'}; \label{eq:cond6}
\end{align}
and, 
\begin{align}
    \Q \numusers \geq \sum_{k=1}^{\numbuyers} N_k. \label{eq:cond7}
\end{align}
\end{subequations}
Note that since $\xi(k^*) < P' \leq \xi(k^*+1)$, we have: 
\[
    \frac{B_{k^*}}{\numusers} < P' \leq \frac{B_{k^*+1}}{\numusers},
\]
which directly implies condition \eqref{eq:cond6}. 
Finally, we have to verify condition \eqref{eq:cond7}. Note that: 
\[
   \sum_{k=1}^{\numbuyers} \buydec = \frac{B_{\leq k^*}}{P'} + \numusers(\numbuyers - k^*).
\]
Additionally, since $P' \geq P^*(k^*)$, we immediately see: 
\[
       \frac{B_{\leq k^*}}{P'} + \numusers(\numbuyers - k^*) \leq \Q \numusers. 
\]
Hence, $(1, (\buydec)_{k=1}^{\numbuyers})$ constitutes an equilibrium at $k^*$ induced by $P'$. This concludes the proof of sufficiency. 
\end{proof}

%%%%
%%%%
%%%%
%%%%
%%%%

\subsection{Part 2: Our Simplified Theorems \ref{thm:eq_exist_highQ}, \ref{thm:eq_exist_modQ} \& \ref{thm:eq_exist_lowQ}}\label{app:constant_thm_simple}
In this section, we present the proof of Theorem \ref{thm:eq_exist_highQ} and show how to simplify the results of Appendix \ref{app:constant_eq_chart} to obtain Theorem \ref{thm:eq_exist_modQ} and with slight modification, Theorem \ref{thm:eq_exist_lowQ}. 

\subsubsection{Proof of Theorem \ref{thm:eq_exist_highQ}: High Benefit Case $(\Q \geq \numbuyers)$}
When $\Q \geq \numbuyers$, we will show that at any price, the only equilibrium possible is with $\alpha = 1$. Pick any price $P > 0$. 
Note that: 
\begin{align}
   \sum_{k=1}^{\numbuyers}N_k = \sum_{k=1}^{\numbuyers} \min \left(\alpha \numusers, \frac{B_k}{P} \right) \leq \numbuyers \alpha \numusers.  \label{ineq:high_benefit}
\end{align}
Recall from \ref{opt:platform} that at equilibrium, $\alpha$ is given by $\alpha = \min \left(1, \frac{\alpha \Q \numusers}{\sum_{k=1}^{\numbuyers}N_k} \right)$. Therefore, 
\[
    \alpha = \min \left(1, \frac{\alpha \Q \numusers}{\sum_{k=1}^{\numbuyers}N_k} \right) \stackrel{\stepone}{\geq} \min \left(1, \frac{\alpha \Q \numusers}{\numbuyers \alpha \numusers} \right) = \min \left(1, \frac{\Q}{\numbuyers} \right) \stackrel{\steptwo}{=} 1. 
\]
Step~\stepone follows from \eqref{ineq:high_benefit} while step~\steptwo follows from the fact that $\Q \geq \numbuyers$. This concludes the proof.

\subsubsection{Proof of Theorem \ref{thm:eq_exist_modQ}: Moderate Benefit Case ($\frac{B_{\leq \numbuyers}}{B_{\numbuyers}} \leq Q < \numbuyers$)} \label{app:constant_mod_red}

We need to establish that for any price $P > 0$, there exists an unique market equilibrium for the two-sided data marketplace and additionally, there is a price threshold which separates the partial participation and full participation regimes. We start by providing a brief proof sketch. The proof structure is as follows: 
\begin{enumerate}
    \item First, we need to show that there exists a unique value of $k^*$, given by $0 < \bar k < \numbuyers$ which admits both types of equilibria (with $\alpha < 1$ and $\alpha = 1$). For all $k^* < \bar k$, there only exists equilibria with partial participation ($\alpha < 1$) and for all $k^* > \bar k$, there only exists equilibria with full participation ($\alpha = 1$).
    \item The price range which induces equilibria at $\bar k$ is partitioned by the price threshold $\bar P = P^*(\bar k)$. To the left of $P^*(\bar k)$, there only exists partial participation equilibria $(\alpha < 1)$, while to the right of $P^*(\bar k)$, there only exists full participation equilibria. 
    \item The range of prices which induce equilibria at different values of $k^* \in \{0,1,2....\numbuyers\}$ are non-overlapping and span $\mathbb{R}_{>0}$.
\end{enumerate}
We will complete the proof in the same order as above. 

\paragraph{Subpart 1} The proof idea is as follows: first, we establish necessary and sufficient conditions on the budget sum $B_{\leq k^*}$ for existence of equilibria at $k^*$ with partial participation (Claim \ref{clm:condk1}) and full participation (Claim \ref{clm:condk2}). We use these claims to establish existence of $\bar k$ and then argue uniqueness of $\bar k$.     

\begin{claim}\label{clm:condk1}
    There exists an equilibrium at $k^* > 0$ with $\alpha < 1$ if and only if $k^*$ satisfies: 
    \begin{align}
    B_{\leq k^*} > \left(\Q - (\numbuyers-k^*) \right) B_{k^*}. \label{condk1}
    \end{align}
\end{claim}
\begin{proof}
By Lemma \ref{lem:price_alphaless1}, we have the necessary and sufficient conditions for the existence of equilibria with $0 < \alpha < 1$ at $k^*$, given by: 
\[
             \gamma(k^*) < P \leq \gamma(k^*+1);
\]
and additionally, if $\Q > \numbuyers - k^*$, 
\[
            P < P^*(k^*).
\]
Therefore, it suffices to show that condition \eqref{condk1} is true if and only if there exists a feasible price $P$ which induces equilibrium at $k^*$ with $\alpha < 1$. 

Suppose, there exists a feasible price $P$ which induces equilibrium at $k^*$ with $\alpha < 1$. Observe that when $\Q \leq \numbuyers - k^*$, condition \eqref{condk1} is satisfied trivially because $\Q - (\numbuyers - k^*) \leq 0$ which makes the RHS $\leq 0$ while $B_{\leq k^*} > 0$ (since $k^* > 0$). In the other case when $\Q > \numbuyers - k^*$, there must exist a feasible price $P$ which satisfies both conditions \eqref{eq:condP1} and \eqref{eq:condP2}. This implies: 
\begin{align*}
         &P^*(k^*) > \gamma(k^*) \\
         \implies& \frac{B_{\leq k^*}}{\numusers \left(\Q - (\numbuyers - k^*) \right)} > \frac{(\numbuyers - k^*)B_{k^*} + B_{\leq k^*}}{\Q \numusers}.
\end{align*}
All numerators and denominators are strictly positive, so cross-multiplying and re-arranging terms, we find that \eqref{condk1} is satisfied. This concludes the proof that \eqref{condk1} is necessary for existence of equilibria at $k^*$ with $\alpha < 1$. 

For the proof of sufficiency, suppose there is some $k^* > 0$ which satisfies \eqref{condk1}. If $\Q \leq \numbuyers - k^*$, we can pick any price $P'$ in the non-empty interval $\left( \gamma(k^*), \gamma(k^*+1) \right]$. Otherwise, if $\Q > \numbuyers - k^*$, \eqref{condk1} implies $P^*(k^*) > \gamma(k^*)$. Therefore, we can again find a feasible price $P'$ which satisfies both conditions \eqref{eq:condP1}-\eqref{eq:condP2}. Then, by Lemma \ref{lem:price_alphaless1}, in either case, there exists $P'$ which will induce an equilibrium at $k^*$ with $\alpha < 1$. 
\end{proof}

\begin{claim}\label{clm:condk2}
There exists an equilibrium at $k^* > 0$ with $\alpha = 1$ if and only if $k^*$ satisfies: 
\begin{align}
    B_{\leq k^*} \leq \left( \Q - (\numbuyers - k^*) \right) B_{k^*+1}.  \label{condk2}
\end{align}
\end{claim}
\begin{proof}
The proof is identical to the proof of Claim \ref{clm:condk1}. We use the conditions on $P$ in Lemma \ref{lem:price_alpha1} for existence of equilibria at $k^* > 0$ with $\alpha = 1$. Condition \eqref{condk2} is derived by noting that there must exist a feasible price $P$ which satisfies both conditions \eqref{eq:condP1_alpha1}-\eqref{eq:condP2_alpha1} and therefore, 
\[
       P^*(k^*) \leq \xi(k^*+1). 
\]
\end{proof}

\paragraph{Existence of $\bar k$} In order to prove existence of $\bar k$, we will first show that there does not exist $0 < k^* < \numbuyers$ such that $k^*$ admits equilibria with $\alpha = 1$ and $k^*+1$ admits equilibria with $\alpha < 1$ (Claim \ref{clm:alpha_ordered}). This will prove that all values of $k^*$ which admit partial participation equilibria \textbf{precede} all values of $k^*$ which admit full participation equilibria. Then we will use contradiction to establish that there must exist some $k^*$ which admits both types of equilibria. 

\begin{claim}\label{clm:alpha_ordered}
There does not exist $0 < k^* < \numbuyers$ such that $k^*$ admits equilibria with $\alpha = 1$ and $k^*+1$ admits equilibria with $\alpha < 1$. 
\end{claim}
\begin{proof}
The proof is by contradiction. Suppose, there exists such a $k^*$. Then, using Claims \ref{clm:condk1} and \ref{clm:condk2}, we have: 
\begin{align}
      B_{\leq k^*} \leq \left( \Q - (\numbuyers - k^*) \right) B_{k^*+1}; \label{ineq:temp1}
\end{align}
and 
\begin{align}
      B_{\leq k^*+1} < \left( \Q - (\numbuyers - k^*-1) \right)B_{k^*+1}. \label{ineq:temp2} 
\end{align}
But condition \eqref{ineq:temp2} implies: 
\[
       B_{\leq k^*+1} < \left( \Q - (\numbuyers - k^*) \right)B_{k^*+1} + B_{k^*+1} \implies B_{\leq k^*} \leq \left( \Q - (\numbuyers - k^*) \right) B_{k^*+1};
\]
which contradicts condition \eqref{ineq:temp1}. This concludes the proof. 
\end{proof}
\paragraph{Proof of existence of $\bar k$:} Let $\widetilde k$ be the last value of $k^*$ which admits equilibria with $\alpha < 1$, therefore, $\widetilde k+1$ admits equilibria with $\alpha = 1$. Since $\bar k$ does not exist, it must be that $\widetilde k$ satisfies condition \eqref{condk1}, while $\widetilde k+1$ violates \eqref{condk1}, i.e., 
\begin{align}
    B_{\leq \widetilde k} > \left( \Q - (\numbuyers - \widetilde k) \right)B_{\widetilde k}; \label{ineq:temp3}
\end{align}
and 
\begin{align}
    B_{\leq \widetilde k+1} \leq \left( \Q - (\numbuyers - \widetilde k-1) \right)B_{\widetilde k+1}. \label{ineq:temp4}
\end{align}
But \eqref{ineq:temp4} implies that $B_{\leq \widetilde k} \leq \left( \Q - (\numbuyers - \widetilde k) \right)B_{\widetilde k+1}$. Combining this with \eqref{ineq:temp3}, we obtain: 
\[
    \left( \Q - (\numbuyers - \widetilde k) \right)B_{\widetilde k} < B_{\leq \widetilde k} \leq \left( \Q - (\numbuyers - \widetilde k) \right)B_{\widetilde k+1};
\]
which implies that $\widetilde k$ admits both partial and full participation equilibria and is a candidate value of $\bar k$. This concludes the proof of the existence result. 

\paragraph{Uniqueness of $\bar k$:} We can establish uniqueness of $\bar k$ using a straightforward argument. Suppose $\bar k$ is not unique, therefore, there must exist $\bar k_1$ and $\bar k_2$ which admit both types of equilibria. WLOG, assume $\bar k_1 < \bar k_2$. But, in this case, we have found two values of $k^*$ where the smaller value admits equilibria with $\alpha = 1$ and the larger value admits equilibria with $\alpha < 1$. This contradicts Claim \ref{clm:alpha_ordered} and proves uniqueness of $\bar k$. \\  
    
The uniqueness of $\bar k$ and Claim \ref{clm:alpha_ordered} directly implies the following result: 
\begin{claim}\label{clm:alpha_monotone}
Let $\bar k > 0$ be the unique value of $k^*$ that admits equilibria with both $\alpha < 1$ and $\alpha = 1$. Then for all $k^* < \bar{k}$, the only possible equilibria at $k^*$ is with $\alpha < 1$ and for all $k^* > \bar k$, the only possible equilibria is with $\alpha = 1$. 
\end{claim}

\paragraph{Subpart 2} Here, we will show that the range of prices which produce equilibria at $\bar k$ with $\alpha < 1$ and $\alpha = 1$ are non-overlapping and separated by $P^*(\bar k)$. In order to prove this result, we will introduce two intermediate results in Claims \ref{clm:P_binding_alphaless1} and \ref{clm:P_binding_alpha1}. 

\begin{claim}\label{clm:P_binding_alphaless1}
The constraint on price $P$ in Lemma \ref{lem:price_alphaless1} given by: 
\[
        P < P^*(k^*),
\]
is binding if and only if $k^* = \bar k$. 
\end{claim}

\begin{claim}\label{clm:P_binding_alpha1}
The constraint on price $P$ in Lemma \ref{lem:price_alpha1} given by: 
\[
        P \geq P^*(k^*),
\]
is binding if and only if $k^* = \bar k$. 
\end{claim}
\paragraph{Proof of Subpart 2}
The proof follows directly from Claims \ref{clm:P_binding_alphaless1} and \ref{clm:P_binding_alpha1}. By Claim \ref{clm:P_binding_alphaless1}, the range of prices $P$ which induce equilibria at $\bar k$ with $\alpha < 1$ is given by: 
\begin{align}
        \gamma(\bar k)  < P < P^*(\bar k). \label{ineq:temp*}
\end{align}
Similarly, by claim \ref{clm:P_binding_alpha1}, the range of prices $P$ which induce equilibria at $\bar k$ with $\alpha = 1$ is given by: 
\begin{align}
           P^*(\bar k) \leq P \leq \xi(\bar k+1). \label{ineq:temp**}
\end{align}
It can be shown easily that the price ranges in \eqref{ineq:temp*} and \eqref{ineq:temp**} are non-empty and clearly, they are non-overlapping, separated by $P^*(\bar k)$. This, combined with Claim \ref{clm:alpha_monotone} proves that the $\alpha < 1$ and $\alpha = 1$ regimes lie to the left and right of $P^*(\bar k)$ respectively.  
\paragraph{Subpart 3} 
For $k^* < \bar k$, the only possible equilibria at $k^*$ is $\alpha < 1$ (by subpart 1). The range of prices which induce equilibrium at $k^*$ with $\alpha < 1$, is given by: 
\[
         \gamma(k^*) < P \leq \gamma(k^*+1). 
\]
It is immediately seen that these price ranges are non-overlapping and their union is the set: 
\begin{align}
\mathcal{R}_1 := \left(0, \gamma(\bar k) \right]. \label{range1}
\end{align}
For $k^* > \bar k$, the only possible equilibria at $k^*$ is $\alpha = 1$ (again, by subpart 1). The range of prices which induce equilibrium at $k^*$ with $\alpha = 1$, is given by: 
\[
          \xi(k^*) < P \leq \xi(k^*+1).
\]
Again, these ranges do not overlap and their union is the set:
\begin{align}
\mathcal{R}_2 := \left(\xi(\bar k+1), +\infty \right). \label{range2}
\end{align}
Finally, from subpart 2, the range of prices which induce equilibrium at $\bar k$, is given by:
\begin{align}
        \mathcal{R}_3 := \left(\gamma(\bar k ), \xi(\bar k+1)\right]. \label{range3}
\end{align}
Clearly, $\mathcal{R}_1 \cup \mathcal{R}_2 \cup \mathcal{R}_3 = \mathbb{R}_{>0}$. Uniqueness of equilibrium at any price $P > 0$ follows from the fact that for each $k^* \in \{0,1,2...\numbuyers\}$, the range of prices which induce equilibria at $k^*$ are non-overlapping. 

In summary, we have shown that there exists a price threshold $\bar P = P^*(\bar k)$ for some $0 < \bar k < \numbuyers$ such that for all $P < \bar P$, there exists a unique partial participation equilibrium while for all $P \geq \bar P$, there exists a unique full participation equilibrium. This concludes the proof of Theorem \ref{thm:eq_exist_modQ}.

\subsubsection{Proof of Theorem \ref{thm:eq_exist_lowQ}: Low Benefit Case $(\Q < \frac{B_{\leq \numbuyers}}{B_{\numbuyers}})$}
We argue that in order to prove Theorem \ref{thm:eq_exist_lowQ}, it is sufficient to show that when $\Q < \frac{B_{\leq \numbuyers}}{B_{\numbuyers}}$, the following statements are true: 
\begin{enumerate}
    \item (Statement 1) $\bar k = \numbuyers$. \label{statement1}
    \item (Statement 2) There are multiple partial participation equilibria at $k^* = \numbuyers$, all occurring at the price $P = \gamma(\numbuyers)$. \label{statement2}
\end{enumerate}
Assuming (for now) that the above statements are true, we present the following sequence of arguments that will help us prove the theorem. Note that we reuse many of the results from section \ref{app:constant_mod_red}.
\begin{enumerate}
    \item Since $\bar k = \numbuyers$, $\numbuyers$ admits both types of equilibria with $\alpha < 1$ and $\alpha = 1$.
    \item By Claim \ref{clm:alpha_monotone}, all $k^* < \numbuyers$ only admit equilibria with $0 < \alpha < 1$. 
    \item By Claim \ref{clm:P_binding_alphaless1}, $\bar k$ is the only value of $k^*$ which admits partial participation equilibria and where the condition $P < P^*(k^*)$ is binding. Since $\bar k = \numbuyers$, for all $k^* < \numbuyers$, the range of prices which induce equilibria at $k^*$ with $\alpha < 1$ is given by $\gamma(k^*) < P \leq \gamma(k^*+1)$. This follows from Lemma \ref{lem:price_alphaless1}. 
    \item By Claim \ref{clm:P_binding_alpha1}, $\bar k$ is the only value of $k^*$ which admits full participation equilibria such that the condition $P \geq P^*(k^*)$ is binding. Hence, using Lemma \ref{lem:price_alpha1}, the range of prices which induce equilibria at $\bar k$ with $\alpha = 1$ is given by $P^*(\bar k) \leq P \leq \xi(\bar k+1)$. Noting that $\bar k = \numbuyers$, $P^*(\bar k) = \gamma(\numbuyers)$ and $\xi(\bar k + 1) = +\infty$, we conclude that for all $P \geq \gamma(\numbuyers)$, there exists equilibria at $\numbuyers$ with $\alpha = 1$. 
    \item Since all the partial participation equilibria at $\numbuyers$ occur at price $\gamma(\numbuyers)$, for all $P > \gamma(\numbuyers)$, $\alpha = 1$ is the only equilibrium.
\end{enumerate}
Now consider price threshold $\bar P = \gamma(\numbuyers)$. For all $P > \bar P$, there only exists equilibria with full participation (by subpart 5). For all $P < \bar P$, the only equilibria are with $\alpha < 1$ where $\alpha$ takes the form: 
\[
    \alpha = \frac{1}{\numbuyers - k^*}\left[\Q - \frac{B_{\leq k^*}}{P\numusers} \right]
\]
for $\gamma(k^*) < P \leq \gamma(k^*+1)$ for all $k^* < \numbuyers$ (subparts 2 and 3). Finally, at $P = \gamma(\numbuyers)$, there exists multiple equilibria with $\alpha \in \left[ \frac{\Q B_{\numbuyers}}{B_{\leq \numbuyers}}, 1 \right]$. This proves Theorem \ref{thm:eq_exist_lowQ}. 

Now that we have established how statements \ref{statement1} and \ref{statement2} prove Theorem \ref{thm:eq_exist_lowQ}, we go back to proving the statements, using Lemma \ref{lem:lowbenefit} and Claims \ref{clm:lowben_feas_alpha} and \ref{clm:lowbenefit_bark1}. 

\begin{lemma}\label{lem:lowbenefit}
There exists an equilibrium at $k^* = \numbuyers$ with $0 < \alpha < 1$ if and only if: 
\begin{align}
    \Q B_{\numbuyers} < B_{\leq \numbuyers}. \label{eq:condB_lowbenefit}
\end{align}
Additionally, the price $P$ which induces all such equilibria is given by:
\begin{align}
    P = \frac{B_{\leq \numbuyers}}{\Q \numusers} = \gamma(\numbuyers). \label{eq:condP_lowbenefit}
\end{align}
\end{lemma}
\begin{proof}
The first part of the proof follows directly from Claim \ref{clm:condk2}. Plugging $k^* = \numbuyers$, we see immediately that there exists an equilibrium at $\numbuyers$ if and only if: 
\[
     \Q B_{\numbuyers} < B_{\leq \numbuyers}.    
\]
We just need to show that all such equilibria are induced at $P = \gamma(\numbuyers)$. Suppose, there exists an equilibrium at $k^* = \numbuyers$ with $0 < \alpha < 1$. By the definition of $k^*$ in Claim \ref{clm:N_monotone}, we have:
\[
N_k = \frac{B_k}{P} \quad \forall~k \in [\numbuyers].
\]
This allows us to derive the expression for $\sum_k N_k$ as follows:
\[
   \sum_{k=1}^{\numbuyers} N_k = \sum_{k=1}^{\numbuyers} \frac{B_k}{P} = \frac{B_{\leq \numbuyers}}{P}.  
\]
Now, recall from \ref{opt:platform} that at equilibrium, $\alpha$ is given by: 
\[
    \alpha = \min \left(1, \frac{\alpha \Q \numusers}{\sum_{k=1}^{\numbuyers}N_k} \right).
\]
Since $\alpha < 1$ by assumption, it must be that $\alpha = \frac{\alpha \Q \numusers}{\sum_{k=1}^{\numbuyers}N_k}$ which implies $\Q \numusers = \sum_{k=1}^{\numbuyers}N_k$ (since $\alpha > 0$). Substituting the expression for $\sum_k N_k$, we obtain the desired condition that $P = \frac{B_{\leq K}}{\Q \numusers} = \gamma(\numbuyers)$. This concludes the proof. 
\end{proof}
We can use Lemma \ref{lem:lowbenefit} to show that any $\alpha \in \left( \frac{\Q B_{\numbuyers}}{B_{\leq \numbuyers}}, 1\right)$ is a partial participation equilibrium at $\numbuyers$ induced by $P = \gamma(\numbuyers)$. 

\begin{claim}\label{clm:lowben_feas_alpha}
When $\Q B_{\numbuyers} < B_{\leq \numbuyers}$, any $\alpha \in \left( \frac{\Q B_{\numbuyers}}{B_{\leq \numbuyers}}, 1\right)$ is a feasible partial participation equilibrium at $\numbuyers$.
\end{claim}
\begin{proof}
By Lemma \ref{lem:lowbenefit}, we know that any partial participation equilibrium at $\numbuyers$ is induced by $P = \gamma(\numbuyers)$. Now, pick any $\alpha'$ such that $\frac{\Q B_{\numbuyers}}{B_{\leq \numbuyers}} < \alpha' < 1$. We need to show that $\alpha'$ is a feasible partial participation equilibrium at $\numbuyers$. It is sufficient to verify the following: 
\begin{align}
    \frac{B_{\numbuyers}}{P} < \alpha' \numusers;\label{useless_cond1}
\end{align}
and, 
\begin{align}
    \Q \numusers = \sum_{k=1}^{\numbuyers}N_k. \label{useless_cond2}
\end{align}
We can verify condition \eqref{useless_cond1} as follows: 
\begin{align*}
    \frac{B_{\numbuyers}}{P} = \numusers \cdot \frac{\Q B_{\numbuyers}}{B_{\leq \numbuyers}} < \alpha' \numusers. 
\end{align*}
Since $\frac{B_{\numbuyers}}{P} < \alpha' \numusers$, by the budget monotonicity assumption, $\frac{B_k}{P} < \alpha' \numusers$ for all $k \in [\numbuyers]$. This implies, $\sum_{k=1}^{\numbuyers}N_k = \frac{B_{\leq \numbuyers}}{P} = \Q \numusers$. This verifies condition \eqref{useless_cond2} and concludes the proof. 
\end{proof}

\begin{claim}\label{clm:lowbenefit_bark1}
There exists equilibria at $\numbuyers$ with $\alpha = 1$. 
\end{claim}
\begin{proof}
Since $B_{\numbuyers +1} = +\infty$, we have: 
\begin{align*}
    B_{\leq \numbuyers} \leq \left(\Q - (\numbuyers - \numbuyers) \right) B_{\numbuyers +1}.
\end{align*}
Therefore, $\numbuyers$ satisfies condition \eqref{condk2}. Hence, by Claim \ref{clm:condk2}, there exists equilibria at $\numbuyers$ with $\alpha = 1$.
\end{proof}
\noindent
Combining the results of Lemma \ref{lem:lowbenefit} and Claim \ref{clm:lowbenefit_bark1}, we see that $\numbuyers$ induces equilibria with both $\alpha < 1$ and $\alpha = 1$. Therefore, $\bar k = \numbuyers$. Further, Lemma \ref{lem:lowbenefit} and Claim \ref{clm:lowben_feas_alpha} show that there are multiple partial participation equilibria at $k^* = \numbuyers$ at the price $P = \gamma(\numbuyers)$.

%%%%%%%%%%%

\section{Proofs for the linear case}\label{app:linear}
\subsection{Part 1: A Full Equilibrium Characterization}
We provide a complete closed-form characterization of market equilibria in the linear case under all possible benefit regimes.   

\subsubsection{Proof of Theorem \ref{thm:lin_highben}: High Benefit Case ($C\numusers > \numbuyers$)} 
We start by introducing the following claim which shows that in the high benefit case, the only possible non-trivial equilibrium is with $\alpha = 1$, for any price $P > 0$. Note that this claim also directly proves Theorem \ref{thm:lin_highben}.
\begin{claim}
For any price $P > 0$, $\alpha$ is a feasible market equilibrium if and only if $\alpha = 1$. 
\end{claim}
\paragraph{Necessary Conditions} Recall from \eqref{opt:lin} that at equilibrium, the participation rate $\alpha$ is given by:
\[
\alpha = \min \left(1, \frac{C\alpha^2 \numusers^2}{\sum_{k=1}^{\numbuyers}N_k} \right).
\]
Given $P > 0$, $N_k = \min \left( \frac{B_k}{P}, \alpha \numusers \right)$ for all $k \in [\numbuyers]$. Therefore, $\sum_{k=1}^{\numbuyers}N_k \leq \numbuyers \alpha \numusers$. This implies, 
\begin{align*}
    \alpha &= \min \left(1, \frac{C\alpha^2 \numusers^2}{\sum_{k=1}^{\numbuyers}N_k} \right) \\
    &\geq \min \left(1, \frac{C\alpha^2 \numusers^2}{\numbuyers \alpha \numusers} \right) \\
    &= \min \left(1, \frac{\alpha C\numusers}{\numbuyers} \right).
\end{align*}
Note that $\alpha < 1$ can never be a solution. This is because $\frac{\alpha C\numusers}{\numbuyers} >\alpha$ (since $\frac{C\numusers}{\numbuyers} > 1$), therefore, $\min\left(1, \frac{\alpha C\numusers}{\numbuyers}\right) > \alpha$ which violates the inequality derived above. Therefore, $\alpha$ must be equal to $1$. 

\paragraph{Sufficient Conditions} For sufficiency, we need to show that $\alpha = 1$ constitutes a feasible equilibrium for any price $P > 0$. It is enough to verify that: 
\[
          C\numusers^2 \geq \sum_{k=1}^{\numbuyers}N_k.
\]
Now, at price $P$ and $\alpha = 1$, $N_k = \min\left(\frac{B_k}{P}, \numusers \right)$. Therefore, $\sum_{k=1}^{\numbuyers}N_k \leq \numbuyers \numusers < C\numusers^2$. The last inequality follows because $C\numusers > \numbuyers$. This concludes the proof.  

\subsubsection{Proof of Theorem \ref{thm:lin_splcase}: Special Case ($C\numusers = \numbuyers$)}
In this segment, we introduce the following two claims: Claim \ref{clm:lin_splcase} finds a necessary and sufficient condition on the buyer threshold $k^*$ for the existence of partial participation equilibria. It also provides insights on the range of such equilibria at any price $P$. Claim \ref{clm:lin_splcase2} shows that any price $P > 0$ also induces a full participation equilibrium. Thus, we have a complete characterization of equilibria at any price $P$ and it should be clear that this proves Theorem \ref{thm:lin_splcase}. 

\begin{claim}\label{clm:lin_splcase}
When $C\numusers = \numbuyers$, there exists an equilibrium at $k^*$ with $0 < \alpha < 1$ if and only if $k^* = 0$. The range of prices $P$ which induce such equilibria spans $\mathbb{R}_{>0}$. Additionally, at each price $P > 0$, $\alpha$ is a partial participation equilibrium at $k^* = 0$ if and only if $\alpha \in \left(0, \frac{B_1}{P\numusers} \right]\cap (0, 1)$.
\end{claim}

\paragraph{Necessary Conditions} Suppose, there exists an equilibrium at $k^*$ with $0 < \alpha < 1$. This implies that the following two conditions must be true: 
\begin{align}\label{eq:lin_cond1}
    \sum_{k=1}^{\numbuyers}N_k = C\alpha \numusers^2;
\end{align}
and 
\begin{align}\label{ineq:lin_condP}
    \frac{B_{k^*}}{P} < \alpha \numusers \leq \frac{B_{k^*+1}}{P}.
\end{align}
Using condition \eqref{ineq:lin_condP} and the budget monotonicity assumption, we can derive $\sum_{k=1}^{\numbuyers}N_k$ as follows: 
\[
\sum_{k=1}^{\numbuyers}N_k = \frac{B_{\leq k^*}}{P} + \alpha \numusers (\numbuyers - k^*).
\]
Now, condition \eqref{eq:lin_cond1} implies: 
\begin{align*}
    \frac{B_{\leq k^*}}{P} + \alpha \numusers (\numbuyers - k^*) = C\alpha \numusers^2 = \numbuyers \alpha \numusers.
\end{align*}
In the last equality step, we use $C\numusers = \numbuyers$. Therefore, we have: 
\begin{align}\label{eq:lin_temp1}
        k^* \alpha \numusers = \frac{B_{\leq k^*}}{P}. 
\end{align}
Now, we will argue that $k^*$ must be zero. It suffices to show that for any $k^* > 0$, conditions \eqref{ineq:lin_condP} and \eqref{eq:lin_temp1} cannot be satisfied simultaneously. We will prove by contradiction. Suppose, there does exist some $k^* > 0$ which satisfies both conditions simultaneously. Therefore, it must be that: 
\[
     \alpha \numusers = \frac{B_{\leq k^*}}{Pk^*},
\]
which when plugged into condition \eqref{ineq:lin_condP}, implies:
\[
       \frac{B_{k^*}}{P} < \frac{B_{\leq k^*}}{Pk^*} \leq \frac{B_{k^*+1}}{P}.
\]
But this means $B_{\leq k^*} > k^* B_{k^*}$ which is a contradiction (due to the budget monotonicity assumption). This concludes the proof. 

\paragraph{Sufficient Conditions} For sufficiency, we need to show that if $k^* = 0$, there exists an equilibrium at $k^*$ with $0 < \alpha < 1$. So, it suffices to come up with any feasible equilibrium at $k^* = 0$ with $0 < \alpha < 1$.
Note that when $k^* = 0$, condition \eqref{ineq:lin_condP} reduces to: 
\[
       \alpha \numusers \leq \frac{B_1}{P},
\]
and \eqref{eq:lin_temp1} is trivially satisfied for any $P > 0$ since $B_{\leq 0} = 0$ (Refer to Defn \ref{defn:Bleq}). Therefore, given price $P > 0$, any $\alpha \in \left(0, \frac{B_1}{P\numusers} \right] \cap \left(0, 1 \right)$ is a feasible equilibrium at $k^* = 0$. This concludes the proof of sufficiency. 

\begin{claim}\label{clm:lin_splcase2}
When $C\numusers = \numbuyers$, there exists an equilibrium with $\alpha = 1$ for any price $P > 0$. 
\end{claim}
\begin{proof}
To prove this claim, it is sufficient to check that for any price $P > 0$, $\alpha = 1$ is a feasible equilibrium. Recall that at equilibrium, $\alpha = \min\left(1, \frac{C\alpha^2 \numusers^2}{\sum_{k=1}^{\numbuyers}N_k} \right)$. Therefore, verifying whether $\alpha = 1$ is a solution is equivalent to checking that:
\[
      C\numusers^2 \geq \sum_{k=1}^{\numbuyers} N_k. 
\]
We know that $\sum_{k=1}^{\numbuyers}N_k \leq \numbuyers \numusers$. So the proof is completed just by noting that $\numbuyers = C\numusers$ which implies $\numbuyers \numusers = C\numusers^2$. 
\end{proof}

\subsubsection{Low Benefit Case ($C \numusers < \numbuyers$)}
In this segment, we provide a full closed-form characterization of equilibria in the low benefit case through Theorem \ref{thm:lin_exist}. We will demonstrate later in Section \ref{app:lin_simplethm} how Theorem \ref{thm:lin_exist} reduces to the simplified Theorem \ref{thm:lin_lowben}. 

We start with Claim \ref{clm:lin_noeq} which characterizes the range of buyer threshold $k^*$ for which there are no equilibria. 

\begin{claim}\label{clm:lin_noeq}
There exists no market equilibrium with buyer threshold $k^*$ if $k^* \leq \numbuyers - C \numusers$. 
\end{claim}
\begin{proof}
The proof is in two parts: we show that there cannot exist an equilibrium at any such $k^*$ with either $0 < \alpha < 1$ or $\alpha = 1$. We will prove by contradiction. 

Suppose there exists some $0 \leq k^* \leq \numbuyers - C\numusers$ such that there exists an equilibrium at $k^*$ with $0 < \alpha < 1$. This implies that conditions \eqref{eq:lin_cond1} and \eqref{ineq:lin_condP} must be true. Following exactly identical steps as in the proof of lemma \ref{clm:lin_splcase}, we have:  
\[
     \frac{B_{\leq k^*}}{P} + \alpha \numusers (\numbuyers - k^*) = C\alpha \numusers^2,
\]
which implies
\[
     \alpha \numusers \left(C \numusers - (\numbuyers - k^*) \right) = \frac{B_{\leq k^*}}{P}.
\]
Since $k^* \leq \numbuyers - C\numusers$, this implies $C\numusers - (\numbuyers - k^*) \leq 0$. Now, if $k^* > 0$, $B_{\leq k^*} > 0$ which leads to a contradiction because the LHS is $\leq 0$ while the RHS is strictly $> 0$. If $k^* = 0$, then the RHS $= 0$, but for the LHS to equal $0$, it must be that $C\numusers = \numbuyers$ which is not true by assumption ($C\numusers < \numbuyers$). Thus, we again have a contradiction.  

For the second part of the proof, we need to show that there does not exist $k^*$ such that $0 \leq k^* \leq \numbuyers - C\numusers$ and there exists an equilibrium at $k^*$ with $\alpha = 1$. If there does exist such a $k^*$, then it must be that: 
\[
          C\numusers^2 \geq \frac{B_{\leq k^*}}{P} + \numusers(\numbuyers - k^*),
\]
which by re-arranging terms, is equivalent to:
\[
          \numusers \left( C\numusers - (\numbuyers - k^*)\right) \geq \frac{B_{\leq k^*}}{P}. 
\]
Using an identical argument as earlier, we can show that the LHS is always $\leq 0$ so there does not exist $k^*$ for which the above inequality holds. This concludes the proof. 
\end{proof}

We now state our main theorem for $k^* > \numbuyers - C \numusers$. This theorem is a stronger version of Theorem \ref{thm:lin_lowben} that also characterizes which values of $k^*$ our equilibria arise at. 

\begin{theorem} \label{thm:lin_exist}
Suppose, $C\numusers < \numbuyers$. Fix $k^* \in \{0,1,2,....\numbuyers\}$ such that $k^* > \numbuyers - C\numusers$. There exists an equilibrium with buyer threshold $k^*$ at price $P$ with: 
\begin{itemize}
    \item full participation of users $(\alpha = 1)$ if and only if:
    \[
          \frac{B_{k^*}}{\numusers} < P \leq \frac{B_{k^*+1}}{\numusers} \quad \text{and} \quad P \geq \frac{B_{\leq k^*}}{C\numusers^2 - \numusers (\numbuyers - k^*)};
    \]
    \item partial participation of users $(0 < \alpha < 1)$ if and only if:
    \[
         P > \frac{B_{\leq k^*}}{C\numusers^2 - \numusers (\numbuyers - k^*)}, 
    \]
    and additionally, $k^*$ satisfies: 
    \[
           \frac{B_{k^*}}{N} < \frac{B_{\leq k^*}}{C\numusers^2 - \numusers (\numbuyers - k^*)} \leq \frac{B_{k^*+1}}{N}.
    \]
    In this case, the participation rate of users $\alpha$ is given by:
    \[
        \alpha = \frac{B_{\leq k^*}}{C\numusers^2 - \numusers (\numbuyers - k^*)} \cdot \frac{1}{P}.
    \]
\end{itemize}    
\end{theorem}

We will complete the proof in two parts: first the partial participation case ($0 < \alpha < 1$) followed by the full participation case ($\alpha = 1$). 
\paragraph{Case 1: $0 < \alpha < 1$ (Partial Participation)}
In this segment, we derive necessary and sufficient conditions for the existence of partial participation equilibria. We start by obtaining a closed form expression for the participation rate $\alpha$ for an equilibrium at buyer threshold $k^*$ when we know that $\alpha$ is fractional. 

\begin{lemma}\label{lem:lin_alphaless1}
If there exists an equilibrium at buyer threshold $k^*$ (with $k^* > \numbuyers - C\numusers$) induced by price $P$ with partial user participation $(0 < \alpha < 1)$, it must be that, at this equilibrium, 
 \begin{align}
    \alpha = \frac{B_{\leq k^*}}{C\numusers^2 - \numusers(\numbuyers - k^*) } \cdot \frac{1}{P}. \label{exp:lin_alphaless1} 
 \end{align}
\end{lemma}
\begin{proof}
Since we know that this is an equilibrium at $k^*$, we have: 
\begin{align}\label{exp:lin_sumNk}
    \sum_{k=1}^{\numbuyers}N_k = \frac{B_{\leq k^*}}{P} + \alpha \numusers (\numbuyers - k^*). 
\end{align}
Now, recall that $\alpha = \min \left(1, \frac{C \alpha^2 \numusers^2}{\sum_{k=1}^{\numbuyers}N_k} \right)$. Since $0 < \alpha < 1$ by assumption, it must be that: 
\begin{align}\label{eq:lin_alphaless1}
 \alpha = \frac{C\alpha^2 \numusers^2}{\sum_{k=1}^{\numbuyers}N_k} \quad \text{which implies} \quad C\alpha \numusers^2 = \sum_{k=1}^{\numbuyers}N_k. 
\end{align}
Combining the results in \eqref{exp:lin_sumNk} and \eqref{eq:lin_alphaless1}, we conclude:
\begin{align*}
    C  \alpha \numusers^2 = \frac{B_{\leq k^*}}{P} + \alpha \numusers (\numbuyers - k^*).
\end{align*}
By re-arranging terms, we derive the final expression for $\alpha$:
\[ 
   \alpha = \frac{B_{\leq k^*}}{C\numusers^2 - \numusers(\numbuyers - k^*) } \cdot \frac{1}{P}.
\]
\end{proof}

\begin{lemma}\label{lem:lin_alphaless1_nec_cond}
 Fix $k^*$ such that $k^* > \numbuyers - C\numusers$. There exists an equilibrium at $k^*$ induced by price $P$ with $0 < \alpha < 1$ if and only if: 
\begin{align}\label{eq:lin_cond1_alphaless1}
    P > \frac{B_{\leq k^*}}{C \numusers^2 - \numusers (\numbuyers - k^*)};  
\end{align}
and,
\begin{align}\label{eq:lin_cond2_alphaless1} 
B_{k^*} < \frac{B_{\leq k^*}}{C \numusers - (\numbuyers - k^*)} \leq B_{k^*+1}. 
\end{align}
\end{lemma}
\begin{proof}
We need to establish that conditions \eqref{eq:lin_cond1_alphaless1} and \eqref{eq:lin_cond2_alphaless1} are both necessary and sufficient for the existence of equilibrium at $k^*$ with $0 < \alpha < 1$. 

\paragraph{Necessary Conditions} We will first show that if there exists an equilibrium at $k^*$ with partial user participation ($0 < \alpha < 1$), then conditions \eqref{eq:lin_cond1_alphaless1} and \eqref{eq:lin_cond2_alphaless1} must be satisfied. We start by showing that condition \eqref{eq:lin_cond1_alphaless1} must be satisfied by direct inspection of Eq. \eqref{exp:lin_alphaless1}. Recall from Lemma \ref{lem:lin_alphaless1} that if there exists an equilibrium at $k^*$ ($k^* > \numbuyers - C\numusers$) induced by $P$ with $0 < \alpha < 1$, then $\alpha$ is given by: 
\[
    \alpha = \frac{B_{\leq k^*}}{C\numusers^2 - \numusers(\numbuyers - k^*) } \cdot \frac{1}{P}.
\]
In order for $\alpha$ to be strictly less than $1$, we must ensure that: 
\[
    \frac{B_{\leq k^*}}{C\numusers^2 - \numusers(\numbuyers - k^*) } \cdot \frac{1}{P} < 1 \quad \text{which implies} \quad P > \frac{B_{\leq k^*}}{C \numusers^2 - \numusers (\numbuyers - k^*)}. 
\]
Now we will establish that condition \eqref{eq:lin_cond2_alphaless1} is also satisfied by using the monotonicity property of $N_k$'s at any equilibrium. We note that if there is an equilibrium at $k^*$, it must be that: 
\begin{align*}
    N_{k^*} < \alpha N; \quad \text{and} \quad N_{k^*+1} = \alpha N.
\end{align*} 
Recall that for all $k \in [\numbuyers]$, $N_k = \min \left(\alpha \numusers, \frac{B_k}{P} \right)$. Therefore, from the condition above, we get:
\[
    \frac{B_{k^*}}{P} < \alpha \numusers \leq \frac{B_{k^*+1}}{P}.
\]
Plugging in the value of $\alpha$ in Eq. \eqref{exp:lin_alphaless1}, we get: 
\[
  \frac{B_{k^*}}{P} < \frac{B_{\leq k^*}}{P \left(C \numusers - (\numbuyers - k^*) \right)} \leq \frac{B_{k^*+1}}{P}, 
\]
and since $P > 0$, we obtain the final condition:
\[
    \quad B_{k^*} < \frac{B_{\leq k^*}}{C \numusers - (\numbuyers - k^*)} \leq B_{k^*+1}.
\]
This concludes the proof of necessary conditions for existence of equilibrium at $k^*$ with $0 < \alpha < 1$. 

\paragraph{Sufficient Conditions} For the proof of sufficiency, we need to show that for any $k^* > \numbuyers - C\numusers$ such that $k^*$ satisfies condition \eqref{eq:lin_cond2_alphaless1} and any price $P$ which satisfies condition \eqref{eq:lin_cond1_alphaless1}, there exists an equilibrium at $k^*$ with $0 < \alpha < 1$ at price $P$. \\ \\
Choose $k^{'} > \numbuyers - C \numusers$ such that:
\[
    B_{k'} < \frac{B_{\leq k'}}{C \numusers - (\numbuyers - k')} \leq B_{k'+1}.
\]
Now, pick any price $P' > \frac{B_{\leq k'}}{C \numusers^2 - \numusers (\numbuyers - k')}$ (clearly $P' > 0$). Define $\alpha' = \frac{B_{\leq k'}}{P' \numusers\left(C\numusers - (\numbuyers - k')\right)}$. Because of our choice of $k'$, it is clear that $\alpha' > 0$. Also, since $P' > \frac{B_{\leq k'}}{C \numusers^2 - \numusers (\numbuyers - k')}$, this implies $\alpha' = \frac{B_{\leq k'}}{P' \numusers\left(C\numusers - (\numbuyers - k')\right)} < 1$. Thus, we have established that $\alpha' \in (0, 1)$. Now, we have to argue that $\alpha'$ is one of the possible equilibria at $k'$ induced by price $P'$. \\ \\
We have the following condition on $k'$: 
\[
     B_{k'} < \frac{B_{\leq k'}}{C \numusers - (\numbuyers - k')} \leq B_{k'+1}.
\]
Since $P' > 0$, dividing by $P'$, we arrive at the following inequality: 
\[
    \frac{B_{k'}}{P'} < \frac{B_{\leq k'}}{P' \left(C \numusers - (\numbuyers - k')\right)} \leq \frac{B_{k'+1}}{P'}.
\]
Using our definition of $\alpha'$, this is equivalent to:
\[
    \frac{B_{k'}}{P'} < \alpha' \numusers \leq \frac{B_{k'+1}}{P'}.
\]
Therefore, by the monotonicity of budget $B_k$'s (refer Assumption \ref{ass:budget_order}), we get: 
\[
     \frac{B_k}{P'} < \alpha' \numusers \quad \forall~k \leq k' \quad \text{and} \quad \frac{B_k}{P'} \geq \alpha' \numusers \quad \forall~k > k', 
\]
which implies:
\[
    N_k < \alpha' \numusers \quad \forall ~k \leq k' \quad \text{and } N_k = \alpha' \numusers \quad \forall~k > k'.
\]
Thus, $\alpha'$ is a candidate equilibrium at $k'$ induced by price $P'$, provided it satisfies:
\[
   \alpha' = \min \left(1, \frac{C \alpha'^2 \numusers^2}{\sum_{k=1}^{\numbuyers} N_k} \right) = \frac{C \alpha'^2 \numusers^2}{\sum_{k=1}^{\numbuyers} N_k} \quad (\text{since }\alpha' < 1)
\]
We can easily verify that the above equality is satisfied, by repeating the steps of the proof of Lemma \ref{lem:lin_alphaless1} in the reverse order. 
This concludes the proof of sufficiency. 
\end{proof}

\paragraph{Case 2: $\alpha = 1$ (Full Participation)} 
Similarly, we derive necessary and sufficient conditions for the existence of equilibria at buyer threshold $k^*$ with full user participation ($\alpha = 1$). 
\begin{lemma}\label{lem:lin_alpha1_nec_cond}
Fix $k^*$ such that $k^* > \numbuyers - C\numusers$. There exists an equilibrium at $k^*$ induced by price $P$, with $\alpha = 1$, if and only if: 
\begin{align}\label{ineq1:lin_alpha1}
 \frac{B_{k^*}}{\numusers} < P \leq \frac{B_{k^*+1}}{\numusers} 
\end{align}
and 
\begin{align}\label{ineq2:lin_alpha1}
 \quad \text{and} \quad P \geq \frac{B_{\leq k^*}}{C \numusers^2 - \numusers (\numbuyers - k^*)}.
\end{align}
\end{lemma}
\begin{proof}
We need to show that conditions \eqref{ineq1:lin_alpha1} and \eqref{ineq2:lin_alpha1} are necessary and sufficient conditions for the existence of equilibria at $k^*$ with $\alpha = 1$. 

\paragraph{Necessary Conditions}
Suppose, given $k^* > \numbuyers - C\numusers$, there exists an equilibrium at $k^*$ with $\alpha = 1$. Then we need to show that conditions \eqref{ineq1:lin_alpha1} and \eqref{ineq2:lin_alpha1} are satisfied. Since the said equilibrium is at $k^*$, by the definition of $k^*$ (Claim \ref{clm:N_monotone}), we must have: 
\[
    \frac{B_{k^*}}{P} < \alpha \numusers \leq \frac{B_{k^*+1}}{P}.
\]
Now, since $\alpha = 1$, the above condition reduces to: 
\[
   \frac{B_{k^*}}{P} < \numusers \leq \frac{B_{k^*+1}}{P}.
\]
Re-arranging terms, we have the following condition on $P$: 
\[
 \frac{B_{k^*}}{\numusers} < P \leq \frac{B_{k^*+1}}{\numusers}.
\]
and thus condition \eqref{ineq1:lin_alpha1} is satisfied. 
Now, recall that $\alpha = \min \left(1, \frac{C\alpha^2 \numusers^2}{\sum_{k=1}^{\numbuyers}N_k} \right)$. Therefore, $\alpha = 1$ implies:  
\[
    C\numusers^2 \geq \frac{B_{\leq k^*}}{P} + \numusers (\numbuyers - k^*); 
\]
and by re-arranging terms, we have:
\[
    P \geq \frac{B_{\leq k^*}}{C \numusers^2 - \numusers (\numbuyers - k^*)}.
\]
Thus condition \eqref{ineq2:lin_alpha1} is also satisfied. This concludes the proof. 

\paragraph{Sufficient Conditions} For sufficiency, let us pick some $k' > \numbuyers - C \numusers$ and assume that there exists some price $P'$ for which conditions \eqref{ineq1:lin_alpha1} and \eqref{ineq2:lin_alpha1} are satisfied. We need to show that $\alpha' = 1$ is a feasible equilibrium at $k'$, induced by price $P'$. It suffices to verify the following conditions: 
\begin{align}\label{verf1}
    \frac{B_{k'}}{P'\numusers} < 1 \leq \frac{B_{k'+1}}{P' \numusers}; 
\end{align}
and 
\begin{align}\label{verf2}
    C\numusers^2 \geq \sum_{k=1}^{\numbuyers}N_k.
\end{align}
Note that \eqref{verf1} is trivially satisfied because $P'$ satisfies \eqref{ineq1:lin_alpha1}. Now, by the budget monotonicity assumption, we see that $\frac{B_k}{P'\numusers} < 1$ for all $k < k'$ and $\frac{B_k}{P'\numusers} \geq 1$ for all $k > k'$. Therefore, 
\[
     \sum_{k=1}^{\numbuyers}N_k = \frac{B_{\leq k'}}{P'} + \numusers (\numbuyers - k'). 
\]
Since $P'$ also satisfies condition \eqref{ineq2:lin_alpha1}, we have $P' \geq \frac{B_{\leq k'}}{C \numusers^2 - \numusers (\numbuyers - k')}$ which implies:
\[
    \frac{B_{\leq k'}}{P'} + \numusers (\numbuyers - k') \leq C \numusers^2 
\]
Since we have already shown the LHS of the above inequality to be equal to $\sum_{k=1}^{\numbuyers}N_k$, we verify that condition \eqref{verf2} is also satisfied. Therefore, $\left(1, (N_k)_{k \in [\numbuyers]}\right)$ constitute an equilibrium at $k'$ induced by price $P'$. This concludes the proof of sufficiency. 
\end{proof}
We can have an alternate characterization of the necessary and sufficient conditions for existence of equilibria at $k^*$ with $\alpha = 1$, one that depends only on $k^*$. We present the following claim. 
\begin{claim}\label{clm:alt_condk_alpha1}
Fix $k^*$ such that $k^* > \numbuyers - C \numusers$. There exists an equilibrium at $k^*$ with $\alpha = 1$, if and only if $k^*$ satisfies: 
\begin{align}\label{eq:alt_condk_alpha1}
    \frac{B_{\leq k^*}}{C\numusers - (\numbuyers - k^*)} \leq B_{k^*+1}. 
\end{align}
\end{claim}
\begin{proof}
The proof is straightforward and follows directly from Lemma \ref{lem:lin_alpha1_nec_cond}. Recall that the following conditions on price $P$ are necessary and sufficient for the existence of equilibria at $k^*$ with $\alpha = 1$: 
\[
     \frac{B_{k^*}}{\numusers} < P \leq \frac{B_{k^*+1}}{\numusers};
\]
and 
\[
   P \geq \frac{B_{\leq k^*}}{C\numusers^2 - \numusers (\numbuyers - k^*)}. 
\]
Therefore, a feasible price $P$ that can induce an equilibrium at $k^*$ with $\alpha = 1$ exists if and only if the two price ranges intersect, or equivalently:
\[
     \frac{B_{\leq k^*}}{C\numusers^2 - \numusers (\numbuyers - k^*)} \leq \frac{B_{k^*+1}}{\numusers}, 
\]
which directly leads to the condition in \eqref{eq:alt_condk_alpha1}. This concludes the proof. 
\end{proof}

\subsection{Part 2: Our Simplified Theorem \ref{thm:lin_lowben}}\label{app:lin_simplethm}
In this section, we show how Theorem \ref{thm:lin_exist} reduces to the simpler/more interpretable form in Theorem \ref{thm:lin_lowben}. Before we proceed, we will introduce two technical lemmas which will be building blocks for the rest of the proofs. 

\begin{lemma}\label{lem:tech1}
Suppose $\hat k > \numbuyers - C\numusers$ satisfies: 
\begin{subequations}
\begin{align}\label{tech1:ineq1}
    B_{\hat k} < \frac{B_{\leq \hat k}}{C\numusers - (\numbuyers - \hat k)}; \quad \text{and},
\end{align}
\begin{align}\label{tech1:ineq2}
    \frac{B_{\leq \hat k}}{C\numusers - (\numbuyers - \hat k)} \leq B_{\hat k + 1}. 
\end{align}
\end{subequations}
Then, for all $k^* > \hat k$, $k^*$ violates \eqref{tech1:ineq1} and satisfies \eqref{tech1:ineq2}. 
\end{lemma}

\begin{lemma}\label{lem:tech2}
Given $\hat k > \numbuyers - C\numusers$ which satisfies both \eqref{tech1:ineq1} and \eqref{tech1:ineq2}, for all $k^*$ such that $\numbuyers - C\numusers < k^* < \hat k$, $k^*$ satisfies \eqref{tech1:ineq1} and violates \eqref{tech1:ineq2}. 
\end{lemma}
\begin{proof}
We will provide the proof for Lemma \ref{lem:tech1} here. The proof for Lemma \ref{lem:tech2} will be exactly identical and so, will be omitted to avoid repetition. It is given that:
\begin{align}\label{tech1_ineq:combined}
    B_{\hat k} < \frac{B_{\leq \hat k}}{C\numusers - (\numbuyers - \hat k)} \leq B_{\hat k + 1}. 
\end{align}
Using the first inequality from \eqref{tech1_ineq:combined}, we see that: 
\begin{align}\label{tech1:kplus1}
    &\frac{B_{\leq \hat k}}{C\numusers - (\numbuyers - \hat k)} \leq B_{\hat k + 1} \nonumber \\
    &\implies B_{\leq \hat k} \leq \left( C\numusers - (\numbuyers - \hat k) \right) B_{\hat k + 1}  \nonumber\\
    &\implies B_{\leq \hat k} + B_{\hat k + 1} \leq \left( C\numusers - (\numbuyers - \hat k) \right) B_{\hat k + 1} + B_{\hat k + 1} \nonumber \\
    &\implies B_{\leq \hat k + 1} \leq \left( C\numusers - (\numbuyers - \hat k - 1) \right) B_{\hat k + 1}  \nonumber\\
    &\implies \frac{B_{\leq \hat k + 1}}{\left( C\numusers - (\numbuyers - \hat k - 1) \right)} \leq  B_{\hat k + 1}.  
\end{align}
The last step follows because $C\numusers > \numbuyers - \hat k - 1$, so we can divide throughout by $\left( C\numusers - (\numbuyers - \hat k - 1) \right)$ while still retaining the direction of the inequality. Thus, we have established that $\hat k + 1$ violates \eqref{tech1:ineq1}. \\
Starting with the inequality we derived in \eqref{tech1:kplus1}, note that we immediately have: 
\[
     \frac{B_{\leq \hat k + 1}}{\left( C\numusers - (\numbuyers - \hat k - 1) \right)} \leq  B_{\hat k + 2}, 
\]
because $B_{\hat k + 1} \leq B_{\hat k + 2}$ by monotonicity of buyer budgets. Therefore, $\hat k + 1$ satisfies \eqref{tech1:ineq2}. By using an exactly identical argument as above, we can show: 
\[
  \frac{B_{\leq \hat k + 1}}{\left( C\numusers - (\numbuyers - \hat k - 1) \right)} \leq  B_{\hat k + 2}. 
\]
This implies that $\hat k + 2$ violates \eqref{tech1:ineq1} and satisfies \eqref{tech1:ineq2}. Finally, a straightforward induction argument extends the proof for all $k^* > \hat k$.   
\end{proof}
We can characterize the price threshold $\bar P$ and the structure of the equilibria in Theorem \ref{thm:lin_lowben} based on whether $\hat k$ exists or not. We will tackle the two cases separately.   

\subsubsection{Case 1 ($\hat k$ exists):}
In this segment, we will characterize the structure of the equilibria when $\hat k$ exists. The sequence of arguments is as follows:
\begin{enumerate}
\item First, we show in Claim \ref{clm:lin_alphaless1_unique} that $\hat k$ is the only value of $k^*$ which admits partial participation equilibria. This implies (by Lemma \ref{lem:lin_alphaless1_nec_cond}) that for all prices $P > \frac{B_{\leq \hat k}}{C\numusers^2 - \numusers (\numbuyers - \hat k)}$, there exists an equilibrium at $\hat k$ of the form $\alpha = \frac{B_{\leq \hat k}}{C\numusers^2 - \numusers (\numbuyers - \hat k)}\cdot \frac{1}{P}$. 
\item We then see that $\hat k$ also admits equilibria with $\alpha = 1$. This is immediate, since $\hat k$ satisfies the condition in Claim \ref{clm:alt_condk_alpha1}. 
\item Next, we show in Claim \ref{clm:lin_binding_constraint} that the constraint $P \geq \frac{B_{\leq k^*}}{C\numusers^2 - \numusers (\numbuyers - k^*)}$ in Lemma \ref{lem:lin_alpha1_nec_cond} is active if and only if $k^* = \hat k$. This implies that the price interval which induces full participation equilibria at $\hat k$ is given by: 
\[
        \frac{B_{\leq \hat k}}{C\numusers^2 - \numusers (\numbuyers - \hat k)} \leq P \leq \frac{B_{\hat k+1}}{\numusers}.
\]
\item By technical lemma \ref{lem:tech1}, we know that all $k^* > \hat k$ admit equilibria with full participation. Since the constraint $P \geq \frac{B_{\leq k^*}}{C\numusers^2 - \numusers (\numbuyers - k^*)}$ is active only when $k^* = \hat k$, we conclude that the price range which induces equilibria at $k^* > \hat k$ with $\alpha = 1$ is given by: 
\[
         \frac{B_{k^*}}{\numusers} < P \leq \frac{B_{k^*+1}}{\numusers}.
\]
\end{enumerate}
In summary, we conclude that $\bar P = \frac{B_{\leq \hat k}}{C\numusers^2 - \numusers (\numbuyers - \hat k)}$. For all $P \geq \bar P$, there exists a full participation equilibria. This follows from subparts 3 and 4 above. Additionally, for all $P > \bar P$, there exists a partial participation equilibria of the form $\alpha = \frac{A}{P}$ where $A = \frac{B_{\leq \hat k}}{C\numusers^2 - \numusers (\numbuyers - \hat k)}$ (by subpart 1). Computation of $A$ reduces to finding $\hat k$ which can be done efficiently. 

\paragraph{Subpart 1} We now prove subpart 1 using the following claim.
\begin{claim}\label{clm:lin_alphaless1_unique}
Suppose there exists $\hat k$ which satisfies \eqref{tech1:ineq1} and \eqref{tech1:ineq2}. Then $\hat k$ is the unique value of $k^*$ which admits equilibria with partial user participation $0 < \alpha < 1$. 
\end{claim}
\begin{proof}
Using technical lemmas \ref{lem:tech1} and \ref{lem:tech2}, we conclude that $\hat k$ is unique. Now, we will argue that $\hat k$ satisfies \eqref{tech1:ineq1} and \eqref{tech1:ineq2} if and only if it admits equilibria with $0 < \alpha < 1$. Observe that this follows directly from the result in Lemma \ref{lem:lin_alphaless1_nec_cond}. This concludes the proof. 
\end{proof}

\paragraph{Subpart 3} For the proof of subpart 3, we just need to prove Claim \ref{clm:lin_binding_constraint}.
\begin{claim}\label{clm:lin_binding_constraint}
Suppose $k^*$ admits full participation equilibria. Then the constraint $P \geq \frac{B_{\leq k^*}}{C\numusers^2 - \numusers (\numbuyers - k^*)}$ in Lemma \ref{lem:lin_alpha1_nec_cond} is binding if and only if $k^* = \hat k$.
\end{claim}
\begin{proof}
Recall from Lemma \ref{lem:lin_alpha1_nec_cond} that there are the following two conditions on price $P$ for the existence of equilibria at $k^*$ with $\alpha = 1$: 
\[
             \frac{B_{k^*}}{\numusers} < P \leq \frac{B_{k^*+1}}{\numusers};
\]
and 
\[
           P \geq \frac{B_{\leq k^*}}{C\numusers^2 - \numusers (\numbuyers - k^*)}. 
\]
Therefore, the constraint $P \geq \frac{B_{\leq k^*}}{C\numusers^2 - \numusers (\numbuyers - k^*)}$ is binding if and only if: 
\[
        \frac{B_{\leq k^*}}{C\numusers^2 - \numusers (\numbuyers - k^*)} > \frac{B_{k^*}}{\numusers} \quad \text{which is equivalent to}
\]
\[
        \frac{B_{\leq k^*}}{C\numusers - (\numbuyers - k^*)} > B_{k^*}.
\]
Note that, we already have $\frac{B_{\leq k^*}}{C\numusers - (\numbuyers - k^*)} \leq B_{k^*+1}$ since $k^*$ admits equilibria with $\alpha = 1$ (by Claim \ref{clm:alt_condk_alpha1}). Therefore, combining both conditions we get: 
\[
      B_{k^*} < \frac{B_{\leq k^*}}{C\numusers - (\numbuyers - k^*)} \leq B_{k^*+1}.
\]
Therefore, $k^*$ must be equal to $\hat k$ because $\hat k$ is the unique value of $k^*$ which satisfies this condition. This concludes the proof. 
\end{proof}

\subsubsection{Case 2 ($\hat k$ does not exist):}
Suppose, $\hat k$ does not exist. This implies that \textbf{there does not exist $k^* > \numbuyers - C\numusers$} which satisfies: 
\[
     B_{k^*} < \frac{B_{\leq k^*}}{C\numusers - (\numbuyers - k^*)} \leq B_{k^*+1}.    
\]
We can infer the following about the equilibria structure in this case: 
\begin{enumerate}
    \item There are no partial participation equilibria. 
    \item Let $\widetilde k$ be the smallest value of $k^* > \numbuyers - C\numusers$ which satisfies: 
    \[
        \frac{B_{\leq k^*}}{C\numusers - (\numbuyers - k^*)} \leq B_{k^*+1}.  
    \]
    Note that $\widetilde k$ is guaranteed to exist because $k^* = \numbuyers$ always satisfies the above condition. By technical lemma \ref{lem:tech1}, all $k^* > \widetilde k$ also satisfy this condition. Therefore, by Claim \ref{clm:alt_condk_alpha1}, for all $k^* \geq \widetilde k$, there exists equilibria at $k^*$ with $\alpha = 1$.    
    \item For $k^* \geq \widetilde k$, the price range which induces equilibria at $k^*$ with $\alpha = 1$, is given by: 
    \[
          \frac{B_{k^*}}{\numusers} < P \leq \frac{B_{k^*+1}}{\numusers}. 
    \]
    This follows from the fact that $\hat k$ does not exist, so the constraint $P \geq \frac{B_{\leq k^*}}{C\numusers^2 - \numusers(\numbuyers - k^*)}$ is never active (refer to Claim \ref{clm:lin_binding_constraint}). 
\end{enumerate}
Based on the above observations, we conclude that $\bar P = \frac{B_{\widetilde k}}{\numusers}$. For all $P \geq \bar P$, there exists a unique equilibria with $\alpha = 1$. There are no  partial participation equilibria.

\end{document}